\documentclass[sigplan,screen]{acmart}
\setcopyright{rightsretained}
\acmPrice{}
\acmDOI{10.1145/3453483.3454078}
\acmYear{2021}
\copyrightyear{2021}
\acmSubmissionID{pldi21main-p337-p}
\acmISBN{978-1-4503-8391-2/21/06}
\acmConference[PLDI '21]{Proceedings of the 42nd ACM SIGPLAN International Conference on Programming Language Design and Implementation}{June 20--25, 2021}{Virtual, Canada}
\acmBooktitle{Proceedings of the 42nd ACM SIGPLAN International Conference on Programming Language Design and Implementation (PLDI '21), June 20--25, 2021, Virtual, Canada}


\usepackage{adjustbox}
\usepackage{algorithm}
\usepackage{algpseudocode}
\usepackage{amsfonts}
\usepackage{amsmath}
\usepackage{amsthm}
\usepackage{booktabs}
\usepackage{comment}
\usepackage{colortbl}
\usepackage[inline]{enumitem}
\usepackage[mathscr]{eucal}
\usepackage{fancyhdr}
\usepackage[T1]{fontenc}
\usepackage{framed}
\usepackage{graphicx}
\usepackage{listings}
\usepackage{lipsum}
\usepackage{mathtools}
\usepackage{multirow}
\usepackage{natbib}
\usepackage{newfloat}
\usepackage{pifont}
\usepackage{siunitx}
\usepackage{subcaption}
\usepackage{stmaryrd}
\usepackage{stackengine}
\usepackage{thm-restate}
\usepackage{tabularx}
\usepackage{textcomp}
\usepackage{tikz}
\usepackage{tikz-qtree}
\usepackage{wrapfig}

\usetikzlibrary{positioning}

\lstdefinestyle{sppl}{
  basicstyle=\ttfamily\footnotesize,
  language=Python,
  columns=fullflexible,
  keepspaces=true,
  upquote=true,
  morekeywords=[1]{condition, prob, switch, cases, range, array},
  keywordstyle=[1]\bfseries\textcolor{NavyBlue},
  morekeywords=[2]{choice, bernoulli, atomic, uniform},
  keywordstyle=[2]\textcolor{Black},
  stringstyle=\ttfamily\color{Black},
  commentstyle=\color{gray}\ttfamily,
  linewidth=\linewidth,
  xleftmargin=3pt,
  xrightmargin=3pt,
  framesep=3pt,
}

\DeclarePairedDelimiter{\abs}{\lvert}{\rvert}

\newcommand{\sppl}{\textrm{\textsc{Sppl}}}
\newcommand{\psii}{\textrm{PSI}}

\newcommand{\blog}{\textrm{BLOG}}
\newcommand{\fairsquare}{\textrm{FairSquare}}
\newcommand{\verifair}{\textrm{VeriFair}}

\newcommand{\spflow}{\textrm{SPFlow}}

\newcommand{\python}{\textrm{Python}}

\newcommand{\asdef}{\eqqcolon}
\newcommand{\defas}{\coloneqq}
\newcommand{\gor}{\mathrel{\vert}}
\newcommand{\set}[1]{\{#1\}}
\newcommand{\ceil}[1]{\lceil#1\rceil}
\newcommand{\floor}[1]{\lfloor#1\rfloor}

\newcommand{\Denot}[2][]{#1\left\llbracket#2\right\rrbracket}

\newcommand{\dom}[1]{\mathsf{#1}}
\newcommand{\SPE}{\dom{SPE}}

\newcommand{\token}[1]{\mathtt{#1}}
\newcommand{\kw}[1]{\textcolor{NavyBlue}{\textbf{\texttt{#1}}}}

\renewcommand{\translate}{\to_\SPE}
\newcommand{\translateR}{\to_\sppl{}}

\newcommand{\translateRStar}{\to^*_\sppl{}}
\newcommand{\translateStar}{\to^*_\SPE}

\newcommand{\valfunc}[1]{\mathbb{#1}}
\newcommand{\Denotv}[2]{\Denot[\valfunc{#1}]{#2}}

\newcommand{\sexpr}[1]{\textup{\texttt{(}}{#1}\textup{\texttt{)}}}
\newcommand{\settt}[1]{\textup{\texttt{\{}}{#1}\textup{\texttt{\}}}}
\newcommand{\bracktt}[1]{\textup{\texttt{[}}{#1}\textup{\texttt{]}}}

\newcommand{\scall}[2]{\token{#1}\sexpr{#2}}
\newcommand{\sintvl}[2]{\sexpr{\sexpr{#1}\,\sexpr{#2}}}

\newcommand{\domfunc}[1]{\mathit{#1}}

\newcommand{\dist}[1]{\mathrm{#1}}

\newcommand{\belse}{\mathbf{else}}
\newcommand{\bbe}{\mathbf{be}}

\newcommand{\bif}{\mathbf{if}}
\newcommand{\bin}{\mathbf{in}}
\newcommand{\blet}{\mathbf{let}}
\newcommand{\bmatch}{\mathbf{match}}
\newcommand{\bthen}{\mathbf{then}}
\newcommand{\bundef}{\mathbf{undefined}}

\newcommand{\bindicator}[1]{\mathbf{1}\left[#1\right]}

\newcommand{\ttrue}{\texttt{\#}\token{t}}
\newcommand{\tfalse}{\texttt{\#}\token{f}}
\newcommand{\tunit}{\texttt{\#}\token{u}}

\newcommand{\dquote}[1]{\texttt{"#1"}}
\newcommand{\squote}[1]{\texttt{\textquotesingle#1\textquotesingle}}

\newcommand{\mli}[1]{\mathit{#1}}

\newcommand{\inj}[3][]{\downarrow\substack{#2\\#3}\,#1}
\makeatletter
\newcommand{\subalign}[1]{%
  \vcenter{%
    \Let@ \restore@math@cr \default@tag
    \baselineskip\fontdimen10 \scriptfont\tw@
    \advance\baselineskip\fontdimen12 \scriptfont\tw@
    \lineskip\thr@@\fontdimen8 \scriptfont\thr@@
    \lineskiplimit\lineskip
    \ialign{\hfil$\m@th\scriptstyle##$&$\m@th\scriptstyle{}##$\hfil\crcr
      #1\crcr
    }%
  }%
}
\makeatother

\newcommand{\fillBlue}{white!95!blue}

\newcommand{\fillGray}{white!90!black}
\newcommand{\fillGreen}{white!95!green}
\newcommand{\fillYell}{white!95!yellow}

\newcommand{\fillCondt}{\fillBlue}
\newcommand{\fillPrior}{\fillYell}
\newcommand{\fillQuery}{\fillGreen}

\DeclareFloatingEnvironment[
  fileext=lop,
  name=Listing
]{listing}
\DeclareCaptionSubType{listing}
\makeatletter
\newcommand{\mathleft}[1]{\@fleqntrue\@mathmargin#1}
\newcommand{\mathcenter}{\@fleqnfalse}
\makeatother

\makeatletter
  \newcommand{\addToLabel}[1]{%
    \protected@edef\@currentlabel{\@currentlabel#1}%
  }
\makeatother

\newcounter{rule}

\newcommand{\staterule}[4][]{%
  \refstepcounter{rule}%
  \addToLabel{(\textsc{#2})}\label{#2}%
  $\begin{array}[b]{@{}l@{}}%
    \mbox{(\textsc{#2})#1}\\%
    \begin{array}{@{}c@{}}
      #3\\
      \hline
      \raisebox{0ex}[2.5ex]{\strut}#4%
    \end{array}
  \end{array}$}

\AtEndPreamble{%
\theoremstyle{acmdefinition}
\newtheorem{remark}[theorem]{Remark}}

\begin{document}

\title{SPPL: Probabilistic Programming with Fast Exact Symbolic Inference}

\author{Feras A.~Saad}
\affiliation{
  \institution{Massachusetts Institute of Technology}
  \city{Cambridge}
  \state{MA}
  \country{USA}
}
\author{Martin C.~Rinard}
\affiliation{
  \institution{Massachusetts Institute of Technology}
  \city{Cambridge}
  \state{MA}
  \country{USA}
}
\author{Vikash K.~Mansinghka}
\affiliation{
  \institution{Massachusetts Institute of Technology}
  \streetaddress{Address}
  \city{Cambridge}
  \state{MA}
  \country{USA}
}
\authorsaddresses{}
%!TEX root = ./paper.tex

\begin{abstract}
We present the Sum-Product Probabilistic Language (\sppl{}), a new
probabilistic programming language that automatically delivers exact
solutions to a broad range of probabilistic inference queries.
\sppl{} translates probabilistic programs into {\em sum-product expressions},
a new symbolic representation and associated semantic domain that
extends standard sum-product networks to support mixed-type
distributions, numeric transformations, logical formulas, and
pointwise and set-valued constraints.
We formalize \sppl{} via a novel translation strategy from
probabilistic programs to sum-product expressions
and give sound exact algorithms for conditioning on and computing
probabilities of events.
\sppl{} imposes a collection of restrictions on probabilistic programs
to ensure they can be translated into sum-product expressions, which
allow the system to leverage new techniques for improving the
scalability of translation and inference by automatically exploiting
probabilistic structure.
We implement a prototype of \sppl{} with a modular architecture and
evaluate it on benchmarks the system targets, showing that it obtains
up to 3500x speedups over state-of-the-art symbolic systems on tasks
such as verifying the fairness of decision tree classifiers, smoothing hidden
Markov models, conditioning transformed random variables, and computing rare
event probabilities.
\end{abstract}

\maketitle

%!TEX root = ./paper.tex

\section{Introduction}
\label{sec:introduction}

Reasoning under uncertainty is a well-established theme across diverse
fields including
  robotics,
  cognitive science,
  natural language processing,
  algorithmic fairness,
  amongst many others~\citep{thrun2005,chater2006,jelinek1997,dwork2012}.
A common approach for modeling uncertainty is to use probabilistic
programming languages (PPLs~\citep{gordon2014}) to both represent complex probability
distributions and perform probabilistic inference
within the language.
There is growing recognition of the utility of PPLs for solving
challenging tasks that involve probabilistic reasoning in various
application domains~\citep{kulkarni2015,ghahramani2015,bolton2019,krapu2019}.

\input{figures/architecture}

Probabilistic inference is central to reasoning about uncertainty and
is a central concern for both PPL implementors and users.
Several PPLs leverage approximate inference
techniques~\citep{thomas1994,goodman2008,wingate2013},
which have been used effectively in a variety of
settings~\citep{sankaranarayanan2013,carpenter2017,towner2019}.
Drawbacks of approximate inference, however, include
  a lack of accuracy and/or soundness guarantees~\citep{luby1993,lew2020};
  difficulties with programs that combine continuous, discrete, or mixed-type distributions~\citep{carpenter2017,wu2018};
  challenges assessing the quality of iterative solvers~\citep{brooks1998};
  and the substantial expertise needed to write custom inference
  programs that deliver acceptable performance~\citep{mansinghka2018,towner2019}.
To address the shortcomings of approximate inference,
several PPLs instead use exact symbolic techniques~\citep{bhat2013,narayanan2016,gehr2016,shan2016,zhang2019}.
These languages can typically express a large class of models, using
general computer algebra to solve queries.
However, the generality of the symbolic computations causes them to
sometimes fail, even on problems with tractable solutions.

\noindentparagraph{Our Work}
We introduce the Sum-Product Probabilistic Language
(\sppl{}), a system that occupies a new point in the expressiveness
vs.\ performance trade-off space for exact symbolic inference.
A key idea in \sppl{} is to incorporate certain modeling restrictions
that avoid the need for general computer algebra, instead using a new,
specialized class of ``sum-product'' symbolic expressions to exactly
represent probability distributions specified by \sppl{} programs.
These new symbolic expressions extend and generalize sum-product
networks~\citep{poon2011}, which are computational graphs that have
received widespread attention for their clear probabilistic semantics
and tractable properties for exact inference---see~\citep{vergari2020}
for a comprehensive and curated literature review.
These sum-product expressions are used to automatically obtain exact
solutions to probabilistic inference queries about \sppl{} programs,
which are fast and scalable in tractable regimes.

\noindentparagraph{System Overview}
Fig.~\ref{fig:system-diagram} shows an overview of our approach.
Given a probabilistic program written in \sppl{}
(Lst.~\ref{lst:sppl-syntax}) a translator
(Lst.~\ref{lst:sppl-translation}) produces a sum-product expression
that represents the prior distribution over all program variables.
Given this expression and a query specified by the user,
the \sppl{} inference engine returns an exact answer, where:
\begin{enumerate}
\item $\kw{simulate}(\dom{Vars})$ returns random samples
    of program variables from their joint probability distribution;
\item $\kw{prob}(\dom{Event})$ returns the probability of an event,
  which is a predicate on program variables;
\item $\kw{condition}(\dom{Event})$ returns a new
  sum-product expression for the posterior distribution
  over program variables, given that the specified event is true.
\end{enumerate}
A key aspect of the system design in Fig.~\ref{fig:system-diagram} is
modularity: modeling, conditioning, and querying are factored into
distinct stages that reflect the essential components of a Bayesian
workflow.
Moreover, the dashed back-edge in Fig.~\ref{fig:system-diagram}
indicates that the new sum-product expression returned by
$\kw{condition}$ can be reused to interactively invoke additional
queries on the posterior distribution.
This closure property enables substantial runtime gains across
multiple datasets and queries.

\noindentparagraph{Trade-offs}
\sppl{} imposes restrictions on
probabilistic programs that specifically rule out the following
constructs:
\begin{enumerate*}[label=(\roman*)]
\item unbounded loops;
\item multivariate numeric transformations;  and
\item arbitrary prior distributions on continuous parameters.
\end{enumerate*}
As a result, \sppl{} is not designed to express model classes such as
regression with a prior on real coefficients; neural networks;
support-vector machines; spatial Poisson processes; urn processes; and
hidden Markov models with unknown transition matrices.
The aforesaid model classes cannot be represented as sum-product
expressions, and most of them do not have tractable algorithms for
exact inference.

We impose these restrictions to ensure that valid \sppl{} programs can
always be translated into finite sum-product expressions, as opposed
to general symbolic algebra expressions.
The resulting sum-product expressions delivered by \sppl{} have a
number of characteristics that make them a particularly useful
translation target for probabilistic programs:

\begin{itemize}[wide=0pt]
\item {\bf Completeness and Decomposibility:}
  By satisfying important completeness and decomposability conditions
  from the literature~\citep[Defs.~4,5]{poon2011}, sum-product
  expressions are guaranteed to represent normalized probability
  distributions.

\item {\bf Efficient Factorization:}
  By specifying multivariate probability distributions compositionally
  in terms of sums and products of simpler distributions, sum-product
  expressions can be simplified by algebraic ``factorization''
  (Fig.~\ref{fig:hmm-spe-factorized},
  Fig.~\ref{fig:memoize-factorization}).

\item {\bf Efficient Deduplication:}
  When an \sppl{} program specifies a generative model
  with conditional independence structure, the translated sum-product
  expression typically contains identical subexpressions that can be
  ``deduplicated'' into a single logical node in memory
  (Fig.~\ref{fig:hmm-spe-factorized}, Fig.~\ref{fig:memoize-deduplication}).

\item {\bf Efficient Caching:}
  Inference algorithms for sum-product expressions proceed from root
  to leaves to root, allowing intermediate results to be cached and
  reused at deduplicated internal subexpressions in a depth-first
  graph traversal.

\item {\bf Closure Under Conditioning:}
  Sum-product expressions are closed under probabilistic conditioning
  (Thm.~\ref{thm:closure}), which allows them to be reused across
  multiple datasets and inference queries about the same
  probabilistic program.

\item {\bf Linear-Time Exact Inferences:}
  For a well-defined class of common queries, inference scales
  linearly in the expression size (Thm.~\ref{thm:condition-linear});
  when \sppl{} delivers a ``small'' expression after
  factorization and deduplication, inference is also fast.
\end{itemize}

It is well-known that a very large class of tractable models can be
cast as sum-product networks~\citep[Thm.~2]{poon2011}.
\sppl{} automatically constructs these representations from generative
probabilistic programs that use standard constructs such as arrays,
if/else branches, for-loops, and numeric and logical operators.
To enable this translation, \sppl{} introduces new sum-product expressions and
inference algorithms that extend standard sum-product networks by supporting
(many-to-one) univariate transformations,
mixed-type base measures, and
pointwise and set-valued constraints.
These constructs make \sppl{} expressive enough to solve
prominent inference tasks in the PPL
literature~\citep{albarghouthi2017,nori2014,wu2018,laurel2020} for
which standard sum-product networks have not been previously used.
Example model classes include most finite discrete models, latent
variable models with discrete hidden states and arbitrary observed
states, and decision trees over discrete and continuous variables.
Taken together, these characteristics make \sppl{} particularly
effective for fast and scalable inference on tractable problems, with
low variance runtime and complete, usable answers to users.
Our experimental evaluation (Sec.~\ref{sec:evaluations}) indicates
that \sppl{} delivers these benefits on the problems it is designed to
solve, whereas more general and expressive techniques in previous
solvers~\citep{gehr2016,albarghouthi2017,bastani2019} typically
exhibit orders of magnitude worse performance on these problems,
runtime has higher variance, and/or results may be unusable, i.e.,
with unsimplified symbolic integrals.

\noindentparagraph{Key contributions}
We identify the following contributions:

\begin{itemize}[wide=0pt]
\item \textbf{New semantic domain for sum-product expressions}
(Sec.~\ref{sec:core}) that extends sum-product
networks~\citep{poon2011} by including mixed-type distributions,
numeric transforms, logical formulas, and events with pointwise and
set-valued constraints.

\item \textbf{Provably sound exact symbolic inference algorithms}
(Sec.~\ref{sec:condition})
based on a proof that sum-product expressions are closed under
conditioning on any event that can be specified in the domain.
We use these algorithms to build an efficient and multi-stage inference
architecture that separates model translation, conditioning, and
querying into distinct stages, enabling interactive workflows
and computation reuse.

\item \textbf{The Sum-Product Probabilistic Language} (Sec.~\ref{sec:translation}),
a PPL built on a novel translation semantics from generative code to
sum-product expressions, which are used to deliver exact inferences to
queries.
We present optimization techniques to improve scalability of
translation and inference by exploiting conditional independences and
repeated structure.

\item \textbf{Empirical measurements of efficacy} (Sec.~\ref{sec:evaluations})
on inference tasks from the literature that \sppl{} targets, which
show that it delivers substantial improvements over existing
baselines, including up to 3500x speedup over state-of-the-art
fairness verifiers~\citep{albarghouthi2017,bastani2019} and symbolic
integration~\citep{gehr2016}, as well as many orders of magnitude
speedup over sampling-based inference~\citep{milch2005} for computing
rare event probabilities.
\end{itemize}

%!TEX root = ./paper.tex

\section{Overview}
\label{sec:overview}

%!TEX root = ../paper.tex

\begin{figure*}[!t]

\centering

% \hrule
\begin{subfigure}[b]{.45\textwidth}
\begin{lstlisting}[style=sppl,numbers=left,frame=single,numbersep=5pt]
Nationality ~ choice({'India': 0.5, 'USA': 0.5})
if (Nationality == 'India'):
    Perfect ~ bernoulli(p=0.10)
    if Perfect:         GPA ~ atom(10)
    else:               GPA ~ uniform(0, 10)
else: # Nationality is 'USA'
    Perfect ~ bernoulli(p=0.15)
    if Perfect:         GPA ~ atom(4)
    else:               GPA ~ uniform(0, 4)
\end{lstlisting}
\captionsetup{aboveskip=0pt, belowskip=0pt}
\caption{Probabilistic Program}
\label{fig:indian-gpa-program}
\end{subfigure}\qquad %
\begin{subfigure}[b]{.45\textwidth}
\centering
\begin{lstlisting}[style=sppl,frame=single]
prob (Nationality == 'USA');
prob (Perfect == 1);
prob (GPA <= x/10) # for x = 0, ..., 120
\end{lstlisting}
\captionsetup{aboveskip=0pt, belowskip=4pt}
\caption{Example Queries on Marginal Probabilities}
\label{fig:indian-gpa-query-marginal}

\begin{lstlisting}[style=sppl,frame=single]
prob ((Perfect == 1)
  or (Nationality == 'India') and (GPA > 3))
\end{lstlisting}
\captionsetup{aboveskip=0pt, belowskip=0pt}
\caption{Example Query on Joint Probabilities}
\label{fig:indian-gpa-query-joint}
\end{subfigure}%

\tikzset{leaf/.style={inner sep=1pt,label={[label distance=-.1cm]below:{#1}}}}
\tikzset{branch/.style={circle,draw,inner sep = 1pt}}

\begin{subfigure}{\textwidth}
\begin{subfigure}[b]{.5\textwidth}
\begin{adjustbox}{max width=\linewidth}
\begin{tikzpicture}
\Tree
  [.\node[branch]{$+$};
    % Subtree.
    \edge node[auto=right]{.5};
    [.\node[branch]{$\times$};
      [.\node[branch]{$+$};
        % Subtree
        \edge node[auto=right]{.1};
        [.\node[branch]{$\times$};
          \node[leaf=Perfect]{$\delta_{\rm True}$};
          \node[leaf=GPA]{$\delta_{10}$}; ]
        % Subtree
        \edge node[auto=left]{.9};
        [.\node[branch]{$\times$};
          \node[leaf=Perfect]{$\delta_{\rm False}$};
          \node[leaf=GPA]{$U(0,10)$};] ]
      \node[leaf=Nationality]{$\delta_{\rm India}$}; ]
    % Subtree.
    \edge node[auto=left]{.5};
    [.\node[branch]{$\times$};
      [.\node[branch]{$+$};
        % Subtree
        \edge node[auto=right]{.15};
        [.\node[branch]{$\times$};
          \node[leaf=Perfect]{$\delta_{\rm True}$};
          \node[leaf=GPA]{$\delta_{4}$};]
        % Subtree
        \edge node[auto=left]{.85};
        [.\node[branch]{$\times$};
          \node[leaf=Perfect]{$\delta_{\rm False}$};
          \node[leaf=GPA]{$U(0,4)$};] ]
      \node[leaf=Nationality]{$\delta_{\rm USA}$};] ]
\end{tikzpicture}
\end{adjustbox}
% \vspace{-.1cm}
\captionsetup{aboveskip=10pt, belowskip=0pt}
\caption{Prior Sum-Product Expression}
\label{fig:indian-gpa-spe-prior}
\end{subfigure}\hfill
\begin{subfigure}[b]{.5\textwidth}
\includegraphics[width=\textwidth]{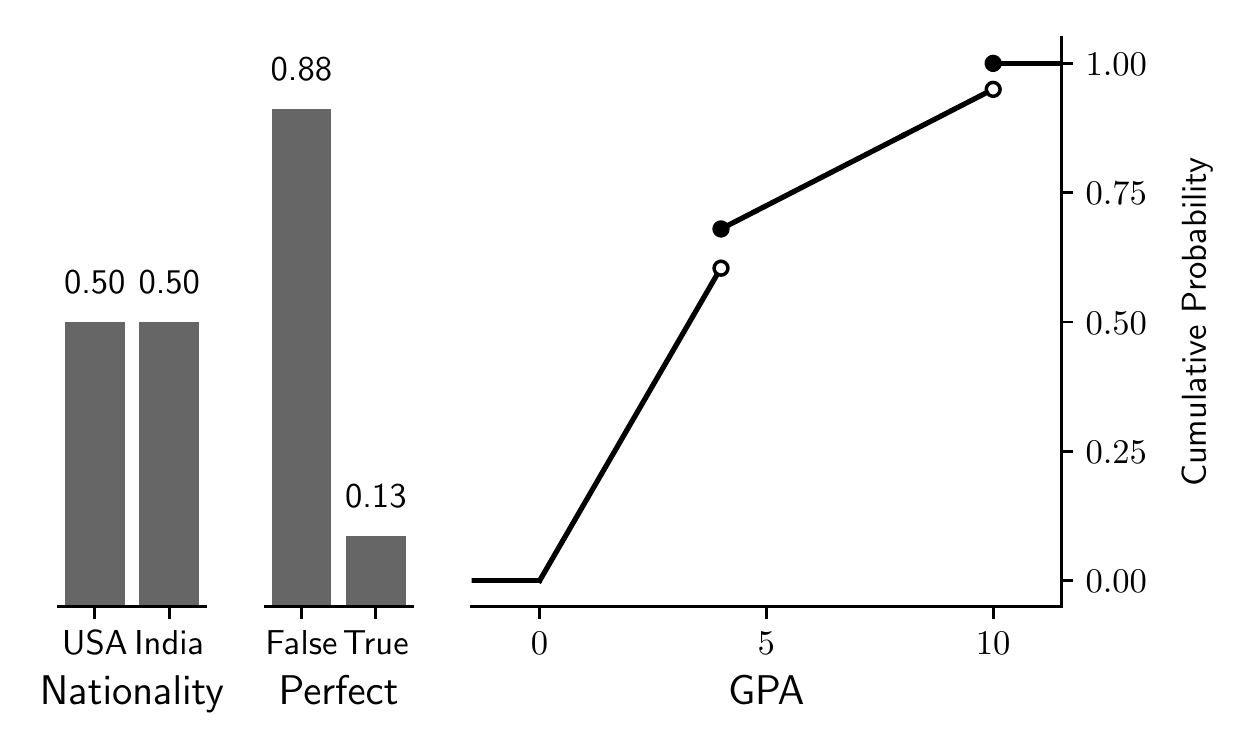}
\captionsetup{aboveskip=0pt, belowskip=0pt}
\caption{Prior Marginal Distributions}
\label{fig:indian-gpa-probs-prior}
\end{subfigure}
\end{subfigure}
\bigskip

\begin{subfigure}{.7\textwidth}
\centering
\begin{lstlisting}[style=sppl,frame=single]
condition ((Nationality == 'USA') and (GPA > 3)) or (8 < GPA < 10)
\end{lstlisting}
\captionsetup{aboveskip=0pt, belowskip=0pt}
\caption{Conditioning the Program}
\label{fig:indian-gpa-condition}
\end{subfigure}

\begin{subfigure}{\textwidth}
\begin{subfigure}[b]{.5\textwidth}
\begin{adjustbox}{max width=\linewidth}
\begin{tikzpicture}
\Tree
  [.\node[branch]{$+$};
    % Subtree.
    \edge node[auto=right]{.33};
    [.\node[branch]{$\times$};
      % Child
      \node[leaf=Nationality]{$\delta_{\rm India}$};
      \node[leaf=Perfect]{$\delta_{\rm False}$};
      \node[leaf=GPA]{$U(8,10)$}; ]
    % Subtree.
    \edge node[auto=left]{.67};
    [.\node[branch]{$\times$};
      [.\node[branch]{$+$};
        % Subtree
        \edge node[auto=right]{.41};
        [.\node[branch]{$\times$};
          \node[leaf=Perfect]{$\delta_{\rm True}$};
          \node[leaf=GPA]{$\delta_{4}$};]
        % Subtree
        \edge node[auto=left]{.59};
        [.\node[branch]{$\times$};
          \node[leaf=Perfect]{$\delta_{\rm False}$};
          \node[leaf=GPA]{$U(3,4)$};] ]
      \node[leaf=Nationality]{$\delta_{\rm USA}$}; ] ]
\end{tikzpicture}
\end{adjustbox}
% \vspace{-.3cm}
\captionsetup{aboveskip=10pt, belowskip=0pt}
\caption{Posterior Sum-Product Expression}
\label{fig:indian-gpa-spe-posterior}
\end{subfigure}\hfill
\begin{subfigure}[b]{.5\textwidth}
\includegraphics[width=\textwidth]{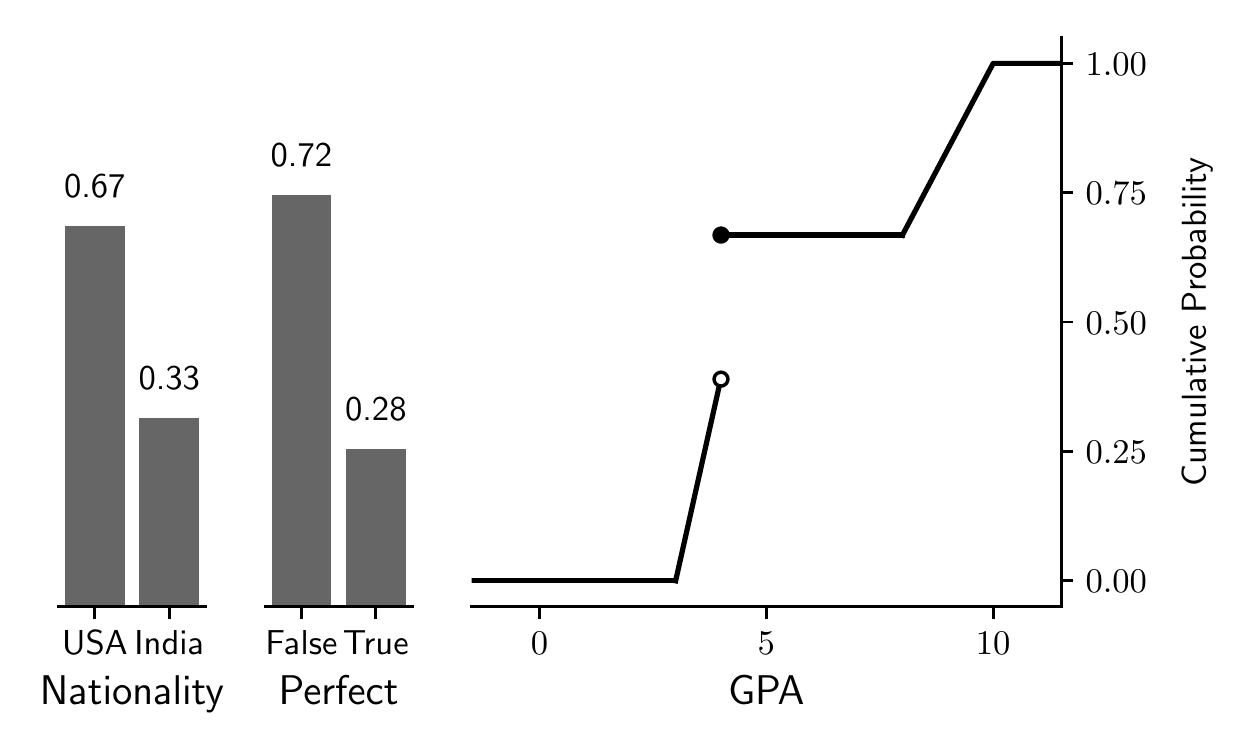}
\captionsetup{aboveskip=0pt, belowskip=0pt}
\caption{Posterior Marginal Distributions}
\label{fig:indian-gpa-probs-posterior}
\end{subfigure}
\end{subfigure}
% \captionsetup{belowskip=-25pt}
% \hrule
% \captionsetup{aboveskip=2pt}
\caption{Analyzing the Indian GPA problem in \sppl{}.}
\label{fig:indian-gpa}
\end{figure*}

We next describe two examples that illustrate the programming style in
\sppl{} and queries supported by the system.

\subsection{Indian GPA Problem}
\label{subsec:example-indian-gpa}

The Indian GPA problem is a canonical example that has been widely
considered in the probabilistic programming literature~\citep{nitti2016,srivastava2017,wu2018,riguzzi2018,narayanan2020}
for its use of a ``mixed-type'' random variable that takes both
continuous and discrete values, depending on the random branch taken
by the program.

\noindentparagraph{Specifying the Prior}
Fig.~\ref{fig:indian-gpa-program} shows the generative process for three
variables ($\texttt{Nationality}$, $\texttt{Perfect}$ and $\texttt{GPA}$)
of a student.
The student's nationality is either India or USA with equal
probability (line 1).
Students from India (line 2) have a $10\%$ probability of a perfect 10
GPA (lines 3-4), otherwise the GPA is uniform over $[0, 10]$ (line 5).
Students from USA (line 6) have a $15\%$ probability of a perfect 4
GPA (lines 6-7), otherwise the GPA is uniform over $[0, 4]$ (line 8).
%
% (Recall that India uses a GPA scale between 0 and 10, and the USA uses
% a GPA scale between 0 and 4.)

\noindentparagraph{Prior Sum-Product Expression}
The graph in Fig.~\ref{fig:indian-gpa-spe-prior} shows the translated sum-product
expression for this program, which represents a sampler
for the distribution over program variables as follows:
\begin{enumerate*}[label=(\roman*)]
\item if a node is a sum ($+$), visit a random child with probability
  equal to the weight of the edge pointing to the child;
\item if a node is a product ($\times$), visit each child exactly
  once, in no specific order;
\item if a node is a leaf, sample a value from the distribution at the
  leaf and assign it to the variable at the leaf.
\end{enumerate*}
Similarly, the graph encodes the joint distribution of the
variables by treating
\begin{enumerate*}[label=(\roman*)]
\item each sum node as a probabilistic mixture;
\item each product node as a tuple of independent variables; and
\item each leaf node as a primitive random variable.
\end{enumerate*}
Thus, the prior distribution is:
\begin{align}
&\Pr[\mathtt{Nationality} = n, \mathtt{Perfect} = p, \mathtt{GPA} \le g]
  \label{eq:sonnetist-prior} \\
&=0.5 \big[
  \delta_{\rm India}(n)
    \cdot(0.1 [(\delta_{\rm True}(p) \cdot \bindicator{10 \le g})] \notag \\
&\;+ 0.9 [(\delta_{\rm False}(p) \cdot
  (g/10 \cdot \bindicator{0 \le g < 10} + \bindicator{10 \le g}))])
    \big] \notag \\
&+ 0.5 \big[
  \delta_{\rm USA}(n) \cdot(
    0.15 [(\delta_{\rm True}(p) \cdot \bindicator{4 \le g})] \notag \\
&\;+ 0.85 [(\delta_{\rm False}(p) \cdot
    (g/4 \cdot \bindicator{0 \le g < 4} + \bindicator{4 \le g}))])
  \big]. \notag
\end{align}
Fig.~\ref{fig:indian-gpa-query-marginal} shows \sppl{} queries for the
prior marginal distributions of the three variables, plotted
in Fig.~\ref{fig:indian-gpa-probs-prior}.
The two jumps in the cumulative distribution function
(CDF\footnote{
  For a real-valued random variable $X$, the cumulative distribution
  function $F\,{:}\,\dom{Real}\,{\to}\,[0,1]$ is given by $F(r) \defas
  \Pr[X \le r]$.})
of \texttt{GPA} at 4 and 10 correspond to the atoms
that occur when \texttt{Perfect} is true.
The piecewise linear behavior on $[0,4]$ and $[4,10]$ follows from the
conditional uniform distributions of \texttt{GPA}.

\noindentparagraph{Conditioning the Program} Fig.~\ref{fig:indian-gpa-condition}
shows an example of the $\kw{condition}$ query, which specifies an event $e$ on which
to constrain executions of the program.
An event is a predicate on (possibly transformed) program variables that
can be used for both $\kw{condition}$
(Fig.~\ref{fig:indian-gpa-condition})
and $\kw{prob}$ (Fig.~\ref{fig:indian-gpa-query-joint}).
\sppl{} is the first system with inference algorithms for sum-product
expressions that handle predicates of this form.
Given $e$, the object of inference is the full posterior distribution:
\begin{align}
\footnotesize
&\Pr[\texttt{Nationality}\,{=}\,n, \texttt{Perfect}\,{=}\,p, \texttt{GPA}\,{\le}\,g \mid e]
\label{eq:truss} \\
& \defas
    {\Pr[
      \texttt{Nationality}\,{=}\,n,
      \texttt{Perfect}\,{=}\,p,
      \texttt{GPA}\,{\le}\,g, e]}/{\Pr[e]}. \notag
\end{align}

\noindentparagraph{Posterior Sum-Product Expression}
Given the prior expression (Fig.~\ref{fig:indian-gpa-spe-prior})
and conditioning event $e$ (Fig.~\ref{fig:indian-gpa-condition}), \sppl{} produces a new
expression (Fig.~\ref{fig:indian-gpa-spe-posterior}) that specifies
a distribution which is precisely equal to Eq.~\eqref{eq:truss},
From Thm.~\ref{thm:closure}, conditioning an \sppl{} program on any
event that can be specified in the language
results in a posterior distribution that also admits an exact
sum-product expression.
Conditioning on $e$ performs several transformations on the prior
expression, which are:
\begin{enumerate}
\item Eliminating the subtree rooted at the parent of leaf $\delta_{10}$,
  which is inconsistent with the conditioning event.
\item Rescaling the distribution $U(0,10)$ at the leaf node
  in the India subtree to $U(8,10)$.
\item Rescaling the distribution $U(0,4)$ at the leaf node
  in the USA subtree to $U(3,4)$.
\item Reweighting the branch probabilities of the sum node in the USA subtree
  from $[.15, .85]$ to $[.41, .59]$, where $.41 = .15 / (.15 + .2125)$
  is the posterior probability of $(\texttt{Perfect} = 1, \texttt{GPA} = 4)$
  given the condition $e$.
  % \begin{align*}
  % \Pr[\texttt{Perfect} = 1, \texttt{GPA} = 4 \mid
  %   \texttt{Nationality} = \texttt{\textquotesingle USA\textquotesingle}, \texttt{GPA} > 3]
  %   &= (.15 \times 1) / c = .15 / c\\
  % \Pr[\texttt{Perfect} = 0, 3 < \texttt{GPA} < 4 \mid
  %   \texttt{Nationality} = \texttt{\textquotesingle USA\textquotesingle}, \texttt{GPA} > 3]
  %   &= (.85 \times .25) / c = .2125 / c.
  % \end{align*}
\item Reweighting the branch probabilities at the root from $[.5, .5]$ to
  $[.33, .67]$ (same rules as in the previous item).
\end{enumerate}

Fig.~\ref{fig:indian-gpa-spe-posterior} shows the posterior expression
obtained by applying these transformations.
Using this expression, the right-hand side of
Eq.~\eqref{eq:truss}, which is the object of inference, is
\begin{align}
&\Pr[\mathtt{Nationality} = n, \mathtt{Perfect} = p, \mathtt{GPA} \le g \mid e]
  \label{eq:sonnetist-post} \\
&=.33 \big[
  \delta_{\rm India}(n) \cdot \delta_{\rm False}(p) \cdot
  (\textstyle\frac{g-8}{2}\cdot\bindicator{8 \,{\le}\, g < 10} + \bindicator{10 \,{\le}\, g})
  \big] \notag \\
&\quad + .67 \big[
  \delta_{\rm USA}(n) \cdot(
    .41 [(\delta_{\rm True}(p) \cdot \bindicator{4 \le g})] \notag\\
&\quad + .59 [(\delta_{\rm False}(p) \cdot
  (\textstyle\frac{g}{4} \cdot \bindicator{0 \,{\le}\, g < 4} + \bindicator{4 \,{\le}\, g}))])
  \big]. \notag
\end{align}
(Floats are shown to two decimal places.)
We can now run the $\kw{prob}$ queries in
Fig.~\ref{fig:indian-gpa-query-marginal} on the conditioned program to
plot the posterior marginal distributions, which
are shown in Fig.~\ref{fig:indian-gpa-probs-posterior}.
The example in Fig.~\ref{fig:indian-gpa} illustrates a typical modular
workflow in \sppl{} (Fig.~\ref{fig:system-diagram}), where modeling
(Fig.~\ref{fig:indian-gpa-program}), conditioning
(Fig.~\ref{fig:indian-gpa-condition}) and querying
(Figs.~\ref{fig:indian-gpa-query-marginal}--\ref{fig:indian-gpa-query-joint})
are separated into distinct and reusable stages that together express
the main components of Bayesian modeling and inference.

\subsection{Scalable Inference in a Hierarchical HMM}
\label{subsec:example-hmm}

\input{figures/hmm}

The next example shows how to
perform efficient smoothing in a hierarchical hidden Markov model (HMM~\citep{murphy2001})
and illustrates the optimization techniques used by
the \sppl{} translator (Sec.~\ref{subsec:translation-opt}), which exploit
conditional independence to ensure that the size of the sum-product expression
grows linearly in the number of timesteps.
The code box in Fig.~\ref{fig:hmm-program} shows a hierarchical hidden
Markov model with Bernoulli hidden states $Z_t$ and Normal--Poisson
observations $(X_t, Y_t)$.
The \texttt{separated} variable indicates whether the mean values of
$X_t$ and $Y_t$ at $Z_t=0$ and $Z_t=1$ are well-separated.
% for example,
% \texttt{mu\_x} specifies that if $\texttt{separated} = 0$, then mean
% of $X_t$ is 5 when $Z_t=0$ and 7 when $Z_t=1$, else if
% $\texttt{separated} = 1$, then the mean of $X_t$ is 15 when $Z=1$
% (and similarly for \texttt{mu\_y} and $Y$).
%
The \texttt{p\_transition} vector specifies that the current state
$Z_{t}$ switches from the previous state $Z_{t-1}$ with $20\%$
probability.
This example leverages the \sppl{} \kw{array}, \kw{for}, and
\kw{switch}-\kw{cases} statements, where the latter is a macro that
expands to $\kw{if}$-$\kw{else}$ statements (as in, e.g., the C language):
\begin{align}
&\kw{switch}\; x\; \kw{cases}\, \sexpr{x' \kw{in}\, \mli{values}}\, \settt{C}
  \label{eq:perquisitor} \\
&\overset{\rm desugar}{\rightsquigarrow} \begin{aligned}[t]
  &\kw{if}\; \texttt{(}x\, \kw{in}\, \mli{values}[0]\texttt{)}\; \settt{C[x'/\mli{values}[0]]} \notag \\
  &\kw{elif} \dots \notag \\
  &\kw{elif}\; \texttt{(} x\, \kw{in}\, \mli{values}[n{-}1]\texttt{)}\; \settt{C[x'/\mli{values}[n-1]]}, \notag
\end{aligned}
\end{align}
where $n$ is the length of $\mli{values}$ and $C[x/E]$ indicates syntactic
replacement of $x$ with expression $E$ in command $C$.

The top and middle plots in Fig.~\ref{fig:hmm-inference-marginal} show a
realization of $X$ and $Y$ that result from simulating the process
for 100 time steps.
The blue and orange regions along the x-axes indicate whether the true
hidden state $Z$ is 0 or 1, respectively (these ``ground-truth''
values of $Z$ are not observed but need to be inferred from $X$ and $Y$).
The bottom plot in Fig.~\ref{fig:hmm-inference-marginal} shows the
exact posterior marginal probabilities
$\Pr[Z_t = 1 \mid x_{0:99}, y_{0:99}]$
for each $t=0,\dots,99$ as inferred by $\sppl{}$ (an
inference known as ``smoothing'').
These probabilities track the true
hidden state, i.e., the posterior probabilities that $Z_t=1$ are low
in the blue and high in the orange regions.

Fig.~\ref{fig:hmm-spe-full} shows a ``naive'' sum-product expression
for the distribution of all program variables up to the first two time steps.
This expression is a sum-of-products, where the products in the
second level are an enumeration of all possible realizations of
program variables, so that the number of terms scales exponentially in
the number of time steps.
Fig.~\ref{fig:hmm-spe-factorized} shows the expression constructed by
\sppl{}, which is (conceptually) based on factoring and sharing common
terms in the two level sum-of-products in Fig.~\ref{fig:hmm-spe-full}.
These factorizations and deduplications exploit conditional
independences and repeated structure in the program
(Sec.~\ref{subsec:translation-opt}), which here delivers a expression
whose size scales linearly in the number of time points.
\sppl{} can also efficiently solve variants of smoothing, e.g.,
computing posterior marginals $\Pr[Z_t \mid x_{0:t}, y_{0:t}]$ (filtering)
or the posterior joint $\Pr[Z_{0:t} \mid x_{0:t}, y_{0:t}]$ for any $t$.

%!TEX root=./paper.tex

\section{A Core Calculus for Sum-Product Expressions}
\label{sec:core}

This section presents a semantic domain of sum-product expressions
that generalizes sum-product networks~\citep{poon2011} and enables
precise reasoning about them.
This domain will be used to
\begin{enumerate*}[label=(\roman*)]
\item establish the closure of sum-product expressions under
conditioning on events expressible in the calculus
(Thm.~\ref{thm:closure});

\item describe sound algorithms for exact Bayesian inference in
our system (Appx.~\ref{appx:condition}); and

\item describe a procedure for translating a probabilistic program
into a sum-product expression in the core language
(Sec.~\ref{sec:translation}).
\end{enumerate*}
Lst.~\ref{lst:core-semantics} shows denotations of the key syntactic
elements (Lst.~\ref{lst:core-syntax} in Appx.~\ref{appx:syntax-core-calculus})
in the calculus, which includes
real and nominal outcomes (Lst.~\ref{lst:core-semantics-outcomes});
real transforms (Lst.~\ref{lst:core-semantics-transformations});
predicates with pointwise and set-valued constraints (Lst.~\ref{lst:core-semantics-events});
primitive distributions (Lst.~\ref{lst:core-semantics-distributions}); and
multivariate distributions specified
  compositionally as sums and products of primitive
  distributions (Lst.~\ref{lst:core-semantics-sum-product}).

%!TEX root = ./paper.tex

\begin{listing*}[!t]
\centering
\footnotesize
% \hrule
\FrameSep0pt
\begin{framed}
\begin{sublisting}[b]{.5\textwidth}
\mathleft{1pt}
\begin{align*}
&\dom{Outcome} \defas \dom{Real} + \dom{String} \\[-2pt]
&\valfunc{V}: \dom{Outcomes} \to \mathcal{P}(\dom{Outcome}) \span\span \\[-2pt]
% EmptySet.
&\Denotv{V}{\varnothing} \defas \varnothing && [\dom{Empty}] \\[-2pt]
% FiniteSetS.
&\Denotv{V}{\settt{s_1\dots s_m}^{b}}
  \defas
  \begin{aligned}[t]
    &\bif\; b\; \bthen\; \cup_{i=1}^{m}\set{(\inj[s_i]{\dom{String}}{\dom{Outcome}})} \\[-2pt]
    &\belse\; \set{(\inj[s]{\dom{String}}{\dom{Outcome}}) \mid \forall i.s \ne s_i}
    \end{aligned}
    && [\dom{FiniteStr}]\\[-2pt]
% FiniteSetR.
&\Denotv{V}{\settt{r_1\dots r_m}}
  \defas \cup_{i=1}^{m}\set{(\inj[r_i]{\dom{Real}}{\dom{Outcome}})}
  && [\dom{FiniteReal}]\\[-2pt]
% Interval.
&\Denotv{V}{\sexpr{\sexpr{b_1\,r_1}\,\sexpr{r_2\,b_2}}}
  \defas \begin{aligned}[t]
    &\set{
    (\inj[r]{\dom{Real}}{\dom{Outcome}}) \mid r_1 {<_{b_1}} r {<_{b_2}} r_2} \\[-2pt]
    &\mbox{where } <_{\ttrue}\defas<; <_{\tfalse}\defas\le; r_1 < r_2
    \end{aligned}
    && [\dom{Interval}]\\[-2pt]
% Union.
&\Denotv{V}{v_1 \amalg \dots \amalg v_m}
  \defas \cup_{i=1}^{m} \Denotv{V}{v_i}
  && [\dom{Union}]
%
% &\Sigma(\dom{Outcome}) \defas
%   \dom{Borel}(\tau_{\dom{Real}} \uplus \tau_{\dom{String}})
\end{align*}
\captionsetup{skip=-5pt, textfont=bf}
\caption{Outcomes}
\label{lst:core-semantics-outcomes}
\hrule

\begin{align*}
&\valfunc{T}: \dom{Transform} \to \dom{Real} \to \dom{Real} \\[-4pt]
&\Denot[\valfunc{T}]{\scall{Id}{x}}
  \defas \lambda r'.r';\;
\Denot[\valfunc{T}]{\scall{Reciprocal}{t}}
  \defas \lambda r'.1/\left({\Denot[\valfunc{T}]{t}(r')}\right); \\[-4pt]
% &\Denot[\valfunc{T}]{\scall{Root}{t\; n}}
%   \defas \lambda r'.\sqrt[n]{\Denot[\valfunc{T}]{t}(r')} \\[-4pt]
% &\Denot[\valfunc{T}]{\scall{Log}{t\; r}}
%   \defas \lambda r'.\log_{r}\left({\Denot[\valfunc{T}]{t}(r')}\right)\;
%   && (\mbox{iff } 0 < r) \\[-4pt]
&\Denotv{T}{\scall{Abs}{t}}
  \defas \lambda r'.|\Denotv{T}{t}(r')|;
  \Denotv{T}{\scall{Root}{t\, n}}
  \defas \lambda r'.\sqrt[n]{\Denot[\valfunc{T}]{t}(r')}; \\[-4pt]
&\Denotv{T}{\scall{Poly}{t\; r_0\; \dots\; r_m}}
  \defas \lambda r'.\textstyle\sum_{i=0}^{m} r_i\left({\Denotv{T}{t}(r')}\right)^i;
  \dots
\end{align*}
\captionsetup{skip=-6pt, textfont=bf}
\caption{Transformations {\normalfont (Lst.~\ref{lst:transform} in Appx.~\ref{appx:transforms-valuation})}}
\label{lst:core-semantics-transformations}
\hrule

\begin{align*}
&\valfunc{E}: \dom{Event} \to \dom{Var} \to \dom{Outcomes} \\[-4pt]
&\Denotv{E}{\sexpr{t\, \token{in}\, v}} x\, \defas
  \bif\, (\domfunc{vars}\, t) = \set{x} \, \bthen\, (\domfunc{preimg}\, t\, v) \, \belse\, \varnothing
  && [\dom{Contains}] \\[-4pt]
&\Denotv{E}{e_1 \sqcap \dots \sqcap e_m} x
  \defas (\domfunc{intersection}\, \Denotv{E}{\mli{e}_1}x\, \dots\, \Denotv{E}{\mli{e}_m}x)
    &&[\dom{Conjunction}] \\[-4pt]
&\Denotv{E}{e_1 \sqcup \dots \sqcup e_m} x
  \defas (\domfunc{union}\; \Denotv{E}{\mli{e}_1}x\; \dots\; \Denotv{E}{\mli{e}_m}x)
  &&[\dom{Disjunction}]
\end{align*}
\captionsetup{skip=-5pt, textfont=bf}
\caption{Events}
\label{lst:core-semantics-events}
\hrule

\begin{comment}
\begin{enumerate}[wide=3pt,labelsep=2pt,align=left,leftmargin=*,label=(C\arabic*)]
\item \label{item:definedness-env-base}
  $\forall\, \scall{Leaf}{x\; d\; \sigma}$.
  $x \in \sigma$ and $\sigma(x) = \token{Id}\sexpr{x}$.

\item \label{item:definedness-env-topo}
  $\forall\, \scall{Leaf}{x\; d\; \sigma}$.
  $\forall m$.
  If $\set{x, x_1, \dots, x_m} = \mathrm{dom}(\sigma)$ \\
  then $(\domfunc{vars}\; \sigma(x_m)) \subset \set{x, x_1, \dots, x_{m-1}}$.

\item \label{item:definedness-prod-scope}
  $\forall \sexpr{S_1 \otimes \dots \otimes S_m}$.
  $\forall i {\ne} j$.
  $(\domfunc{scope}\, S_i) \cap (\domfunc{scope}\, S_j) = \varnothing$.

\item \label{item:definedness-sum-scope}
$\forall \sexpr{S_1\, w_1}\oplus \dots \oplus \sexpr{S_m\,w_m}$.
  $\forall i$. $(\domfunc{scope}\, S_i) = (\domfunc{scope}\, S_1)$.

\item \label{item:definedness-sum-weights}
$\forall \sexpr{S_1\, w_1}\oplus \dots \oplus \sexpr{S_m\,w_m}$.
  $w_1 + \dots + w_n > 0$.
\end{enumerate}
\captionsetup{skip=4pt, textfont=bf}
\caption{Definedness Conditions for Sum-Product}
\label{lst:core-semantics-definedness}
\end{comment}
\begin{align*}
&\Denot[\mathbb{P}_0]{S}: \SPE \to \dom{Event} \to \dom{Natural} \times [0, \infty) \\[-4pt]
&\Denot[\mathbb{P}_0]{\scall{Leaf}{x\, d\, \sigma}}\;
  \sexpr{\scall{Id}{x}\,\token{in}\,\settt{\mli{rs}}} \defas \bmatch\, d && [\dom{Leaf}] \\[-4pt]
  &\quad \begin{aligned}[t]
    &\vartriangleright \scall{DistR}{F\, r_1\, r_2} \Rightarrow\, \bmatch\, {\mli{rs}} \\[-4pt]
    &\quad \begin{aligned}[t]
      &\vartriangleright r \Rightarrow
        (1, \bindicator{r_1 \le r \le r_2} F'(r) / \left[F(r_2) - F(r_1)\right]) \\[-4pt]
      &\vartriangleright s \Rightarrow (1, 0)
    \end{aligned} \\[-4pt]
    &\vartriangleright \belse \Rightarrow
      \blet\, w\, \bbe\, \Denotv{P}{{\scall{Leaf}{x\, d\, \sigma}}}\sexpr{\scall{Id}{x}\,\token{in}\,\settt{\mli{rs}}}
      \,\bin\, (\bindicator{w = 0}, w)
    \end{aligned} \span\span \\[-4pt]
&\Denot[\mathbb{P}_0]{\sexpr{S_1\, w_1}\oplus \dots \oplus \sexpr{S_m\,w_m}}\;
  \sqcap_{i=1}^{\ell} \sexpr{\scall{Id}{x_i}\,\token{in}\,\settt{\mli{rs_i}}}
  \defas \\[-4pt]
  &\quad \begin{aligned}[t]
    &\blet_{1 \le i \le m}\, (d_i, p_i) \,\bbe\,
        \Denot[\mathbb{P}_0]{S_i}
        \left( \sqcap_{i=1}^{\ell} \sexpr{\scall{Id}{x_i}\,\token{in}\,\settt{\mli{rs_i}}}\right)\\[-4pt]
    &\bin\, \begin{aligned}[t]
      &\bif\, \forall_{1 \le i \le m}.\ p_i = 0\; \bthen\; (1, 0) \\[-4pt]
      &\belse\, \begin{aligned}[t]
        &\blet\, d^* \,\bbe\, \min\set{d_i \mid 1\le i \le m, 0 < p_i} \\[-4pt]
        &\bin\,  (d^*, \textstyle\sum_{i=1}^{m}\bindicator{d_i = d^*}w_i p_i)
        \end{aligned}
      \end{aligned}
    \end{aligned} && [\dom{Sum}]\\[-4pt]
&\Denot[\mathbb{P}_0]{S_1 \otimes \dots \otimes S_m}\;
  \sqcap_{i=1}^{\ell} \sexpr{\scall{Id}{x_i}\,\token{in}\,\settt{\mli{rs_i}}}
  \defas && [\dom{Product}] \\[-4pt]
  &\quad \begin{aligned}[t]
    &\blet_{1 \le i \le m}\, (d_i, p_i) \,\bbe\,
      \bmatch\; \set{x_1, \dots, x_m} \cap (\domfunc{scope}\, S_i) \\[-4pt]
      &\quad \vartriangleright \set{n_1, \dots, n_k}
        \Rightarrow  \Denot[\mathbb{P}_0]{S_i}
          \sqcap_{t=1}^{k} \sexpr{\scall{Id}{x_{n_t}}\,\token{in}\,\settt{\mli{rs_t}}} \\[-4pt]
      &\quad \vartriangleright \set{} \Rightarrow (0, 1) \\[-4pt]
    &\bin\, (\textstyle \sum_{i=1}^{n}d_i, \textstyle\prod_{i=1}^{m} p_i)
    \end{aligned}
\end{align*}
\captionsetup{skip=0pt, textfont=bf}
\caption{Sum-Product Expressions (Density Semantics)}
\label{lst:core-semantics-sum-product-density}
\end{sublisting}\vrule%
\begin{sublisting}[b]{.5\textwidth}
\mathleft{2pt}
\begin{align*}
&\valfunc{D}: \dom{Distribution} \to \dom{Outcomes} \to [0,1]\\[-2pt]
% DistStr
&\Denotv{D}{\scall{DistS}{\sexpr{s_i\, w_i}_{i=1}^{m}}} v
  \defas && [\dom{DistStr}]\\[-2pt]
&\begin{aligned}[t]
  &\; \bmatch\, (\domfunc{intersection}\, v\, \settt{s_1 \dots s_m}^{\tfalse}) \\[-2pt]
  &\quad \vartriangleright
    \varnothing
      \gor \settt{r'_1 \dots r'_m}
      \gor \sexpr{\sexpr{b_1\,r_1}\, \sexpr{r_2\,b_2}} \Rightarrow 0 \\[-2pt]
  &\quad \vartriangleright v_1 \amalg \dots \amalg v_m \Rightarrow
    \textstyle\sum_{i=1}^{m}
      \Denotv{D}{\scall{DistS}{\sexpr{s_i\, w_i}_{i=1}^{m}}}v_i\\[-2pt]
  &\quad \vartriangleright
    \settt{s'_1\, \dots\, s'_k}^{b}
    \Rightarrow \begin{aligned}[t]
    &\blet\, w\, \bbe\,
      \textstyle \sum_{i=1}^{m}(w_i \,\bif\, s_i \in \set{s'_j}_{j=1}^k\, \belse\, 0)\\[-2pt]
    &\bin\, \bif\, \bar{b}\, \bthen\, w\, \belse\, 1-w
  \end{aligned}
\end{aligned} \span\span \\[-2pt]
% DistReal (Continuous)
&\Denotv{D}{\scall{DistR}{F\, r_1\, r_2}} v
  \defas \bmatch\, (\domfunc{intersection}\,
    \sexpr{\sexpr{\tfalse\,r_1}\, \sexpr{r_2\,\tfalse}}\, v) \hspace{-.25cm}
    && [\dom{DistReal}] \\[-2pt]
&\begin{aligned}[t]
  &\quad \vartriangleright
    \varnothing
      \gor \settt{r'_1 \dots r'_m}
      \gor \settt{s'_1\, \dots\, s'_k}^{b} \Rightarrow 0 \\[-2pt]
  &\quad \vartriangleright v_1 \amalg \dots \amalg v_m \Rightarrow
    \textstyle\sum_{i=1}^{m}
      \Denotv{D}{\scall{DistR}{F\, r_1\, r_2}}v_i\\[-2pt]
  &\quad \vartriangleright
    \sexpr{\sexpr{b'_1\,r'_1}\, \sexpr{r'_2\,b'_2}}
    \Rightarrow \frac{F(r'_2) - F(r'_1)}{F(r_2) - F(r_1)}
\end{aligned} \\[-2pt]
% DistReal (Atomic)
&\Denotv{D}{\scall{DistI}{F\, r_1\, r_2}} v
  \defas \bmatch\, (\domfunc{intersection}\,
    \sexpr{\sexpr{\tfalse\,r_1}\, \sexpr{r_2\,\tfalse}}\, v) \hspace{-.25cm}
    && [\dom{DistInt}] \\[-2pt]
&\begin{aligned}[t]
  &\quad \vartriangleright
    \varnothing
      \gor \settt{s'_1\, \dots\, s'_k}^{b} \Rightarrow 0 \\[-2pt]
  &\quad \vartriangleright v_1 \amalg \dots \amalg v_m \Rightarrow
    \textstyle\sum_{i=1}^{m}
      \Denotv{D}{\scall{DistI}{F\, r_1\, r_2}} v_i \\[-2pt]
  &\quad \vartriangleright \settt{r'_1 \dots r'_m}
    \Rightarrow \begin{aligned}[t]
    \frac
      {\displaystyle\sum_{i=1}^{m}
        \left[\begin{aligned}[m]
        &\bif\, (r'_i = \floor{r'_i})\, \wedge (r_1 \le r'_i \le r_2) \\[-2pt]
        &\bthen\, F(r') - F(r'-1) \,\belse\, 0
      \end{aligned}\right]}
      {F(\floor{r_2}) - F(\ceil{r_1}-1)}
  \end{aligned} \\[-2pt]
  &\quad \vartriangleright
    \sexpr{\sexpr{b'_1\,r'_1}\, \sexpr{r'_2\,b'_2}}
    \Rightarrow \begin{aligned}[t]
      &\blet\; \tilde{r}_1 \;\bbe\;
        \floor{r'_1} - \bindicator{(r'_1 = \floor{r'_1}) \wedge \bar{b'_1}} \\[-2pt]
      &\bin\,\blet\; \tilde{r}_2 \;\bbe\;
        \floor{r'_2} - \bindicator{(r'_2 = \floor{r'_2}) \wedge \bar{b'_2}} \\[-2pt]
      &\bin\, \frac{F(\tilde{r}_2) - F(\tilde{r}_1)}{F(\floor{r_2}) - F(\ceil{r_1}-1)}
    \end{aligned}
\end{aligned}\span\span
\end{align*}
\captionsetup{skip=4pt, textfont=bf}
\caption{Primitive Distributions}
\label{lst:core-semantics-distributions}
\hrule

\begin{align*}
&\valfunc{P}: \SPE \to \dom{Event} \to [0,1] \\[-3pt]
% Leaf
&\Denotv{P}{\scall{Leaf}{x\, d\, \sigma}}e
  \defas \Denotv{D}{d}(\Denotv{E}{(\domfunc{subsenv}\; e\; \sigma)}x)
  && [\dom{Leaf}] \\[-3pt]
% Sum
&\Denotv{P}{\sexpr{S_1\, w_1}\oplus \dots \oplus \sexpr{S_m\,w_m}}e
  \defas \begin{aligned}[t]
    &\blet\, Z\, \bbe\, \textstyle\sum_{i=i}^m{w_i} \\[-3pt]
    &\bin\, \textstyle\sum_{i=1}^{m} (\Denotv{P}{S_i}e){w_i} / Z
  \end{aligned} && [\dom{Sum}]\\[-3pt]
% Product
&\Denotv{P}{S_1 \otimes \dots \otimes S_m}e \defas \bmatch\; (\domfunc{dnf}\, e) && [\dom{Product}] \\[-3pt]
&\quad \begin{aligned}[t]
  &\vartriangleright \sexpr{t \; \token{in}\; v} \Rightarrow
  \begin{aligned}[t]
    &\blet\; n \;\bbe\, \min\set{1\,{\le}i\,{\le}m \mid (\domfunc{vars}\, e) \subset (\domfunc{scope}\, S_i)} \\[-3pt]
    &\bin\; \Denotv{P}{S_n}e
    \end{aligned} \span\span\\[-3pt]
  &\vartriangleright \sexpr{e_1 \sqcap \dots \sqcap e_\ell} \Rightarrow \span\span\\[-3pt]
  &\qquad \prod_{1 \le i \le m} \left[\begin{aligned}[m]
    &\bmatch\; \set{1 \le j \le \ell \mid (\domfunc{vars}\, e_j) \subset (\domfunc{scope}\, S_i) }\span\span\\[-3pt]
    &\vartriangleright \set{n_1, \dots, n_k}
      \Rightarrow  \Denotv{P}{S_i}(e_{n_1} \sqcap \dots \sqcap e_{n_k}) \span\span\\[-4pt]
    &\vartriangleright \set{} \Rightarrow 1
    \end{aligned}\right]
  \span\span\\[-4pt]
  &\vartriangleright \sexpr{e_1 \sqcup \dots \sqcup e_\ell} \Rightarrow
  \sum\limits_{J \subset [\ell]}\left[
        (-1)^{|J|-1}\;\Denotv{P}{S_1 \otimes \dots \otimes S_m}
        (\sqcap_{i \in J}\, e_i)\right]
\end{aligned} \hspace{-.25cm}
\end{align*}
\captionsetup{belowskip=0pt,aboveskip=0pt, textfont=bf}
\caption{Sum-Product Expressions (Distribution Semantics)}
\label{lst:core-semantics-sum-product}
\end{sublisting}
\end{framed}

% \hrule
\captionsetup{aboveskip=2pt}
\caption{Syntax and semantics of a core calculus for sum-product expressions and related domains.}
\label{lst:core-semantics}
\end{listing*}

\noindentparagraph{Basic Outcomes} Random variables in the calculus take values in the
$\dom{Outcome} \defas \dom{Real} + \dom{String}$ domain.
The symbol $+$ here indicates a sum (disjoint-union) data type, whose
elements are formed by the injection operation, e.g.,
$\inj[r]{\dom{Real}}{\dom{Outcome}}$ for $r \in \dom{Real}$.
This domain is used to model mixed-type random variables, such
as $X$ in the following \sppl{} program:
\begin{lstlisting}[style=sppl,frame=single,basicstyle=\ttfamily\footnotesize]
Z ~ normal(0, 1)
if   (Z <= 0):    X ~ "negative" # string
elif (0 < Z < 4): X ~ 2*exp(Z)   # continuous real
elif (4 <= Z):    X ~ atomic(4)  # discrete real
\end{lstlisting}
The $\dom{Outcomes}$ domain
% An element $v \in \dom{Outcomes}$ (Lst.~\ref{lst:core-semantics-outcomes})
denotes a subset of $\dom{Outcome}$ as defined by the valuation
function $\valfunc{V}$ (Lst.~\ref{lst:core-semantics-outcomes}).
For example,
$\sexpr{\sexpr{b_1\,r_1}\,\sexpr{r_2\,b_2}}$ specifies a
(open, closed, or clopen) real interval
and $\settt{s_1\dots s_m}^{b}$ is a set of strings, where $b = \ttrue$
indicates the complement (meta-variables such as $m$
indicate an arbitrary but finite number of repetitions of
a domain variable or subexpression).
The operations $\domfunc{union}$, $\domfunc{intersection}$, and
$\domfunc{complement}$ operate on $\dom{Outcomes}$ in the usual way
(while preserving certain semantic invariants, see Appx.~\ref{appx:aux-fn})

\noindentparagraph{A Sigma Algebra of Outcomes}  To speak precisely about random
variables and measures on $\dom{Outcome}$,
we define a sigma-algebra $\mathcal{B}(\dom{Outcome}) \subset \mathcal{P}(\dom{Outcome})$ as follows:
\begin{enumerate}[wide=0pt]
\item Let $\tau_{\dom{Real}}$ be the usual topology on $\dom{Real}$.
\item Let $\tau_{\dom{String}}$ be the discrete topology on $\dom{String}$.
\item Let $\tau_{\dom{Outcome}} \defas \tau_{\dom{Real}} \uplus \tau_{\dom{String}}$
be the disjoint-union topology on $\dom{Outcome}$,
where $U$ is open iff
$\set{r \mid (\inj[r]{\dom{Real}}{\dom{Outcome}}) \in U}$ is open in $\dom{Real}$ and
$\set{s \mid (\inj[s]{\dom{String}}{\dom{Outcome}}) \in U}$ is open in $\dom{String}$.
\item Let $\mathcal{B}(\dom{Outcome})$ be the Borel sigma-algebra of $\tau_{\dom{Outcome}}$.
\end{enumerate}

\begin{remark}
\label{remark:outcomes-measurable}
As measures on $\dom{Real}$ are defined by their values on open
intervals and measures on $\dom{String}$ on singletons, we can speak
of probability measures on $\mathcal{B}(\dom{Outcome})$ as mappings
from $\dom{Outcomes}$ to $[0,1]$.
\end{remark}

%!TEX root = ../paper.tex

\begin{figure*}[t]
\centering

\begin{subfigure}{\linewidth}
\begin{subfigure}[b]{.17\linewidth}
\begin{lstlisting}[style=sppl,frame=single,basicstyle=\ttfamily\scriptsize]
X ~ normal(0, 2)
if X < 1:
 Z ~ -X**3 + X**2 + 6*X
else:
 Z ~ 5*sqrt(X) + 11
\end{lstlisting}
\caption{Prior Program}
\label{subfig:poly-invert-prior-program}
\end{subfigure}\hfill
\begin{subfigure}[b]{.3\linewidth}
\begin{adjustbox}{max width=\linewidth}
\begin{tikzpicture}
\Tree
  [.\node[circle,draw,inner sep = 1pt,](root){$+$};
    % Subtree.
    \edge node[auto=right]{.69};
    [.\node[name=t1, label={[name=t1l,label distance=.25cm]above left:{$({-\infty}, 1)$}}]{$X{\sim}{N}(0,2)$};
      \edge[-stealth,thick,color=red] node[auto=right,color=black]{$Z$};
      $-X^3{+}X^2{+}6X$ ]
    % Subtree.
    \edge node[auto=left]{.31};
    [.\node[name=t2,label={[name=t2l,label distance=.25cm]above right:{$[1,\infty)$}}]{$X{\sim}N(0,2)$};
      \edge[-stealth,thick,color=red] node[auto=right,color=black]{$Z$};
      $5\sqrt{X}{+}11$ ] ]
\draw[-stealth] (t1l) -- (t1);
\draw[-stealth] (t2l) -- (t2);
\end{tikzpicture}
\end{adjustbox}
\caption{Prior Sum-Product Expression}
\label{subfig:poly-invert-prior-spe}
\end{subfigure}\hfill
\begin{subfigure}[b]{.1\linewidth}
\begin{lstlisting}[style=sppl,frame=single,basicstyle=\ttfamily\scriptsize]
condition
  Z**2 <= 4
  and Z >= 0
\end{lstlisting}
\caption{Condition}
\label{subfig:poly-invert-condition-program}
\end{subfigure}\hfill
\begin{subfigure}[b]{.37\linewidth}
\begin{adjustbox}{max width=\linewidth}
\begin{tikzpicture}
% \tikzset{level distance=2cm}
\tikzset{level 1/.style={level distance=2cm}}
\tikzset{level 2/.style={level distance=1.25cm}}
\Tree
  [.\node[circle,draw,inner sep = 1pt,](root){$+$};
    % Subtree.
    \edge node[auto=right]{.16};
    [.\node[name=t1,label={[name=t1l,label distance=.25cm]above left:$[-2.2,-2]$}]{$X{\sim}N(0,2)$};
      \edge[-stealth,thick,color=red] node[auto=right,color=black]{$Z$};
      $-X^3{+}X^2{+}6X$ ]
    % Subtree.
    \edge node[pos=.5, inner sep =.5pt, auto=right]{.49};
    [.\node[name=t2,label={[name=t2l,label distance=.15cm]above right:$[0,0.32]$}]{$X{\sim}N(0,2)$};
      \edge[-stealth,thick,color=red] node[auto=right,color=black]{$Z$};
      $-X^3{+}X^2{+}6X$ ]
    % Subtree.
    \edge node[auto=left]{.35};
    [.\node[name=t3,label={[name=t3l,label distance=.25cm]above right:$[3.2,4.8]$}]{$X{\sim}N(0,2)$};
      \edge[-stealth,thick,color=red] node[auto=right,color=black]{$Z$};
      $5\sqrt{X}{+}11$ ] ]
\draw[-stealth] (t1l) -- (t1);
\draw[-stealth] (t2l) -- (t2);
\draw[-stealth] (t3l) -- (t3);
\end{tikzpicture}
\end{adjustbox}
\caption{Conditioned Sum-Product Expression}
\label{subfig:poly-invert-post-spe}
\end{subfigure}
\vspace{-4pt}
\end{subfigure}

\FrameSep0pt
\begin{framed}
\begin{subfigure}[b]{.5\linewidth}
\includegraphics[width=\textwidth]{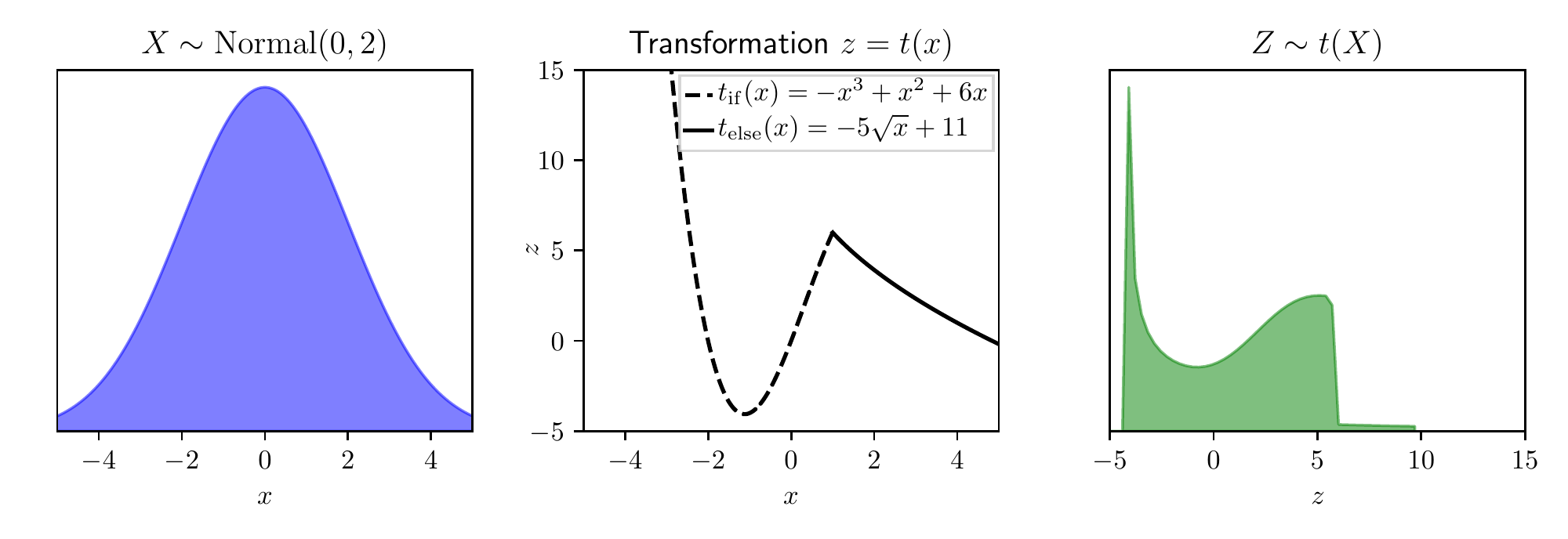}
\captionsetup{aboveskip=-5pt, belowskip=0pt}
\caption{Prior Marginal Distributions}
\label{subfig:poly-invert-prior-dist}
\end{subfigure}\vrule%
\begin{subfigure}[b]{.5\linewidth}
\includegraphics[width=\textwidth]{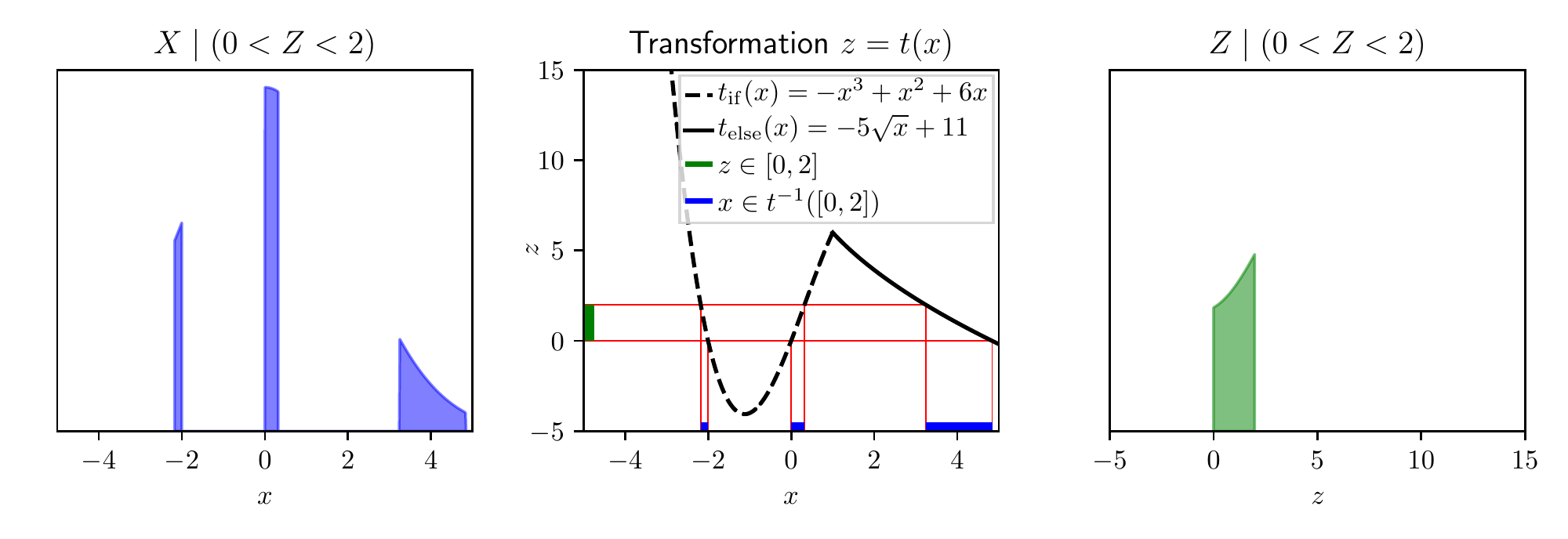}
\captionsetup{aboveskip=-5pt, belowskip=0pt}
\caption{Conditioned Marginal Distributions}
\label{subfig:poly-invert-post-dist}
\end{subfigure}
\end{framed}

\captionsetup{aboveskip=2pt,belowskip=0pt}
\caption{Inference on a stochastic many-to-one transformation of a real random variable in \sppl{}.}
\label{fig:poly-invert-main}
\label{fig:poly-invert}
% \vspace{-.5cm}
\end{figure*}

\noindentparagraph{Real Transformations} Lst.~\ref{lst:core-semantics-transformations}
shows real transformations that can be applied to
variables in the calculus.
The $\dom{Identity}$ $\dom{Transform}$, written $\scall{Id}{x}$, is a
terminal subexpression of any $\dom{Transform}$ $t$ and contains a single
variable name that specifies the ``dimension'' over which $t$ is defined.
The list of all transforms is in
Appx.~\ref{appx:transforms-valuation}.
The key operation involving transforms is computing the preimage of $\dom{Outcomes}$
$v$ under $t$ using
$\domfunc{preimg} : \dom{Transform} \to \dom{Outcomes} \to \dom{Outcomes}$
which satisfies the following properties:
{\small\begin{align*}
(\inj[r]{\dom{Real}}{\dom{Outcome}}) \in \Denotv{V}{\domfunc{preimg}\; t\; v}
  &\iff \Denotv{T}{t}(r) \in \Denotv{V}{v}
% \label{eq:dermogastric-1}
\\
(\inj[s]{\dom{String}}{\dom{Outcome}}) \in \Denotv{V}{\domfunc{preimg}\; t\; v}
  &\iff (t \in \dom{Identity}) \wedge (s \in \Denotv{V}v).
% \label{eq:dermogastric-2}
\end{align*}}%
Appx.~\ref{appx:transforms-preimage} presents a symbolic solver
that implements $\domfunc{preimg}$ for each $\dom{Transform}$,
which is leveraged to enable exact probabilistic inferences on
transformed variables in \sppl{}.
Fig.~\ref{fig:poly-invert-main} and Appx.~\ref{appx:transforms-example}
show example inferences with transforms.
%
% As with $\domfunc{union}$, $\domfunc{intersection}$, and
% $\domfunc{complement}$, all subexpressions in a $\dom{Union}$ returned
% by $\domfunc{preimg}$ are disjoint.

\noindentparagraph{Events} Lst.~\ref{lst:core-semantics-events} shows
the $\dom{Event}$ domain, which specifies predicates on variables.
The valuation $\Denotv{E}{e}: \dom{Var} \to \dom{Outcomes}$
of an $\dom{Event}$ takes a
variable $x$ and returns the set $v \in \dom{Outcomes}$ of elements
that satisfy the predicate along dimension $x$, leveraging the
properties of $\domfunc{preimg}$.
%
\begin{comment}
The following example shows how an ``informal'' predicate $\phi(X_1,X_2)$
maps to an $\dom{Event}$ $e$:
\begin{align}
\set{0\,{\le}\,X_1{<}\,1}\,{\cup}\,\set{1/X_2\,{>}\,6}
\,{\equiv}\,
\sexpr{\token{Id}\sexpr{\mathtt{X}_1}\, \token{in}\,
    \sexpr{\sexpr{\tfalse\,0}\,\sexpr{1\,\ttrue}}}
\,{\sqcup}\,
\sexpr{\texttt{1/}\token{Id}\sexpr{\mathtt{X}_2}\, \token{in}\,
    \sexpr{\sexpr{\ttrue\,6}\,\sexpr{\infty\,\ttrue}}},
\label{eq:villiferous}
\end{align}
so that
  $\Denotv{E}e \texttt{X}_1 = \sexpr{\sexpr{\tfalse\,0}\,\sexpr{1\,\ttrue}}$ and
  $\Denotv{E}e \texttt{X}_2 = \sexpr{\sexpr{\tfalse\,{-\infty}}\,\sexpr{6\,\tfalse}}$.
%
\end{comment}
This domain specifies measurable sets of an $n$-dimensional distribution
on variables $\set{x_1, \dots, x_n}$ as follows:
let $\sigma_{\rm gen}(\set{A_1, A_2, \dots})$ be the sigma-algebra
generated by $A_1, A_2, \dots$, and define
% \begin{equation}
$\mathcal{B}^n(\dom{Outcome})$
  $\defas \sigma_{\rm gen}(\set{\textstyle\prod_{i=1}^{n}U_i \mid \forall_{1 \le i \le n}.\ U_i \in \mathcal{B}(\dom{Outcomes})}).$
% \label{eq:leatherware}
% \end{equation}
%
In other words, $\mathcal{B}^n(\dom{Outcome})$ is the
$n$-fold product sigma-algebra generated by open rectangles
of $\dom{Outcomes}$.
Any $e \in \dom{Event}$
specifies a measurable set $U$ in $\mathcal{B}^n(\dom{Outcome})$,
whose $i$th coordinate $U_i = \Denotv{E}{e}x_i$ if $x_i \in
\domfunc{vars}\, e$, and $U_i = \dom{Outcomes}$ otherwise.
Any $\dom{Transform}$ in $e$ is solved and any $\dom{Var}$ that does
not appear in $e$ is marginalized, as in the next example.
\begin{example}
\label{example:event-set}
Let $\set{\token{X}, \token{Y}, \token{Z}}$ be elements of $\dom{Var}$.
Then
\begin{equation*}
e \defas \scall{Reciprocal}{\scall{Id}{\token{X}}} \, \token{in}\, \sintvl{\tfalse\;1}{2\;\tfalse}
\end{equation*}
corresponds to the $\mathcal{B}^3(\dom{Outcome})$-measurable set
\begin{equation*}
\set{\left(\inj[r]{\dom{Real}}{\dom{Outcome}}\right) \mid 1/2 \le r  \le 1}
\times \dom{Outcomes} \times \dom{Outcomes}.
\end{equation*}
\end{example}

As in Remark~\ref{remark:outcomes-measurable}, we may speak about
probability distributions on $\mathcal{B}^n(\dom{Outcome})$
as mappings from $\dom{Event}$ to $[0,1]$.

\noindentparagraph{Primitive Distributions} Lst.~\ref{lst:core-semantics-distributions}
presents the primitive distributions out of which multivariate
distributions are constructed.
The $\dom{CDF}$ domain contains cumulative distribution
functions $F$, whose quantile function is denoted $F^{-1}$ and
derivative $F'$.
$\dom{CDF}$ is in 1-1 correspondence with all distributions and
  random variables on $\dom{Real}$~\citep[Thms~12.4, 14.1]{billingsley1986}.
The $\dom{Distribution}$ domain specifies continuous real,
  atomic real (on the integers) and nominal distributions.
The denotation $\Denotv{D}{d}$ of a $\dom{Distribution}$
is a distribution
on $\dom{Outcomes}$ (recall Remark~\ref{remark:outcomes-measurable}).
For example, $\scall{DistR}{F\, r_1\, r_2}$ is the restriction of
$F$ to a positive measure interval $[r_1, r_2]$.
The distributions specified by $\token{DistR}$ and
$\token{DistI}$ can be simulated using a variant of the integral
probability transform
(Prop.~\ref{prop:inverse-cdf} in Appx.~\ref{appx:syntax-core-calculus}),
which also defines their sampling semantics.

\noindentparagraph{Sum-Product Expressions}
Lsts.~\ref{lst:core-semantics-sum-product-density}
and~\ref{lst:core-semantics-sum-product} show the probability density and
distribution semantics of the $\SPE$ domain, respectively, whose
elements are probability measures.
The following conditions specify well-definedness for $\SPE$:
\begin{enumerate}[wide=3pt,labelsep=2pt,align=left,leftmargin=*,label=(C\arabic*)]
\item \label{item:definedness-env-base}
  $\forall\, \scall{Leaf}{x\; d\; \sigma}$.
  $x \in \sigma$ and $\sigma(x) = \token{Id}\sexpr{x}$.

\item \label{item:definedness-env-topo}
  $\forall\, \scall{Leaf}{x\; d\; \sigma}$.
  If $\mathrm{dom}(\sigma) = \set{x, x_1, \dots, x_m}$ for some $m > 0$
  then $\forall_{1 \le t \le m}.\ (\domfunc{vars}\; \sigma(x_t)) \subset \set{x, x_1, \dots, x_{t-1}}$.

\item \label{item:definedness-prod-scope}
  $\forall \sexpr{S_1 \otimes \dots \otimes S_m}$.
  $\forall i {\ne} j$.
  $(\domfunc{scope}\, S_i) \cap (\domfunc{scope}\, S_j) = \varnothing$.

\item \label{item:definedness-sum-scope}
$\forall \sexpr{S_1\, w_1}\oplus \dots \oplus \sexpr{S_m\,w_m}$.
  $\forall i$. $(\domfunc{scope}\, S_i) = (\domfunc{scope}\, S_1)$.

\item \label{item:definedness-sum-weights}
$\forall \sexpr{S_1\, w_1}\oplus \dots \oplus \sexpr{S_m\,w_m}$.
  $w_1 + \dots + w_n > 0$.
\end{enumerate}

For $\dom{Leaf}$,
\ref{item:definedness-env-base} ensures that
  $\sigma$ maps the leaf variable $x$ to the $\dom{Identity}$ $\dom{Transform}$ and
\ref{item:definedness-env-topo} ensures
  there are no cyclic dependencies or undefined variables in $\dom{Environment}$ $\sigma$.
Condition
  \ref{item:definedness-prod-scope} ensures the scopes of all children
  of a $\dom{Product}$ are disjoint and
  \ref{item:definedness-sum-scope} ensures the scopes of all children
  of a $\dom{Sum}$ are identical,
  which together ensure completeness and
  decomposability from sum-product networks~\citep[Defs.~4, 5]{poon2011}.

In Lst.~\ref{lst:core-semantics-sum-product}, the denotation
$\Denotv{P}S$ of $S\in\SPE$ is a map from $e \in \dom{Event}$ to its
probability under the $n$-dimensional distribution defined by $S$,
where $n \defas \abs{\domfunc{scope}\, S}$ is the number of variables
in $S$.
A terminal node $\token{Leaf}\sexpr{x\, d\, \sigma}$ is comprised of a
$\dom{Var}$ $x$, $\dom{Distribution}$ $d$, and $\dom{Environment}$ $\sigma$
that maps other variables to a $\dom{Transform}$ of $x$,
e.g., $\mathtt{Z} \mapsto \scall{Poly}{\scall{Root}{\scall{Id}{\mathtt{X}}\; 2}\; [11, 5]}$.
%
% For example, in Fig~\ref{subfig:poly-invert-prior-spe}, the environments at the
% leaves in the left and right subtrees are:
% \begin{align}
% \sigma_{\rm left} &= \set{
%     \mathtt{X} \mapsto \token{Id}\sexpr{\mathtt{X}},\;
%     \mathtt{Z} \mapsto \token{Poly}\sexpr{\token{Id}\sexpr{\mathtt{X}}\; [0, 6, 1, -1]}}
% \\
% \sigma_{\rm right} &= \set{
%     \mathtt{X} \mapsto \scall{Id}{\mathtt{X}},\;
%     \mathtt{Z} \mapsto \scall{Poly}{\scall{Root}{\scall{Id}{\mathtt{X}}\; 2}\; [11, 5]}}
% \end{align}

When assessing the probability of $e$ at a $\dom{Leaf}$,
$\domfunc{subsenv}$ (Lst.~\ref{lst:core-semantics-auxfn-subsenv} in Appx.~\ref{appx:syntax-core-calculus}) rewrites $e$
as an $\dom{Event}$ $e'$ on one
variable $x$, so that the probability of $\dom{Outcomes}$ that satisfy
$e$ is exactly $\Denotv{D}{d}(\Denotv{E}{e'}x)$.
The $\domfunc{scope}$ function
(Lst.~\ref{lst:core-semantics-auxfn-scope} in Appx.~\ref{appx:syntax-core-calculus})
returns the list of variables in $S$.
For a $\dom{Sum}$, the probability of $e$ is a weighted average
of the probabilities under each subexpression.
For a $\dom{Product}$, the semantics are defined in terms of
$(\domfunc{dnf}\, e)$ (Lst.~\ref{lst:dnf-appx} in Appx.~\ref{appx:aux-fn}),
leveraging inclusion-exclusion.

In Lst.~\ref{lst:core-semantics-sum-product-density}, the denotation
$\Denot[\mathbb{P}_0]{S}$ defines the density semantics of $\SPE$,
used for measure zero events such as $\set{X=3, Y=\pi, Z=\dquote{foo}}$
under a mixed-type base measure.
These semantics, which define the density as a pair, adapt
``lexicographic likelihood-weighting'', an approximate inference
algorithm for discrete-continuous Bayes Nets~\citep{wu2018}, to exact
inference using ``lexicographic enumeration'' for $\SPE$.

%!TEX root=./paper.tex

\section{Conditioning Sum-Product Expressions}
\label{sec:condition}

We next present the main theoretical result for exact inference on
probability distributions defined by an expression $S \in \SPE$ and describe the
inference algorithm for conditioning on an $\dom{Event}$
(Lst.~\ref{lst:core-semantics-events}) in the core calculus,
which includes transformations and predicates with set-valued constraints.
\begin{restatable}[Closure under conditioning]{theorem}{condclosed}
\label{thm:closure}
Let $S\,{\in}\,\SPE$ and $e\,{\in}\,\dom{Event}$ be given, where
$\Denotv{P}{S}e\,{>}\,0$.
There exists an algorithm which, given $S$ and $e$, returns
$S'\,{\in}\,\SPE$ such that, for all $e'\,{\in}\,\dom{Event}$, the probability
of $e'$ according to $S'$ is equal to the conditional probability
of $e'$ given $e$ according to $S$, i.e.,
\begin{align}
\Denotv{P}{S'}e' \equiv
  \Denotv{P}{S}(e' \mid e) \defas
  \frac{\Denotv{P}{S}({e \sqcap e'})}{\Denotv{P}{S}e}.
\label{eq:sensize}
\end{align}
\end{restatable}
Thm.~\ref{thm:closure} is a structural conjugacy
property~\citep{diaconis1979} for the family of probability
distributions defined by the $\SPE$ domain, where both the prior and
posterior are identified by elements of $\SPE$.
We establish Eq.~\eqref{eq:sensize} constructively, by
describing a new algorithm
$\domfunc{condition}: \SPE \to \dom{Event} \to \SPE$ that satisfies
\begin{align}
\Denotv{P}{(\domfunc{condition}\, S\, e)}e' = \Denotv{P}{S}(e' \mid e)
  \label{eq:zemstvo}
\end{align}
for all $e, e' \in \dom{Event}$ with $\Denotv{P}{S}e > 0$.
Refer to Appx.~\ref{appx:condition} for the proof.
Fig.~\ref{fig:hyperrectangle} shows a conceptual example of how
$\domfunc{condition}$ works, where the prior distribution is a
$\dom{Product}$ $S$ and the conditioned distribution is a
$\dom{Sum}$-of-$\dom{Product}$ $S$'.
Fig.~\ref{fig:poly-invert-main} shows an example of the closure
property when the expression has transformed variables
(details in Appx.~\ref{appx:transforms}).

\begin{remark}
\label{remark:condition-measure-zero}
Thm.~\ref{thm:closure} refers to a positive probability $\dom{Event}$
$e$. As with sum-product networks, $\SPE$ is also
closed under conditioning on a $\dom{Conjunction}$ of possibly measure zero
equality constraints on non-transformed variables, which appear in many
PPL interfaces~\citep{saad2016,towner2019,molina2020}.
Appx.~\ref{subsec:condition-spe-equality} presents the
$\domfunc{condition}_0$ algorithm for inference on such events,
leveraging the generalized mixed-type density semantics in
Lst.~\ref{lst:core-semantics-sum-product-density}.
\end{remark}

The next result, Thm.~\ref{thm:condition-linear}, states a sufficient
requirement for inference using $(\domfunc{condition}\, S\, e)$ to
scale linearly in the size of $S$, which holds for both zero and
positive measure events.
\begin{restatable}{theorem}{condlinear}
\label{thm:condition-linear}
The runtime of $(\domfunc{condition}\, S\, e)$ scales linearly in the
number of nodes in the graph representing $S$ whenever $e$ is a single $\dom{Conjunction}$
$\sexpr{t_1\,\token{in}\,v_1} \sqcap \dots \sqcap \sexpr{t_m\,\token{in}\,v_m}$
of $\dom{Containment}$ constraints on non-transformed variables.
\end{restatable}

\input{figures/hyperrectangle}

%!TEX root=./paper.tex

\section{Translating Probabilistic Programs to Sum-Product Expressions}
\label{sec:translation}

%!TEX root = ../paper.tex

\begin{listing}[t]
\footnotesize
\begin{sublisting}[b]{\linewidth}
\FrameSep1pt
\begin{framed}
\begin{align*}
&x\in\dom{Var};\;
  y\,{\in}\,\dom{ArrayVar};\;
  n\,{\in}\,\dom{Natural};\;
  b\,{\in}\,\dom{Boolean};\;
  r\,{\in}\,\dom{Real};\;
  s\,{\in}\,\dom{String}; \\
&o_{\rm arith} \,{\in}\,\set{\texttt{+}, \texttt{-}, \texttt{*}, \texttt{/}, \texttt{**}};\quad
  o_{\rm bool}\,{\in}\,\set{\kw{and}, \kw{or}};\quad
  o_{\rm neg}\,{\in}\,\set{\kw{not}};\\
&o_{\rm rel}\,{\in}\,\set{\texttt{<=}, \texttt{<}, \texttt{>}, \texttt{>=}, \texttt{==}, \texttt{in}};\quad
D          \,{\in}\,\set{\token{normal}, \token{poisson}, \token{choice}, \dots}; \\
&E \begin{aligned}[t]&\in \dom{Expr}
  \defas x \gor n \gor b \gor r \gor s \gor y\bracktt{E}
  \gor D\sexpr{E^*}
  \gor \sexpr{E_1, \dots, E_m} \\
  &\gor E_1\, o_{\rm arith}\, E_2
  \gor o_{\rm neg}\, E
  \gor E_1\, o_{\rm bool}\, E_2
  \gor E_1\, o_{\rm rel}\, E_2
\end{aligned} \\
&C \begin{aligned}[t]&\in \dom{Command}
  \defas x\, \texttt{=}\, E
  \gor y\bracktt{E_1}\,\texttt{=}\,E_2
  \gor x\,\texttt{\textasciitilde}\,E
  \gor y\bracktt{E_1}\,\texttt{\textasciitilde}\,E_2
  \gor y\,\texttt{=}\,\kw{array}\sexpr{E} \\
  &\gor \kw{skip}
  \gor C_1\texttt{;} C_2
  \gor \kw{if}\; E\; \settt{C_1}\; \kw{else}\; \settt{C_2}
  \gor \kw{condition}\sexpr{E} \\
  &\gor \kw{for}\; x\;\kw{in}\; \kw{range}\sexpr{E_1, E_2}\;\settt{C}
  \gor \kw{switch}\; x_1\; \kw{cases}\; \sexpr{x_2 \; \kw{in}\; E}\; \settt{C}
\end{aligned}
\end{align*}
\end{framed}
\end{sublisting}%
\captionsetup{aboveskip=2pt, belowskip=0pt}
\caption{Source syntax of $\sppl{}$.}
\label{lst:sppl-syntax}
% \hrule\smallskip
\medskip
%!TEX root =  ../paper.tex
% \begin{listing}[t]
% \hrule
\footnotesize
\staterule{Sample}
{E \Downarrow d; \hfill \mbox{where } x \not\in \domfunc{scope}\, S}
{\langle x\,\texttt{\textasciitilde}\,E, S \rangle \translate S \otimes (x\; d\; \set{x\mapsto\token{Id}\sexpr{x}})}

\staterule{Transform-Leaf}
{E \Downarrow t; \quad \textrm{where}\ \domfunc{vars}\, t\,{\in}\,\mathrm{dom}(\sigma), x \not\in \mathrm{dom}(\sigma)}
{\langle x\,\texttt{=}\,E, \scall{Leaf}{x'\, d\, \sigma} \rangle \translate \scall{Leaf}{x'\; d\; (\sigma \cup \set{x\mapsto t})}}

\staterule{Transform-Sum}
{E \Downarrow t, \forall_{1\le i\le m}. \langle x = E, S_i\rangle \translate S'_i}
{\langle x = E, \oplus_{i=1}^{m}\sexpr{S_i\, w_i} \translate \oplus_{i=1}^{m}\sexpr{S'_i\, w_i}}

\staterule{Transform-Prod}
{E \Downarrow t, \langle x = E, S_j \rangle \translate S_j';
\quad \mbox{where } j \defas \min\set{i {\mid} (\domfunc{vars}\, E)\,{\in}\,\domfunc{scope}\; S_i} > 0}
{\langle x\,\texttt{=}\,E, \otimes_{i=1}^{m} S_i \rangle \translate \otimes_{i=1,i\ne j}^{m} S_i \otimes S'_j}

% \staterule{Condition}
% {E \Downarrow e
%   % Phantom for alignment.
%   \phantom{\begin{aligned}[b] &a \\[-4pt] &b \end{aligned}}}
% {\langle \kw{condition}\sexpr{E}, S \rangle \translate \domfunc{condition}\, S\, e}
% \qquad

\staterule{IfElse}
{E \Downarrow e,
    \langle C_1, \domfunc{condition}\, S\, e \rangle \translate S_1,
    \langle C_2, \domfunc{condition}\, S\, (\domfunc{negate}\, e) \rangle \translate S_2
    }
{\langle\kw{if}\; E\; \settt{C_1}\; \kw{else}\; \settt{C_2}, S \rangle
  \translate \sexpr{S_1\, \Denotv{P}{S}e} \oplus \sexpr{S_2\, (1-\Denotv{P}{S}e)}}

% \staterule{IfElse}
% {E \Downarrow e,
%   \begin{aligned}[t]
%     &\langle C_1, \domfunc{condition}\, S\, e \rangle \translate S_1, \\
%     &\langle C_2, \domfunc{condition}\, S\, (\domfunc{negate}\, e) \rangle \translate S_2;
%     \quad \mbox{where } w \defas \Denotv{P}{S}e > 0
%     \end{aligned}
%     }
% {\langle\kw{if}\; E\; \kw{then}\; \settt{C_1}\; \kw{else}\; \settt{C_2}, S \rangle
%   \translate \sexpr{S_1\, w} \oplus \sexpr{S_2\, (1-w)}}

% \staterule{Sequence}
% {\langle C_1, S \rangle \translate S_1, \langle C_2, S_1 \rangle \translate S'}
% {\langle C_1\texttt{;} C_2, S \rangle \translate S'}

% \staterule{For-Exit}
% {E_1 \Downarrow n_1, E_2 \Downarrow n_2;
% \hfill \mbox{where } n_2 \le n_1}
% {\begin{aligned}[t]
%   \langle \kw{for}\;x\;\kw{in}\;\kw{range}\sexpr{E_1, E_2}\;\settt{C}, S \rangle
%   \translate S \\[-4pt]
%   \phantom{\hfill}
%   \end{aligned}}

\staterule{For-Repeat}
{E_1 \Downarrow n_1, E_2 \Downarrow n_2;
\hfill \mbox{where } n_1 < n_2
}
{\begin{aligned}[t]
  &\langle \kw{for}\;x\;\kw{in}\;\kw{range}\sexpr{E_1, E_2}\;\settt{C}, S \rangle \\[-2pt]
  &\qquad\translate
    \langle C[x/n_1]\texttt{;} \;
    \kw{for}\;x\;\kw{in}\;\kw{range}\sexpr{n_1\;{+}\;1, E_2}\;\settt{C},
    S
    \rangle
  \end{aligned}}
\smallskip
\hrule
\captionsetup{aboveskip=2pt}
\caption{Example rules for translating an \sppl{} command $C$
(Lst.~\ref{lst:sppl-syntax}) to an element of $\SPE$
(Lst.~\ref{lst:core-semantics-sum-product}).}
\label{lst:sppl-translation}
% \end{listing}

\vspace{-.4cm}
\end{listing}

We next present a probabilistic language called
\sppl{} and show how to translate each program in the language to an
element $S\in\SPE$ that symbolically represents
(via $\Denotv{P}{S}$)
the probability distribution specified by the program.
As in Fig.~\ref{fig:system-diagram},
$S$ can then be used to answer queries about an $\dom{Event}$ $e$:
\begin{enumerate}[wide=0pt]
\item[\kw{simulate}:] Samples from the distribution defined by $\Denotv{P}{S}$;

\item[\kw{prob}:] Computes the probability of $e$,
  using $\Denotv{P}{S}e$ (Lst.~\ref{lst:core-semantics-sum-product});

\item[\kw{condition}:] Conditions on $e$,
  using $\domfunc{condition}$ (Eq.~\eqref{eq:zemstvo}).
\end{enumerate}
Lst.~\ref{lst:sppl-syntax} shows the source syntax of \sppl{}, which
contains standard constructs of an imperative language
such as \kw{array} data structures, \kw{if}-\kw{else} statements, and
bounded \kw{for} loops.
The \kw{switch}-\kw{case} macro is defined in Eq.~\eqref{eq:perquisitor}.
Random variables are defined using ``sample''
(\texttt{\textasciitilde}) and $\kw{condition}(E)$ can be used to
restrict executions to those for which
$E \in \dom{Expr}$ evaluates to $\ttrue$
as part of the prior definition.
Lst.~\ref{lst:sppl-translation} defines a relation
  $\langle C, S \rangle \translate S'$, which translates a ``current''
  $S\in\SPE$ and $C\in\dom{Command}$ into $S'\in\SPE$,
  where the initial step operates on an ``empty'' $S$.
(Lst.~\ref{lst:spe-translation} in Appx.~\ref{appx:translation-reverse} defines
a semantics-preserving inverse of $\translate$).
The $\Downarrow$ relation evaluates $E\in\dom{Expr}$ to
other domains in the core calculus using straightforward rules.
% (Lst.~\ref{lst:core-syntax-basic}--~\ref{lst:core-syntax-distributions})
% using rules similar to Eq.~\eqref{eq:villiferous}.
%
We briefly describe key rules of $\translate$:
\begin{enumerate}[align=left,leftmargin=*]
\item[\ref{Transform-Leaf}] updates the environment $\sigma$ at each $\dom{Leaf}$.
\item[\ref{Transform-Sum}]  delegates to all subexpressions.
\item[\ref{Transform-Prod}] delegates to the subexpression whose scope contains the transformed variable.
\item[\ref{For-Repeat}]  unrolls a \kw{for} loop into a $\dom{Command}$ sequence.
  % where $C[x/n_1]$ indicates syntactic replacement of $x$ with $n_1$
  % in the body $C$ of the loop.
\item[\ref{IfElse}] returns a $\dom{Sum}$ with two subexpressions, where
  the \kw{if} branch is conditioned on the test $\dom{Event}$ and the
  $\kw{else}$ branch is conditioned on the negation of the test
  $\dom{Event}$.
  This translation step involves running inference
  (using $\domfunc{condition}$, Eq.~\eqref{eq:zemstvo})
  on the current $S\in\SPE$ translated so far.
\end{enumerate}
The rule for $\kw{condition}\sexpr{E}$ (not shown) calls
$(\domfunc{condition}\, S\, e)$ (Eq.~\ref{eq:zemstvo}) or
$(\domfunc{condition}_0\, S\, e)$ (Remark~\ref{remark:condition-measure-zero}),
where $E \Downarrow e$.
To ensure \sppl{} programs translate to well-defined
element of $\SPE$, per \ref{item:definedness-env-base}--\ref{item:definedness-sum-weights}),
each program must satisfy these restrictions:
\begin{enumerate}[label=(R\arabic*), wide=0pt]
\item \label{item:restriction-prod}
   Variables $x$ in $x\,\texttt{\textasciitilde}\,E$~\ref{Sample}
   and $x\,\texttt{=}\,E$~\ref{Transform-Leaf}
   must be fresh (ensures conditions~\ref{item:definedness-env-base},
   \ref{item:definedness-env-topo}
   and~\ref{item:definedness-prod-scope}).

\item \label{item:restriction-sum} The branches in an $\kw{if}$-$\kw{else}$
  statement must define identical variables (ensures
  conditions~\ref{item:definedness-sum-scope}
  and~\ref{item:definedness-sum-weights}).

\item \label{item:restriction-transforms}
  Derived random variables are obtained via (many-to-one) \textit{univariate}
  transformations (Lst.~\ref{lst:core-semantics-transformations}).

\item \label{item:restriction-rvs}
  Parameters of distributions $D$ or $\kw{range}$ must be either
  constants or random variables with finite support.
\end{enumerate}

\ref{item:restriction-transforms} is required since the distribution
of a multivariate transform (e.g., $Z = X/Y^2$) is typically
intractable and does not factor into $\dom{Sum}$ and $\dom{Product}$
expressions.
\ref{item:restriction-rvs} is required to ensure a finite-size
$\SPE$: distributional parameters with infinite support require
integrals (uncountable support) or infinite series (countable
support), which are not in $\SPE$.
Lst.~\ref{lst:sppl-valid-invalid} shows an example of using
\kw{switch} and \kw{condition} to work around these restrictions by
discretization and truncation.

%!TEX root=../paper.tex

\begin{listing}[t]
\begin{subfigure}{\linewidth}
\captionsetup{skip=0pt}
\caption{Invalid program (translates to an infinite-sized $\SPE$)}
\begin{lstlisting}[style=sppl,frame=single,basicstyle=\ttfamily\footnotesize]
mu ~ beta(a=4, b=3, scale=7)
num_loops ~ poisson(mu)       # invalid (real integral)
for i in range(0, num_loops): # invalid (infinite series)
  [... commands ... ]
\end{lstlisting}
\end{subfigure}

\begin{subfigure}{\linewidth}
\captionsetup{skip=0pt}
\caption{Valid program (translates to a finite-sized $\SPE$)}
\begin{lstlisting}[style=sppl,frame=single,basicstyle=\ttfamily\footnotesize]
mu ~ beta(a=4, b=3, scale=7)
# binspace partitions [0,7] into 10 intervals
switch (mu) cases (m in binspace(0, 7, n=10)):
  num_loops ~ poisson(m.mean()) # discretization
condition (num_loops < 50)      # truncation
switch num_loops cases (n in range(50)):
  for i in range(0, n):
    [... commands ... ]
\end{lstlisting}
\end{subfigure}%
\captionsetup{skip=0pt}
\caption{Examples of valid and invalid \sppl{} programs.}
\label{lst:sppl-valid-invalid}
\vspace{-.5cm}
\end{listing}

\subsection{Building Compact Sum-Product Expressions}
\label{subsec:translation-opt}

As discrete Bayesian networks can be encoded as \sppl{} programs,
it is possible to write programs where exact inference is
NP-Hard~\citep{cooper1990}, which corresponds to an element of $\SPE$
that is exponentially large.
It is well known that the complexity of exact inference in Bayesian
networks is worst case exponential in the treewidth, which is the only
structural restriction that can ensure tractability~\citep{chandrasekeran2008}.
As computing treewidth is NP-Complete~\citep{arnborg1987complexity},
for fundamental theoretical reasons we cannot generally check
conditions needed for even simple \sppl{} programs, such as those that
only use if/else statements on binary variables, to translate into a
``small'' expression.

However, many models of interest contain (conditional) independence
relationship~\citep{koller2009}
that induce a compact factorization of the model into tractable subparts,
as in, e.g., Sec.~\ref{subsec:example-hmm}.
$\sppl{}$ uses several optimization techniques to improve scalability of
translation (Lst.~\ref{lst:sppl-translation}) and inference
(Eq.~\eqref{eq:zemstvo}) by automatically exploiting
independences and repeated structure, when they exist,
to build compact sum-product expressions.
%
% Figs.~\ref{fig:hmm-spe-full} and~\ref{fig:hmm-spe-factorized}
% (Sec.~\ref{subsec:example-hmm}) show an example where factorization
% and deduplication optimizations reduces the scaling from exponential
% to linear, which discuss next.

\noindentparagraph{Factorization} Using standard algebraic manipulations,
a sum-product expression can be made smaller
without changing its semantics (Lst.~\ref{lst:core-semantics-sum-product}) by
``factoring out'' common terms (Fig.~\ref{fig:memoize-factorization}),
provided that the new expression satisfies~\ref{item:definedness-env-base}--\ref{item:definedness-sum-weights}.
Factorization plays a key role in the \ref{IfElse} rule of
$\translate$: since all statements before the \kw{if}-\kw{else}
are shared by the bodies of the $\kw{if}$ and
$\kw{else}$ branches, statements outside the branch
that are independent of statements inside
the branch often produce subexpressions that can be factored out.

%!TEX root=../paper.tex

\begin{figure}[t]
\centering

\tikzset{leaf/.style={inner sep=1pt,label={[label distance=-.1cm]below:{#1}}}}
\tikzset{branch/.style={circle,draw,inner sep = 1pt}}
\begin{subfigure}[b]{.5\linewidth}
\begin{adjustbox}{max width=\linewidth}
\begin{tikzpicture}
\node[
  name=original,
  label={[name=ali, label distance=0pt]above:\scriptsize\bfseries Original}
]{
  \begin{tikzpicture}[scale=.8]
    \Tree[.\node[name=root0, branch]{$+$};
      [.\node[branch]{$\times$};
          [.$S$ \edge[fill=blue!50,roof]; {\bfseries \phantom{subtree}} ]
          $S_1$ ]
      [.\node[branch]{$\times$};
          [.$S$ \edge[fill=blue!50,roof]; {\bfseries \phantom{subtree}} ]
          $S'_1$ ]]
    \end{tikzpicture}
};

\node[
  name=factorized,
  right=3 of original.north,
  anchor=north,
  label={[name=salman, label distance=0pt]above:\scriptsize\bfseries Factorized}
] {
    \begin{tikzpicture}[scale=.8]
    \Tree[.\node[name=root1, branch]{$\times$};
        [.$S$ \edge[fill=blue!50,roof]; {\bfseries \phantom{subtree}} ]
        [.\node[branch]{$+$};
            $S_1$
            $S'_1$ ] ]
    \end{tikzpicture}
};
\draw[-latex, red, line width=1] (ali) -- (salman);
\end{tikzpicture}
\end{adjustbox}
\captionsetup{skip=0pt}
\caption{Factorization}
\label{fig:memoize-factorization}
\end{subfigure}\vrule%
\begin{subfigure}[b]{.5\linewidth}
\centering
\begin{adjustbox}{max width=\linewidth}
\begin{tikzpicture}

\node[
  name=original,
  label={[name=ali, label distance=0pt]above:\scriptsize\bfseries Original}
]{
  \begin{tikzpicture}[scale=.5]
  \Tree[.\node[branch]{$+$};
      % \edge node[auto=right]{$w$};
      [.\node[branch]{$\times$};
          $X$
          [.\node[branch]{$+$};
              $\dots$
              [.\node[branch]{$\times$};
                  $Y$
                  [.$S$ \edge[fill=blue!50,roof]; {\bfseries \phantom{subtree}} ] ]
               ] ]
      % \edge node[auto=left]{$1-w$};
      [.\node[branch]{$\times$};
          [.$S$ \edge[fill=blue!50,roof]; {\bfseries \phantom{subtree}} ]
          $X$
          $Y$ ]]
  \end{tikzpicture}
};

\node[
  name=deduplicated,
  right = 3.5 of original.north,
  anchor=north,
  label={[name=salman, label distance=0pt]above:\scriptsize\bfseries Deduplicated}
]{
  \begin{tikzpicture}[scale=.5]
  \Tree[.\node[branch]{$+$};
      % \edge node[auto=right]{$w$};
      [.\node[branch]{$\times$};
          $X$
          [.\node[branch]{$+$};
              $\dots$
              [.\node[branch]{$\times$};
                  $Y$
                  [.\node[](foo){$S$}; \edge[fill=blue!50,roof]; {\bfseries \phantom{subtree}} ] ]
               ] ]
      % \edge node[auto=left]{$1-w$};
      [.\node[branch](bar){$\times$};
          \edge[draw=none];
          [.\node{}; \edge[draw=none,roof]; {\bfseries \phantom{subtree}} ]
          $X$
          $Y$ ]]
  \draw[dashed, line width=.5] (bar.south) -- (foo);
  \end{tikzpicture}
};
\draw[-latex, red, line width=1] (ali) -- (salman);
\end{tikzpicture}
\end{adjustbox}
\captionsetup{skip=7pt}
\caption{Deduplication}
\label{fig:memoize-deduplication}
\end{subfigure}
% \hrule
\captionsetup{aboveskip=2pt, belowskip=0pt}
\caption{Exploiting independences and repeated
structure during translation of \sppl{} programs to build
compact sum-product expressions. Blue subtrees are identical
components.}
\label{fig:memoize}
\vspace{-.3cm}
\end{figure}
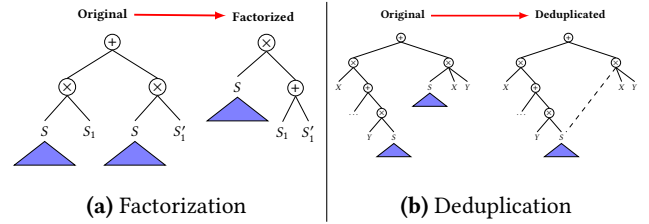
%!TEX root = ../paper.tex

\begin{table}[t]
\captionsetup{skip=0pt}
\caption{Measurements of $\SPE$ graph size with and without
the factorization and deduplication optimizations in Fig.~\ref{fig:memoize}.}
\label{table:compression}
\begin{adjustbox}{max width=\linewidth}
\begin{tabular}{@{}lrrr}
\toprule
\multirow{2}{*}{\bfseries Benchmark}
  & \multicolumn{2}{c}{\bfseries No.\ of Nodes in Translated $\SPE$}
  & \multirow{2}{*}{\bfseries \begin{tabular}{c}Data Compression \\ Ratio (unopt/opt)\end{tabular}}
  \\ \cmidrule(lr){2-3}
~ & Unoptimized
  & Optimized
  & \\ \midrule
  % & \textbf{Ratio} (unopt/opt) \\ \midrule
Hiring~\citep{albarghouthi2017} & 33 & 27 & 1.2x\\
Alarm~\citep{nori2014} & 58 & 45 & 1.3x \\
Grass~\citep{nori2014} & 130 & 59 & 2.2x \\
Noisy OR~\citep{nori2014} & 783 & 132 & 4.1x \\
Clinical Trial~\citep{nori2014} & 43761 & 4131 & 10.6x \\
Heart Disease~\citep{spiegelhalter1993} & 1041235 & 6257 & 166.4x \\
Hierarchical HMM (Sec.~\ref{subsec:example-hmm}) & 29273397577908185 & 1787 & 16381308101795x \\ \bottomrule
\end{tabular}
\end{adjustbox}
\vspace{-.25cm}
\end{table}

\noindentparagraph{Deduplication}
When a sum-product expression contains duplicate subexpressions that cannot be factored out without
violating the definedness conditions, we instead resolve duplicates into a
single physical representative.
Fig.~\ref{fig:memoize-deduplication} shows an example where the
left and right components of the original expression contain an
identical subexpression $S$ (in blue), but factorization would lead
to an invalid sum-product expression.
The optimizer represents the computation graph of this expression
using a single data structure $S$ shared by the left and right subtrees
(see also Figs.~\ref{fig:hmm-spe-full}--\ref{fig:hmm-spe-factorized}).

\noindentparagraph{Memoization} While deduplication reduces
memory overhead, memoization is used to reduce runtime overhead.
Consider either $\SPE$ in Fig.~\ref{fig:memoize-deduplication}: calling
$\domfunc{condition}$ on the $\dom{Sum}$ root will dispatch the query
to the left and right subexpressions (Lst.~\ref{lst:condition-sum}).
We cache the results of
$(\domfunc{condition}\; S\; e)$ or $\Denotv{P}{S}e$ when $S$ is
visited in the left subtree to avoid recomputing the result when
$S$ is visited again in the right subtree via a depth-first traversal.
Memoization delivers large runtime gains not only for solving
queries but also for detecting duplicates returned by
$\domfunc{condition}$ in the \ref{IfElse} translation step.

\noindentparagraph{Measurements}
Table~\ref{table:compression} shows measurements of performance gains
delivered by the factorization and deduplication optimizations on
seven benchmarks.
Compression ratios range between $1.2$x to $1.64\times 10^{13}$x and
are highest in the presence of independence or repeated structure.
The deduplication and memoization optimizations together enable fast
detection of duplicate subtrees by comparing logical memory addresses
of internal nodes in $O(1)$ time, instead of computing hash functions
that require an expensive subtree traversal.

%!TEX root = ../paper.tex

\begin{table*}
\centering
\footnotesize
\captionsetup{skip=0pt}
\caption{Runtime measurements and speedup for 15 fairness verification tasks using \sppl{},
\fairsquare{}~\citep{albarghouthi2017}, and \verifair{}~\citep{bastani2019}.}
\label{table:fairness}
% \begin{adjustbox}{max width=\linewidth}
\begin{tabular*}{\linewidth}{|l@{\extracolsep{\fill}}lrlrrrrr|}
% \toprule
\hline
% \multirow{2}{*}{\bfseries \shortstack{Decision \\ Program}}
\textbf{Decision}
   & \textbf{Population}
   & \textbf{Lines}
   & \textbf{Fairness}
   & \multicolumn{3}{c}{\bfseries Wall-Clock Runtime (seconds)}
   & \multicolumn{2}{c|}{\bfseries \sppl{} Speedup Factor}
   \\ \cline{5-7} \cline{8-9}
\textbf{Program}
   & \textbf{Model}
   & \textbf{of Code}
   & \textbf{Judgment}
   & \fairsquare{}
   & \verifair{}
   & \sppl{}
   & vs.\ \fairsquare{}
   & vs.\ \verifair{}
   \\ \hline
~                           & Independent  & 15  & Unfair & 1.4   & 16.0 & 0.01 & \textbf{140x}  & \textbf{1600x} \\
$\mathrm{DT}_4$             & Bayes Net.~1 & 25  & Unfair & 2.5   & 1.27 & 0.03 & \textbf{83x}   & \textbf{42x}   \\
~                           & Bayes Net.~2 & 29  & Unfair & 6.2   & 0.91 & 0.03 & \textbf{206x}  & \textbf{30x}   \\ \hline
~                           & Independent  & 32  & Fair   & 2.7   & 105  & 0.03 & \textbf{90x}   & \textbf{3500x} \\
$\mathrm{DT}_{14}$          & Bayes Net.~1 & 46  & Fair   & 15.5  & 152  & 0.07 & \textbf{221x}  & \textbf{2171x} \\
~                           & Bayes Net.~2 & 50  & Fair   & 70.1  & 151  & 0.08 & \textbf{876x}  & \textbf{1887x} \\ \hline
~                           & Independent  & 36  & Fair   & 4.1   & 13.6 & 0.03 & \textbf{136x}  & \textbf{453x}  \\
$\mathrm{DT}_{16}$          & Bayes Net.~1 & 49  & Unfair & 12.3  & 1.58 & 0.08 & \textbf{153x}  & \textbf{19x}   \\
~                           & Bayes Net.~2 & 53  & Unfair & 30.3  & 2.02 & 0.08 & \textbf{378x}  & \textbf{25x}   \\ \hline
~                           & Independent  & 62  & Fair   & 5.1   & 2.01 & 0.06 & \textbf{85x}   & \textbf{33x}   \\
$\mathrm{DT}^{\alpha}_{16}$ & Bayes Net.~1 & 58  & Fair   & 15.4  & 21.6 & 0.12 & \textbf{128x}  & \textbf{180x}  \\
~                           & Bayes Net.~2 & 45  & Fair   & 53.8  & 24.5 & 0.12 & \textbf{448x}  & \textbf{204x}  \\ \hline
~                           & Independent  & 93  & Fair   & 15.6  & 23.1 & 0.05 & \textbf{312x}  & \textbf{462x}  \\
$\mathrm{DT}_{44}$          & Bayes Net.~1 & 109 & Unfair & 264.1 & 19.8 & 0.09 & \textbf{2934x} & \textbf{220x}  \\
~                           & Bayes Net.~2 & 113 & Unfair & t/o   & 20.1 & 0.09 & ---            & \textbf{223x}  \\ \hline
\end{tabular*}
% \end{adjustbox}
% \vspace{-.5cm}
\end{table*}

%!TEX root = ./paper.tex

\section{Evaluation}
\label{sec:evaluations}

We implemented a prototype of \sppl{}%
\footnote{Available in supplement and online at \url{https://github.com/probcomp/sppl}.}
and evaluated its performance on
benchmark problems from the literature.
Sec.~\ref{subsec:evaluations-fairness} compares the runtime
  of verifying fairness properties of decision trees
  using \sppl{} to \fairsquare{}~\citep{albarghouthi2017}
  and \verifair{}~\citep{bastani2019}, two state-of-the-art
  fairness verification tools.
Sec.~\ref{subsec:evaluations-psi} compares the runtime of conditioning
  and querying probabilistic programs using \sppl{} to \psii{}~\citep{gehr2016}, a
  state-of-the-art tool for exact symbolic inference.
Sec.~\ref{subsec:evaluations-blog} compares the runtime of computing
  exact rare event probabilities in \sppl{} to sampling-based
  estimation in \blog{}~\citep{milch2005}.
Experiments were run on Intel i7-8665U 1.9GHz CPU with
16GB RAM.

\subsection{Fairness Benchmarks}
\label{subsec:evaluations-fairness}

Characterizing the fairness of classification algorithms is a growing application area
in machine learning~\citep{dwork2012}.
Recently,~\citet{albarghouthi2017} precisely cast the problem of
verifying the fairness of a classifier in terms of computing ratios of
conditional probabilities in a probabilistic program that specifies
the data generating and classification processes.
Briefly, if
\begin{enumerate*}[label=(\roman*)]
\item $D$ is a decision program that classifies whether applicant $A$
  should be hired;
\item $H$ is a population program that generates random applicants; and
\item $\phi_{\rm m}$ (resp.\ $\phi_{\rm q}$)
  is a predicate on $A$ that is true if the applicant is a minority
  (resp.\ qualified),
\end{enumerate*}
then $D$ is $\epsilon$-fair on $H$ (where $\epsilon>0$) if
\begin{align}
\frac
  {\Pr_{A \sim H}\left[D(A) \mid \phi_{\rm m}(A) \wedge \phi_{\rm q}(A) \right]}
  {\Pr_{A \sim H}\left[D(A) \mid \neg\phi_{\rm m}(A) \wedge \phi_{\rm q}(A) \right]}
  > 1 - \epsilon,
\label{eq:sheeplike}
\end{align}
i.e., the probability of hiring a qualified minority applicant is
$\epsilon$-close to that of hiring a qualified non-minority applicant.

In this evaluation, we compare the runtime needed by \sppl{} to obtain
a fairness judgment (Eq.~\eqref{eq:sheeplike}) for machine-learned
decision and population programs against the
\fairsquare{}~\citep{albarghouthi2017} and
\verifair{}~\citep{bastani2019} solvers.
We evaluate performance on the decision tree benchmarks from
\citet[Sec.~6.1]{albarghouthi2017}, which are one-third of the full
benchmark set.
\sppl{} cannot solve the neural network and
support-vector machine benchmarks, as they contain multivariate
transforms which do not have exact tractable solutions and are ruled
out by the \sppl{} restriction \ref{item:restriction-transforms}.
\fairsquare{} and \verifair{} can express these benchmarks as they have
approximate inference.

Table~\ref{table:fairness} shows the results.
The first column shows the decision making program (DT$_n$ means
``decision tree'' with $n$ conditionals); the second column shows the
population model used to generate data; the third column shows the
lines of code (in \sppl{}); and the fourth column shows the result of
the fairness analysis (\fairsquare{}, \verifair{}, and \sppl{} produce
the same judgment on all fifteen benchmarks).
The remaining columns show the runtime and speedup factors.
We note that \sppl{}, \verifair{}, and \fairsquare{} are all
implemented in \python{}, which allows for a fair comparison.
The measurements indicate that \sppl{} consistently obtains
probability estimates in milliseconds, whereas the two baselines
can each require over 100 seconds.
The \sppl{} speedup factors are up to 3500x (vs.\ \verifair{})
and 2934x (vs.\ \fairsquare{}).
We further observe that the runtimes in \fairsquare{} and \verifair{}
vary significantly.
For example, \verifair{} uses rejection sampling to estimate
Eq.~\eqref{eq:sheeplike} with a stopping rule to determine when
the estimate is close enough, leading to unpredictable runtime
(e.g., ${>}100$ seconds for $\mathrm{DT}_{14}$ but ${<}1$ second for
$\mathrm{DT}_4$, Bayes Net.~2).
\fairsquare{}, which uses symbolic volume computation and
hyperrectangle sampling to approximate Eq.~\eqref{eq:sheeplike},
is faster than \verifair{} in some cases
(e.g., $\mathrm{DT}_{14}$), but times out in others
($\mathrm{DT}_{44}$, Bayes Net.~2).
In contrast, \sppl{}, computes exact probabilities for
Eq.~\eqref{eq:sheeplike} and its runtime does not vary significantly
across the various benchmark problems.
%
% These comparisons illustrate a trade-off between performance and
% expressiveness:
The performance--expressiveness trade-off here is that
\sppl{} computes exact probabilities and is substantially faster on
the decision tree problems that it can express.
\fairsquare{} and \verifair{} compute approximate probabilities that
enable them to express more fairness problems, at the cost by of a
higher and less predictable runtime on the decision trees.

\subsection{Comparison to Symbolic Integration}
\label{subsec:evaluations-psi}
\input{figures/workflow}

We next compare
\sppl{} to \psii{}~\citep{gehr2016}, a state-of-the-art symbolic
inference engine, on benchmark problems that include discrete,
continuous, and transformed random variables.
\psii{} can express more inference problems than \sppl{}, as it uses
general computer algebra without having
restrictions~\ref{item:restriction-transforms} and
\ref{item:restriction-rvs} in \sppl{}.
As a result, \sppl{} can solve 14/21 benchmarks listed in
\citep[Table~1]{gehr2016}.
We first discuss key architecture novelties in \sppl{} that contribute
to its performance gains.

\noindentparagraph{Workflow Comparison}
In \sppl{}, the multi-stage modeling and inference workflow
(Fig.~\ref{fig:architecture-comparison-modular}) involves three steps
that reflect the key elements of a Bayesian inference problem:
\begin{enumerate}[label=(S\arabic*),wide=0pt]
\item \label{item:sppl-stage-1} Translating the model program into a prior
   $\dom{SP}$ $S$.

\item \label{item:sppl-stage-2} Conditioning $S$ on data to
obtain a posterior $\dom{SP}$ $S'$.

\item \label{item:sppl-stage-3} Querying $S'$, using, e.g., \kw{prob}
  or \kw{simulate}.
\end{enumerate}
An advantage of this multi-stage workflow is that multiple tasks
can be run at a given stage without rerunning previous stages.
For example, multiple datasets can be observed
in~\ref{item:sppl-stage-2} without translating the prior
expression in~\ref{item:sppl-stage-1} once per dataset; and,
similarly, multiple queries can be run in~\ref{item:sppl-stage-3}
without conditioning on data in~\ref{item:sppl-stage-2}
once per query.
In contrast, \psii{} adopts a single-stage workflow
(Fig.~\ref{fig:architecture-comparison-monolithic}), where a single
program contains the prior distribution over variables, ``observe''
(i.e., ``condition'') statements for conditioning on a dataset, and a
``return'' statement for the query.
\psii{} converts the program into a
symbolic expression for the distribution over the return value: if
this expression is ``complete'' (i.e., no unevaluated symbolic
integrals) it can be used to obtain interpretable answers (e.g.,
for plotting or tabulating); otherwise, the result is ``partial'' and is
too complex to be used for practical purposes.
A consequence of the single-stage workflow in a system like
\psii{} is that the entire solution is recomputed from scratch
on a per-dataset or per-query basis.

\noindentparagraph{Runtime Comparison}
Table~\ref{table:sppl-psi} compares the runtime of
\sppl{} and \psii{} on seven benchmarks problems:
Digit Recognition~\citep{gehr2016};
TrueSkill~\citep{laurel2020};
Clinical Trial~\citep{gehr2016};
Gamma transforms (described below);
Student Interviews~\citep{laurel2020} (two variants); and
Markov Switching (two variants, from Sec.~\ref{subsec:example-hmm});
The second column shows the distributions in each benchmark, which
include continuous, discrete, and transformed variables.
The third column shows the number of datasets on which to condition
the program.
The next three columns show the time needed to translate the program
(stage~\ref{item:sppl-stage-1}), condition the program on a dataset
(stage~\ref{item:sppl-stage-2}), and query the posterior
(stage~\ref{item:sppl-stage-3})---entries in the latter two columns
are written as $n \times t$, where $n$ is the number of datasets and
$t$ the average time per dataset.
For \psii{}: \begin{enumerate*}[label=(\roman*)]
\item modeling and observing data are a single stage, shown in
the merged gray cell; and
\item querying the posterior times out whenever the system returns a
  result with unsimplified integrals ($\ltimes$).
\end{enumerate*}
The last column shows the overall runtime for solving all $n$ tasks.

For benchmarks that both systems solve completely,
\sppl{} realizes speedups between 3x (Digit Recognition) to 3600x
(Markov Switching$_{3}$).
In addition, the measurements show the advantage of our multi-stage workflow;
for example, in TrueSkill, which uses a Poisson--Binomial
distribution, \sppl{} translation ($3.4$ seconds) is more expensive
than both conditioning on data ($0.7$ seconds) and querying ($0.1$
seconds), which highlights the benefit of amortizing the translation
cost over several datasets or queries.
In \psii{}, solving TrueSkill takes $2 \times 41.6$ seconds, but the
solution contains unsimplified integrals and is thus unusable.
The Markov Switching and Student Interviews benchmarks show that
\psii{} may not perform well in the presence of many discrete random
variables.

%!TEX root = ../paper.tex

\begin{comment}
\begin{figure}[t]
\centering
\begin{subfigure}[b]{.25\textwidth}
\captionsetup{belowskip=0pt, aboveskip=2pt}
% \begin{mdframed}
\includegraphics[width=\textwidth]{{assets/dispersion.digits.vertical.pldi}.pdf}
% \end{mdframed}
\caption*{Digit Recognition}
\end{subfigure}%
\begin{subfigure}[b]{.25\textwidth}
% \begin{mdframed}
\includegraphics[width=\textwidth]{{assets/dispersion.markov.vertical.pldi}.pdf}
% \end{mdframed}
\captionsetup{belowskip=0pt, aboveskip=2pt}
\caption*{Markov Switching}
\end{subfigure}%
\begin{subfigure}[b]{.25\textwidth}
% \begin{mdframed}
\includegraphics[width=\textwidth]{{assets/dispersion.interviews.vertical.pldi}.pdf}
% \end{mdframed}
\captionsetup{belowskip=0pt, aboveskip=2pt}
\caption*{Student Interviews}
\end{subfigure}%
\begin{subfigure}[b]{.25\textwidth}
% \begin{mdframed}
\includegraphics[width=\textwidth]{{assets/dispersion.clinical.vertical.pldi}.pdf}
% \end{mdframed}
\captionsetup{belowskip=0pt, aboveskip=2pt}
\caption*{Clinical Trial}
\end{subfigure}
\hrule

\captionsetup{aboveskip=5pt}
\caption{Distribution of end-to-end inference runtime for four
benchmark problems from Table~\ref{table:sppl-psi} using \sppl{} and
\psii{}. For each benchmark, one inference query is repeated over ten
distinct datasets (dots).}
\label{fig:dispersion}
\end{figure}
\end{comment}

\begin{table}[t]
\footnotesize
\captionsetup{skip=0pt}
\caption{Distribution of end-to-end inference runtime for four
benchmarks from Table~\ref{table:sppl-psi} using
\psii{}~\citep{gehr2016} and \sppl{}.
% For each benchmark, one query is repeated over ten datasets.
}
\label{table:dispersion}

\begin{tabularx}{\linewidth}{@{\extracolsep{\fill}}lrr}
\toprule
\multirow{2}{*}{\bfseries Benchmark} & \multicolumn{2}{c}{\bfseries Mean/Std Runtime (sec/sec)}
  \\ \cmidrule[.5pt]{2-3}
~ & PSI & \sppl{} \\ \hline
Digit Recognition & 26.5/1.3 & 15.9/0.5 \\
Markov Switching & 22.5/3.8 & 0.1/0.0 \\
Student Interviews & 539/663 & 7.8/0.2 \\
Clinical Trial & 107.3/153.2 & 12.7/0.3 \\ \bottomrule
\end{tabularx}
\vspace{-.25cm}
\end{table}

%!TEX root = ../paper.tex

\begin{table*}[t]
\captionsetup{belowskip=0pt,aboveskip=0pt}
\caption{Runtime comparison of $\psii{}$~\citep{gehr2016} and $\sppl{}$.}
\label{table:sppl-psi}

\centering

\newcommand{\cellPrior}[1]{\cellcolor{\fillPrior}{#1}}
\newcommand{\cellCondt}[1]{\cellcolor{\fillCondt}{#1}}
\newcommand{\cellQuery}[1]{\cellcolor{\fillQuery}{#1}}
\newcommand{\cellPrCdt}[1]{\cellcolor{gray!10}{#1}}

\begin{subtable}{.75\linewidth}
\begin{adjustbox}{max width=\linewidth}
\begin{tabular}{|llrlrrrr|}
\hline
\multirow{2}{*}{\bfseries Benchmark}
  & \multirow{2}{*}{\bfseries  Distribution}
  & \multirow{2}{*}{\bfseries\shortstack{Datasets ${\diamond}$}}
  & \multirow{2}{*}{\bf System}
  & \multicolumn{3}{c}{\bfseries Wall-Clock Runtime of Inference Stages}
  & \multirow{2}{*}{\bfseries \shortstack{Overall \\ Time}}
  % \\
  \\ \cline{5-7}
~
  & ~
  & ~
  & ~
  & \multicolumn{1}{c}{Translation ${\dagger}$}
  & \multicolumn{1}{c}{Conditioning ${\ddagger}$}
  & \multicolumn{1}{c}{Querying ${\star}$}
  & ~
  % \\ \hline
  \\ \hline
% DIGIT RECOGNITION
Digit
  & \multirow{2}{*}{$\mathsf{C}{\times}\mathsf{B}^{784}$}
  & \multirow{2}{*}{10}
  & \sppl{}
  & \cellPrior{6.9 sec}
  & \cellCondt{$10 \times 7.7$ sec}
  & \cellQuery{$10 \times ({<}0.01 \mbox{ sec})$}
  & 84 sec \\
Recognition
  & ~
  & ~
  & \psii{}
  & \multicolumn{2}{c}{\cellPrCdt{$10 \times 24.3$ sec}}
  & \cellQuery{$10 \times ({<}0.01 \mbox{ sec})$}
  & 244 sec \\ \hline
% TRUE SKILL
\multirow{2}{*}{TrueSkill}
  & \multirow{2}{*}{$\mathsf{P}{\times}\mathsf{Bi}^2$}
  & \multirow{2}{*}{2}
  & \sppl{}
  & \cellPrior{$3.4$ sec}
  & \cellCondt{$2 \times 0.7$ sec}
  & \cellQuery{$2 \times 0.1 \mbox{ sec}$}
  & $4.9$ sec \\
~ & ~
  & ~
  & \psii{}
  & \multicolumn{2}{c}{\cellPrCdt{$2 \times 41.60$ sec}}
  & \cellQuery{$\ltimes$}
  & $\oslash$ \\ \hline
% CLINICAL TRIAL
Clinical
  & \multirow{2}{*}{
    \shortstack[l]{
      $\mathsf{B}{\times}\mathsf{U}^{3}$\\
      ${\times}\mathsf{B}^{50}{\times}\mathsf{B}^{50}$
  }}
  & \multirow{2}{*}{10}
  & \sppl{}
  & \cellPrior{$9.5$ sec}
  & \cellCondt{$10 \times 2.2$ sec}
  & \cellQuery{$10 \times ({<}0.01 \mbox{ sec})$}
  & $31$ sec \\
Trial & ~
  & ~
  & \psii{}
  & \multicolumn{2}{c}{\cellPrCdt{$10 \times 107.3$ sec}}
  & \cellQuery{$10 \times ({<}0.01 \mbox{ sec})$}
  & 1073 sec \\ \hline
% GAMMA POLYNOMIAL
Gamma
  & \multirow{2}{*}{
    \shortstack[l]{
      $\mathsf{G}{\times}\mathsf{T}$\\
      ${\times}(\mathsf{T}{+}\mathsf{T})$
  }}
  & \multirow{2}{*}{5}
  & \sppl{}
  & \cellPrior{$0.02$ sec}
  & \cellCondt{$5 \times 0.52$ sec}
  & \cellQuery{$5 \times 0.03 \mbox{ sec}$}
  & $2.8$ sec \\
Transforms &
  & ~
  & \psii{}
  & \multicolumn{2}{c}{\cellPrCdt{$5 \times 0.68 \mbox{ sec}$; i/e}}
  & \cellQuery{$\ltimes$}
  & $\oslash$ \\ \hline
% STUDENT INTERVIEWS 2
Student
  & \multirow{2}{*}{
    \shortstack[l]{
      $\mathsf{P}{\times}\mathsf{B}^{2}{\times}\mathsf{Bi}^{4}$\\
      $\times\mathsf{(A{+}Be)}^{2}$
  }}
  & \multirow{2}{*}{10}
  & \sppl{}
  & \cellPrior{$4.0$ sec}
  & \cellCondt{$10 \times 0.7$ sec}
  & \cellQuery{$10 \times \mbox{0.2 sec}$}
  & $13.5$ sec \\
Interviews$_{2}$ & ~
  & ~
  & \psii{}
  & \multicolumn{2}{c}{\cellPrCdt{$10 \times 540 \mbox{ sec}$; h/m (35GB)}}
  & \cellcolor{\fillQuery}{$\ltimes$}
  & $\oslash$ \\ \hline
% STUDENT INTERVIEWS 3
% Student
%   & \multirow{2}{*}{
%     \shortstack[l]{
%       $\mathsf{P}{\times}\mathsf{B}^{3}{\times}\mathsf{Bi}^{6}$\\
%       $\times\mathsf{(A{+}Be)}^{3}$
%   }}
%   & \multirow{2}{*}{10}
%   & \sppl{}
%   & \cellPrior{$12.8$ sec}
%   & \cellCondt{$10 {\times} 2.1$ sec}
%   & \cellQuery{$10 {\times} \mbox{0.4 sec}$}
%   & $38$ sec \\
% Interviews$_{3}$ & ~
%   & ~
%   & \psii{}
%   & \multicolumn{2}{c}{\cellPrCdt{h/m (3GB+)}}
%   & \cellcolor{\fillQuery}{$?$}
%   & $?$ \\ \hline
% STUDENT INTERVIEWS 10
Student
  & \multirow{2}{*}{
    \shortstack[l]{
      $\mathsf{P}{\times}\mathsf{B}^{10}{\times}\mathsf{Bi}^{20}$\\
      $\times\mathsf{(A{+}Be)}^{10}$
  }}
  & \multirow{2}{*}{10}
  & \sppl{}
  & \cellPrior{$24.6$ sec}
  & \cellCondt{$10 \times 3.9$ sec}
  & \cellQuery{$10 \times 1.2 \mbox{ sec}$}
  & $75$ sec \\
Interviews$_{10}$ & ~
  & ~
  & \psii{}
  & \multicolumn{2}{c}{\cellcolor{gray!10}o/m (64GB+)}
  & \cellcolor{\fillQuery}{$\oslash$}
  & $\oslash$ \\ \hline
% MARKOV SWITCHING 3
Markov
  & \multirow{2}{*}{
    \shortstack[l]{
      $\mathsf{B}{\times}\mathsf{B}^{3}$\\
      $\times\mathsf{N}^{3}{\times}\mathsf{P}^{3}$
  }}
  & \multirow{2}{*}{10}
  & \sppl{}
  & \cellPrior{$0.05$ sec}
  & \cellCondt{$10 \times ({<}0.01 \mbox{ sec})$}
  & \cellQuery{$10 \times ({<}0.01 \mbox{ sec})$}
  & $0.5$ sec \\
Switching$_{3}$ &
  & ~
  & \psii{}
  & \multicolumn{2}{c}{\cellPrCdt{$10 \times 182.9$ sec}}
  & \cellQuery{$10 \times ({<}0.01 \mbox{ sec})$}
  & $1829$ sec \\ \hline
% MARKOV SWITCHING 100
Markov
  & \multirow{2}{*}{
    \shortstack[l]{
      $\mathsf{B}{\times}\mathsf{B}^{100}$\\
      $\times\mathsf{N}^{100}{\times}\mathsf{P}^{100}$
  }}
  & \multirow{2}{*}{10}
  & \sppl{}
  & \cellPrior{$4.1$ sec}
  & \cellCondt{$10 \times 6.5$ sec}
  & \cellQuery{$10 \times 0.5$ sec}
  & $74$ sec \\
Switching$_{100}$ &
  & ~
  & \psii{}
  & \multicolumn{2}{c}{\cellPrCdt{o/m (64GB+)}}
  & \cellcolor{\fillQuery}{$\oslash$}
  & $\oslash$ \\ \hline
\end{tabular}
\end{adjustbox}
\end{subtable}\hfill
\begin{subtable}{.24\linewidth}
\captionsetup{aboveskip=5pt,belowskip=0pt,textfont={footnotesize,bf}}
\caption*{Legend}
% \begin{framed}
\scriptsize
\begin{tabular}{@{}l@{\extracolsep{4pt}}l@{}l}
  \textsf{A}:  Atomic
  & \textsf{B}:  Bernoulli
  & \textsf{Be}: Beta \\
  \textsf{Bi}: Binomial
  & \textsf{C}:  Categorical \\
  \textsf{N}:  Normal
  & \textsf{G}:  Gamma
  & \textsf{P}:  Poisson \\
  \textsf{T}:  Transform
  & \textsf{U}:  Uniform \span
\end{tabular}
\\[5pt]
$\diamond$ Number of distinct datasets on which to condition the benchmark
  program.
\\[5pt]
${\dagger,\ddagger}$ Runtime of first two phases
  in Fig.~\ref{fig:architecture-comparison};
  $\psii{}$ implements these phases in a single computation.
\\[5pt]
${\star}$ Runtime of final phase
  in Fig.~\ref{fig:architecture-comparison};
  same query used for all datasets of a given benchmark program.
\\[5pt]
\begin{tabular}{@{}l@{\extracolsep{5pt}}l@{}}
  h/m &High-Memory \\
  o/m &Out-of-Memory \\
  i/e &Integration Error \\
  $\ltimes$ &Unsimplified Symbolic Integrals \\
  $\oslash$ &No Value
\end{tabular}
% \end{framed}
\end{subtable}
% \vspace{-.5cm}
\vspace{-.3cm}
\end{table*}

The Gamma Transform benchmark tests the robustness of many-to-one
transforms of random variables (Lst.~\ref{lst:core-semantics-transformations}),
where $X\,{\sim}\,\mathrm{Gamma}(3,1)$; $Y\,{=}\,1/\exp{X^2}$ if $X <1$ and $Y\,{=}\,1/\ln{X}$
otherwise; and $Z\,{=}\,-Y^3 + Y^2 + 6Y$. Each of the $n=5$ datasets
specifies a different constraint $\phi(Z)$ and a query about
the posterior $Y\,{\mid}\,\phi(Z)$, which needs to compute and
integrate out $X\,{\mid}\,\phi(Z)$.
\psii{} reports that there is an error in its answer for all
five datasets, whereas \sppl{}, using the symbolic transform
solver from Appx.~\ref{appx:transforms-preimage}, solves all five
problems effectively.

Table~\ref{table:dispersion} compares the runtime variance
using \sppl{} and \psii{} for four of the benchmarks in
Table~\ref{table:sppl-psi}, repeating one query over 10 datasets.
In all benchmarks, the \sppl{} variance is lower than that
of \psii{}, with a maximum standard deviation $\sigma=0.5~\mathrm{sec}$.
In contrast, the spread of \psii{} runtime is high for
  Student Interviews ($\sigma=540~\mathrm{sec}$,
    range 64--1890 $\mbox{sec}$)
  and Clinical Trial ($\sigma=153~\mathrm{sec}$,
    range 2.75--470 $\mbox{sec}$).
In \psii{}, the symbolic analyses are sensitive to the numeric values
in the dataset, leading to unpredictable runtime across different
datasets, even for a fixed query pattern.
In \sppl{}, the runtime depends only on the query
pattern not the observed data and therefore behaves predictably across
different datasets.

As with the fairness benchmarks in
Sec.~\ref{subsec:evaluations-fairness},
\psii{} trades off expressiveness with efficacy on tractable problems,
and our measurements show that its runtime and memory do not scale well or
are unpredictable on benchmarks that \sppl{} solves very efficiently.
Moreover, the evaluations show that \psii{} can return unusable inference
results to the user and that it needs to recompute entire symbolic solutions
from scratch for each new dataset or query, whereas \sppl{} is less
expressive than \psii{} but carries neither of these limitations.

\subsection{Comparison to Sampling-Based Estimates}
\label{subsec:evaluations-blog}

We next compare the runtime and accuracy of estimating
probabilities of rare events in a canonical Bayesian network~\citep{koller2009}
using \sppl{} and \blog{}~\citep{milch2005}.
As discussed by \citet[Sec~12.13]{koller2009}, rare events
are the rule, not the exception, in many applications, as the
probability of a predicate $\phi(X)$ decreases exponentially with the
number of observed variables in $X$.
Small estimation errors can magnify substantially when, e.g.,
taking ratios of probabilities.

In Fig.~\ref{fig:blog}, each subplot shows the runtime and probability
estimates for a low-probability predicate $\phi$.
In BLOG, the rejection sampler estimates the
probability of $\phi$ by computing the fraction of times it holds in a
size $n$ i.i.d.\ random sample from the prior.
The horizontal red line shows the ``ground truth'' probability.
The x marker shows the runtime needed by \sppl{} to (exactly)
compute the probability and the dots show the estimates from \blog{}
with increasing runtime (i.e., more samples $n$).
\sppl{} consistently returns an exact answer in less than 2ms.
The accuracy of \blog{} estimates improve as the runtime increases: by
the strong law of large numbers, these estimates converge to the true
value, but the fluctuations for any single run can be large (the
standard error decays as $1/\sqrt{n}$).
Each ``jump'' corresponds to a new sample $X^{(j)}$ that satisfies
$\phi(X^{(j)})$, which increases the estimate.
Without ground truth, it is hard to predict how much
computation is needed for \blog{} to obtain accurate results:
estimates for predicates with $\log{p}=-12.73$ and $\log{p}=-17.32$
did not converge within the allotted time, while those for
$\log{p} = -14.48$ converged after 180 seconds.

%!TEX root = ../paper.tex

\begin{figure*}[t]
\includegraphics[width=\linewidth]{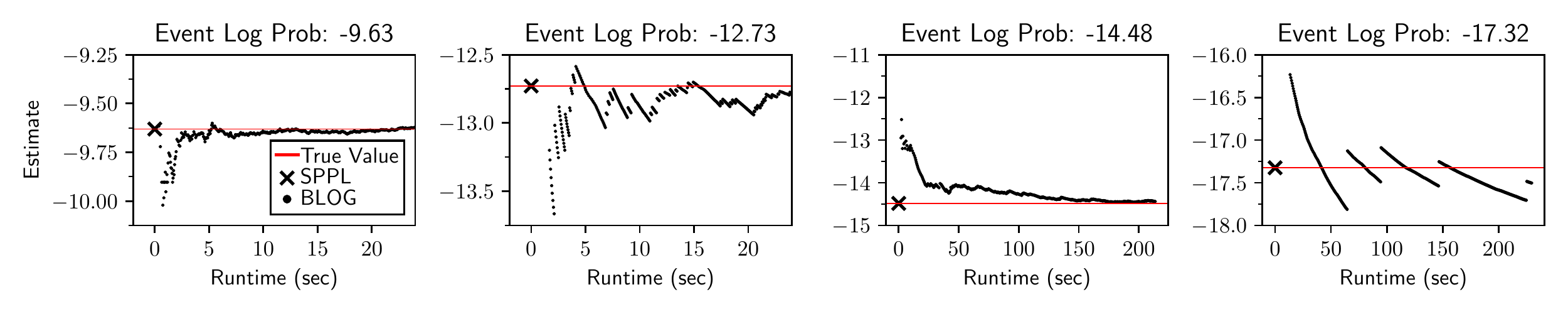}
\captionsetup{skip=-2pt,belowskip=-10pt}
\caption{Runtime comparison for computing probabilities using
exact inference in \sppl{} and rejection sampling in \blog{}.
%
% Each panel shows a specific event.
% %
% As the probability of the event decreases, the runtime needed to
% obtain an accurate estimate using sampling-based inference in
% \blog{} increases, whereas \sppl{} delivers exact answers in
% milliseconds for all events.
}
\label{fig:blog}
% \hrule
\end{figure*}

%!TEX root=./paper.tex

\vspace{-.25cm}

\section{Related Work}
\label{sec:related}

\sppl{} is distinguished by being the first system to deliver exact
symbolic inference by translating probabilistic programs to
sum-product expressions, which extend and generalize sum-product networks.
We briefly discuss related approaches.

\noindent{\bfseries Symbolic Integration}
Several systems deliver exact inference by translating
a probabilistic program and observed dataset into a symbolic
expression whose solution is the answer to the
query~\citep{bhat2013,narayanan2016,gehr2016,shan2016,zhang2019}.
Our approach to exact inference, which uses sum-product expressions
instead of general computer algebra, enables effective performance on
a range of models and queries, primarily at the expense of the
expressiveness of the language on continuous priors.
The state-of-the-art solver, \psii{}~\citep{gehr2016}, can
effectively solve many inference problems that \sppl{} cannot express
due to restrictions~\ref{item:restriction-prod}--\ref{item:restriction-rvs},
including higher-order programs~\citep{gehr2020}.
However, comparisons on benchmarks that \sppl{} targets
(Sec.~\ref{subsec:evaluations-psi}) find \psii{} has less scalable and
higher variance runtime, and can return partial results with
unsimplified symbolic integrals.
In contrast, \sppl{} exploits conditional independences, when they
exist, to improve scalability (Sec.~\ref{subsec:translation-opt}) and
delivers complete, usable answers to users.
Moreover, \sppl{}'s multi-stage workflow
(Fig.~\ref{fig:architecture-comparison}) allows expensive computations
such as translation and conditioning to be amortized
over multiple datasets or queries, whereas \psii{} recomputes the
symbolic solution from scratch each time.
Hakaru~\citep{narayanan2016} is a symbolic solver that delivers
exact inference in a multi-stage workflow based on
program transformations, and can disintegrate against a variety of
base measures~\citep{narayanan2020}.
This paper compares against PSI because the reference Hakaru implementation
crashes or delivers incorrect or partial results on several
benchmark problems~\citep[Table 1]{gehr2016}, and, as mentioned by the
developers, does not support constructs such as arrays needed to
support dozens or hundreds of observations.

\noindent{\bfseries Symbolic Execution and Volume Computation}:
Previous work has addressed the problem of computing the probability
of a predicate by integrating a distribution defined by a
program~\citep{geldenhuys2012,sankaranarayanan2013,toronto2015,albarghouthi2017}.
For example, \citet{geldenhuys2012} present a probabilistic symbolic
execution technique that uses model counting to compute path
probabilities, assuming that all program variables are discrete and
uniformly distributed.
While \sppl{} can model a variety of distributions, due to
restriction~\ref{item:restriction-transforms} it only supports
predicates that specify rectangular regions, whereas several of the
aforementioned systems can (approximately) handle non-rectangular regions.
More specifically, predicates in \sppl{} may include combinations of nonlinear
transforms, each of a single variable, which are solved into linear
expressions that specify unions of disjoint hyperrectangles
(Appx.~\ref{appx:transforms-preimage}).
Table~\ref{table:fairness} shows that \sppl{} delivers substantial
speedup on the hyperrectangular regions specified by the important
class of decision trees, which are widely used in interpretable
machine learning applications.

\noindent{\bfseries Sum-Product Networks}:
The \spflow{} library~\citep{molina2020} is an object-oriented
``graphical model toolkit'' in Python for constructing and querying
sum-product networks.
\sppl{} leverages a new and more general sum-product representation
(Lst.~\ref{lst:core-semantics}) and solves probability and
conditioning queries that are not supported by \spflow{}
(Thm.~\ref{thm:closure}), which include mixed random variables,
numeric transforms, and logical predicates with set-valued constraints.
In addition, we introduce a novel translation strategy
(Sec.~\ref{sec:translation}) that allows users to specify models as
generative code in a PPL (using e.g., variables, arrays, arithmetic
and logical expressions, loops, branches) without having to manually
manipulate low-level data structures.
``Factored sum-product networks''~\citep{stuhlmuller2012} have been
used as intermediate representations for converting a probabilistic
program and any functional interpreter into a system of equations
whose solution is the marginal probability of the program's return
value.
These algorithms handle recursive procedures and leverage dynamic
programming, but only apply to discrete variables, cannot handle
transforms, and require solving fixed-points.
Moreover, they have not been quantitatively evaluated on PPL benchmark problems.

\noindent{\bfseries Weighted Model Counting/Integration}:
A common approach to probabilistic inference is using algorithmic
reductions from probabilistic programs to weighted-model counting
(WMC) or integration (WMI) via knowledge
compilation~\citep{darwiche2002,chavira2008,fierens2011,vlasselaer2015,belle2015}.
For example, Symbo~\citep{pedro2019} leverages WMI for exact inference
in hybrid domains, using sentinel decision diagrams as the
representation and the PSI solver to symbolically integrate over
continuous variables.
Dice~\citep{holtzen2020} leverages WMC for scaling exact inference in
discrete probabilistic programs and uses binary decision diagram
representations that automatically exploit program structure to
factorize inference.
The representations in Dice enable substantial computation reuse for
querying and/or conditioning, such as computing ``all-marginal''
probabilities by reusing the same compiled representation multiple
times.
\sppl{} also leverages factorization and computation reuse, but uses a
different representation based on sum-product expressions that handle
additional computations such as numeric transforms and continuous and
mixed-type random variables.

\noindent{\bfseries Probabilistic Program Synthesis}:
Existing PPL synthesis systems for tabular
data~\citep{chasins2017,saad2019bayesian} produce programs in
languages that are subsets of \sppl{}, which enable
automatic synthesis of full \sppl{} programs from data.
\sppl{} can also unify and extend custom PPL query engines used in
these systems for tasks such as similarity search and dependence
detection~\citep{saad2016probabilistic,saad2017dependencies,saad2017search}.
It may also be fruitful to use structure discovery methods for
time series~\citep{saad2018trcrpm,ahmed2008dynamic} or relational
data~\citep{kemp2006} to synthesize \sppl{} programs for these
domains.

%!TEX root=./paper.tex

\section{Conclusion}
\label{sec:conclusion}

We have presented \sppl{}, a new system that automatically delivers
exact answers to a range of probabilistic inference queries.
A key insight in \sppl{} is to impose restrictions on probabilistic
programs that enable them to be translated to sum-product expressions, which are
highly effective representations for inference.
Our evaluation highlights the efficacy of \sppl{} on inference tasks
in the literature and underscores the importance of key design
decisions, including the multi-stage inference workflow and techniques
used to build compact expressions by exploiting probabilistic
structure.
In addition to its efficacy as a standalone language, we further
anticipate that \sppl{} could be useful as an embedded domain-specific
language within more expressive PPLs, combining the benefits of exact
and approximate inference.

\bibliography{paper}

\clearpage
\appendix
%!TEX root = ./paper.tex

\section{Syntax of Core Calculus}
\label{appx:syntax-core-calculus}

The metalanguage in this paper follows that of~\citet[Appx.~A]{turbak2008}.
For completeness, Lst.~\ref{lst:core-syntax} shows the syntax of
the core calculus whose semantics are given in
Lst.~\ref{lst:core-semantics} from the main text.
Prop.~\ref{prop:inverse-cdf} below establishes that distributions specified
by $\dom{DistInt}$ and $\dom{DistReal}$ (Lst.~\ref{lst:core-semantics-distributions})
with $\dom{CDF}$ $F$
can be sampled using a variant of the integral probability transform.

\begin{proposition}
\label{prop:inverse-cdf}
Let $F$ be a $\dom{CDF}$ and $r_1$, $r_2$ real numbers
such that $F(r_1)<F(r_2)$.
Let $U\sim\dist{Uniform}(F(r_1), F(r_2))$ and
define the random variable
$X \defas F^{-1}(U)$. Then for all real numbers $r$,
\begin{align}
\tilde{F}(r) \defas \Pr[X \le r] = \begin{cases}
0 & \mbox{if } r < r_1 \\
\displaystyle \frac{F(r)-F(r_1)}{F(r_2)-F(r_1)} & \mbox{if } r_1 \le r \le r_2 \\
1 & \mbox{if } r_2 < r
\end{cases}
\end{align}
\end{proposition}

\begin{proof}
Immediate from $\Pr[X \le r] = \Pr[ U \le F(r)]$ and the uniformity of $U$
on $[r_1, r_2]$.
\end{proof}

%!TEX root = ./paper.tex

\section{Definitions of Auxiliary Functions}
\label{appx:aux-fn}

Sec.~\ref{sec:core} refers to the following operations on the
$\dom{Outcomes}$ domain:
\begin{align}
\domfunc{union} &: \dom{Outcomes}^* \to \dom{Outcomes}, \\
\domfunc{intersection} &: \dom{Outcomes}^* \to \dom{Outcomes}, \\
\domfunc{complement} &: \dom{Outcomes} \to \dom{Outcomes}.
\end{align}
Any implementation satisfies the following invariants:
\begin{align}
&\begin{aligned}[m]
  &v_1 \amalg \dots \amalg v_m = \domfunc{union}\, v^* \\
  &\quad \iff \forall i \ne j. \domfunc{intersection}\, v_i\, v_j = \varnothing,
  \label{eq:invariant-union-union}
\end{aligned} \\
&\begin{aligned}[m]
  &v_1 \amalg \dots \amalg v_m = \domfunc{intersection}\, v^* \\
  &\quad \iff \forall i \ne j. \domfunc{intersection}\, v_i\, v_j = \varnothing,
  \label{eq:invariant-union-intersection}
\end{aligned} \\
&\begin{aligned}
  &v_1 \amalg \dots \amalg v_m = \domfunc{complement}\, v \\
  &\quad \iff \forall i \ne j. \domfunc{intersection}\, v_i\, v_j = \varnothing.
  \label{eq:invariant-union-complement}
\end{aligned}
\end{align}
Lst.~\ref{lst:complement} shows an implementation of the
$\domfunc{complement}$ function, which operates separately on the $\dom{Real}$
and $\dom{String}$ components; $\domfunc{union}$ and $\domfunc{intersection}$
are implemented similarly.
Lst.~\ref{lst:vars} shows the $\domfunc{vars}$ function, which returns the
variables in a $\dom{Transform}$ or $\dom{Event}$ expression.
Lst.~\ref{lst:negate} shows the $\domfunc{negate}$ function, which returns
the logical negation of an $\dom{Event}$.

%!TEX root = ./paper.tex

\section{Transforms of Random Variables}
\label{appx:transforms}

This appendix describes the $\dom{Transform}$ domain in the core
calculus (expanding Lst.~\ref{lst:core-semantics-transformations}),
which is used to express numerical transformations of real random
variables.

\subsection{Valuation of Transforms}
\label{appx:transforms-valuation}
Lst.~\ref{lst:transform} shows the valuation function $\valfunc{T}$
which defines each $t$ as a $\dom{Real}$ function on $\dom{Real}$.
Each real function $\Denot{T}{t}$ is defined on an input $r'$ if and
only if $\left(\inj[r']{\dom{Real}}{\dom{Outcome}}\right) \in
(\domfunc{domainof}\, t)$.
Lst.~\ref{lst:domainof} shows the implementation of
$\domfunc{domainof}$.

\subsection{Preimage Computation}
\label{appx:transforms-preimage}

Lst.~\ref{lst:preimage} shows the algorithm that implements
\begin{align}
\domfunc{preimg}: \dom{Transform} \to \dom{Outcomes} \to \dom{Outcomes},
\end{align}
which, as discussed in Sec.~\ref{sec:core} of the main text, satisfies
% (Eqs.~\eqref{eq:dermogastric-1} and ~\eqref{eq:dermogastric-2} from
% the main text):
{\small\begin{align*}
(\inj[r]{\dom{Real}}{\dom{Outcome}}) \in \Denotv{V}{\domfunc{preimg}\; t\; v}
  &\iff \Denotv{T}{t}(r) \in \Denotv{V}{v},
  % \label{eq:dermogastric-1-appx}
  \\
(\inj[s]{\dom{String}}{\dom{Outcome}}) \in \Denotv{V}{\domfunc{preimg}\; t\; v}
  &\iff (t \in \dom{Identity}) \wedge (s \in \Denotv{V}v).
  % \label{eq:dermogastric-2-appx}
\end{align*}}
The implementation of $\domfunc{preimg}$ uses several helper
functions:
\begin{enumerate}[wide, leftmargin=*]
\item[(Lst.~\ref{lst:finv})] $\domfunc{finv}$: computes the preimage of
each $t \in \dom{Transform}$ at a single $\dom{Real}$.

\item[(Lst.~\ref{lst:polyLim})] $\domfunc{polyLim}$: computes
the limits of a polynomial at the infinites.

\item[(Lst.~\ref{lst:polySolve})] $\domfunc{polySolve}$: computes the
set of values at which a polynomial is equal to a given value
(possibly positive or negative infinity).

\item[(Lst.~\ref{lst:polyLte})] $\domfunc{polyLte}$: computes the
set of values at which a polynomial is less than or equal a given
value.
\end{enumerate}

In addition, we assume access to a general root finding algorithm
$\domfunc{roots}: \dom{Real}^+ \to \dom{Real}^*$ (not shown), that
returns a (possibly empty) list of roots of the degree-$m$
polynomial with specified coefficients.
In the reference implementation of $\sppl{}$, the $\domfunc{roots}$ function uses
symbolic analysis for polynomials whose degree is less than or equal
to two and semi-symbolic analysis for higher-order polynomials.

\subsection{Example of Exact Inference on a Many-to-One Random Variable Transformation}
\label{appx:transforms-example}

This appendix shows how \sppl{} enables exact inference on many-to-one
transformations of real random variables described in the previous
section, where the transformation is itself determined by a stochastic
branch (Fig.~\ref{fig:poly-invert} in main text).
Fig.~\ref{subfig:poly-invert-prior-program} shows an \sppl{}
program that defines a pair of random variables $(X,Z)$, where $X$ is
normally distributed; and $Z = -X^3 + X^2 + 6X$ if $X < 1$, otherwise
$Z=5\sqrt{X} +1$.
The first plot of Fig.~\ref{subfig:poly-invert-prior-dist} shows the prior distribution
of $X$; the middle plot shows the transformation $t$ that defines
$Z = t(X)$, which is a piecewise sum of $t_{\rm if}$ and $t_{\rm else}$;
and the final plot shows the distribution of $Z = t(X)$.
Fig.~\ref{subfig:poly-invert-prior-spe} shows the sum-product expression
representing this program, where the root node is a sum whose left and right
children have weights $0.691...$ and $0.309...$, which corresponds to the
prior probabilities of $\set{X < 1}$ and $\set{1 \le X}$.
Nodes labeled $X~\sim N(\mu,\sigma)$ with an incoming directed edge
from a node labeled $(r_1, r_2)$ denotes that the random variable
is constrained to the interval $(r_1, r_2)$ (and similarly for closed intervals).
Deterministic transformations are denoted by using
red directed edges from a leaf node (i.e., $X$) to a numeric expression
(e.g., $5\sqrt{X} + 11$), with the name of the transformed variable
along the edge (i.e., $Z$).

Fig.~\ref{subfig:poly-invert-condition-program} shows an \sppl{} query
that conditions the program on an event
$\set{Z^2 \le 4} \cap \set{Z \ge 0}$ involving the transformed variable $Z$.
The inference engine performs the following analysis on the query:
\begin{align}
&\set{Z^2 \le 4}\,{\cap}\,\set{Z \ge 0} \\
&\equiv \set{Z \in [0,2]}
    % && \mbox{(simplifying the event)}
    \label{eq:formation}
    \\
&\equiv \set{X \in t^{-1}([0,2])}
    \qquad \qquad \mbox{recall } (Z \defas t(X)) \\
&\equiv \set{X \in t_{\rm if}^{-1}([0,2])} \cup \set{X \in t_{\rm else}^{-1}([0,2])}
  % &&\mbox{(inverting the event)}
  \label{eq:phenomenology} \\
&\equiv
  \begin{aligned}[t]
  &\underbrace{\set{-2.174... \le X \le -2} \cup \set{0 \le X \le .321...}}_{\mbox{\rm \footnotesize constraints from left subtree}} \\
  &\qquad \cup \underbrace{\set{81/25 \le X \le 121/25}}_{\mbox{\rm \footnotesize constraint from right subtree}}
  \end{aligned} \label{eq:examinership}
\end{align}
Eq.~\eqref{eq:formation} shows the first stage of inference, which
solves any transformations in the conditioning event and yields
$\set{0 \le Z \le 2}$.
The conditional distribution of $Z$ is shown in the final plot of
Fig.~\ref{subfig:poly-invert-post-dist}.
The next step is to dispatch the simplified event to the left and
right subtrees.
Each subtree will compute the constraint on $X$ implied by the event
under the transformation in that branch, as shown in
Eq.~\eqref{eq:phenomenology}.
The middle plot of Fig.~\eqref{subfig:poly-invert-post-dist} shows the
preimage computation under $t_{\rm if}$ from the left subtree, which gives two
intervals, and $t_{\rm else}$ from the right subtree, which gives one
interval.

The final step is to transform the prior expression (Fig.~\ref{subfig:poly-invert-prior-spe})
by conditioning each subtree on the intervals in Eq.~\eqref{eq:examinership},
which gives the posterior expression (Fig.~\ref{subfig:poly-invert-post-spe}).
The left subtree in Fig.~\ref{subfig:poly-invert-prior-spe},
which originally corresponded to
$\set{X < 1}$, is split in Fig.~\ref{subfig:poly-invert-post-spe} into two subtrees that represent the events
$\set{-2.174... \le X \le -2}$ and $\set{0 \le X \le 0.321...}$, respectively, and whose
weights $0.159...$ and $0.494...$ are the (renormalized) probabilities of
these regions under the prior distribution
(first plot of Fig.~\ref{subfig:poly-invert-prior-dist}).
The right subtree in Fig.~\ref{subfig:poly-invert-prior-spe}, which
originally corresponded to $\set{1 \le X}$, is now restricted to
$\set{81/25 \le X \le 121/25}$ in Fig.~\ref{subfig:poly-invert-post-spe}
and its weight $0.347...$ is again the (renormalized) prior probability of
the region.
The graph in Fig.~\ref{subfig:poly-invert-post-spe}
represents the distribution of $(X,Z)$ conditioned on the query
in Eq.~\eqref{eq:formation}.
The new sum-product expression be used to run further queries, such as
using $\kw{simulate}$ to generate $n$ i.i.d.~random samples
$\set{(x_i,z_i)}_{i=1}^{n}$ from the posterior distributions in
Fig.~\ref{subfig:poly-invert-post-dist} or $\kw{condition}$ to
condition the program on further events.

%!TEX root=./paper.tex

\section{Conditioning Sum-Product Expressions}
\label{appx:condition}

This section presents algorithms for exact inference, that is,
conditioning the distribution defined by an element of $\SPE$
(Lst.~\ref{lst:core-semantics-sum-product}).
Sec.~\ref{subsec:condition-spe} focuses on a positive probability
$\dom{Event}$ (Lst.~\ref{lst:core-semantics-events}) and
Sec.~\ref{subsec:condition-spe-equality} focuses on a
$\dom{Conjunction}$ of equality constraints on non-transformed
variables, such as $\set{X=3} \cap \set{Y=4}$ (see also
Remark~\ref{remark:condition-measure-zero} in the main text).
We will first prove Thm.~\ref{thm:closure} from the main text, which
establishes that $\SPE$ is closed under conditioning on any positive
probability $\dom{Event}$.
For completeness, we restate the Thm.~\ref{thm:closure} below.
\condclosed*
Thm.~\ref{thm:closure} is a structural conjugacy
property~\citep{diaconis1979} for the family of probability
distributions defined by the $\SPE$ domain, where both the prior and
posterior are identified by elements of $\SPE$.
In Sec.~\ref{subsec:condition-spe}, we present the domain function
$\domfunc{condition}$ (Eq.~\eqref{eq:zemstvo}, main text)
which proves Thm.~\ref{thm:closure} by construction.
We first discuss several preprocessing algorithms that are key
subroutines used by $\domfunc{condition}$.

\subsection{Algorithms for Event Preprocessing}
\label{subsec:condition-dnf}

% These transformations include \begin{enumerate*}
% \item solving all $\dom{Transform}$ expressions;
% \item representing $\dom{Union}$ $\dom{Outcomes}$ as $\dom{Disjunction}$ $\dom{Events}$; and
% \item decomposing an arbitrary $\dom{Disjunction}$ event to a
%   $\dom{Disjunction}$ with pairwise disjoint clauses.
% \end{enumerate*}

\noindentparagraph{Normalizing an Event}
The $\domfunc{dnf}$ function (Lst.~\ref{lst:dnf-appx})
converts an $\dom{Event}$ $e$ to DNF, which we define below.

\begin{definition}
\label{def:dnf}
An $\dom{Event}$ $e$ is said to be in disjunctive normal form (DNF) if
and only if one of the following holds:
\begin{enumerate}[label={(\ref*{def:dnf}.\arabic*)},wide=4pt]
\item\label{eq:noneducation-0} $e \in \dom{Containment}$
\item\label{eq:noneducation-1} $\begin{aligned}[t]
  e &= e_1 \sqcap \dots \sqcap e_m \in \dom{Conjunction} \\
  &\implies \forall_{1 \le i \le m}.\ e_i \in \dom{Containment}
  \end{aligned}$
\item\label{eq:noneducation-2} $\begin{aligned}[t]
  e &= e_1 \sqcup \dots \sqcup e_m \in \dom{Disjunction} \\
  &\implies \forall_{1 \le i \le m}.\ e_i \in \dom{Containment} \cup \dom{Conjunction}
  \end{aligned}$
\end{enumerate}
Terms $e$ and $e_i$ in~\ref{eq:noneducation-0} and~\ref{eq:noneducation-1}
are called ``literals'' and terms $e_i$ in~\ref{eq:noneducation-2} are
called ``clauses''.
\end{definition}

We next define the notion of an $\dom{Event}$ in ``solved'' DNF.
\begin{definition}
\label{def:dnf-solved}
An $\dom{Event}$ $e$ is in solved DNF if all the following conditions hold:
\begin{enumerate*}[label=(\roman*)]
\item $e$ is in DNF;
\item all literals within a clause $e_i$ of $e$ have different variables; and
\item each literal $\sexpr{t\,\token{in}\,v}$ of $e$ satisfies
$t \in \dom{Identity}$ and $v \not\in \dom{Union}$.
\end{enumerate*}
\end{definition}

\begin{example}
Using informal notation, the solved DNF form of the event
$\set{X^2 \ge 9} \cap \set{\abs{Y} < 1}$
is a disjunction with two conjunctive clauses:
$[\set{X \in (-\infty, -3)} \cap \set{Y \in (-1,1)}]
\cup [\set{X \in (3, \infty)} \cap \set{Y \in (-1,1)}]$.
\end{example}

Lst.~\ref{lst:dnf-normalize} shows the $\domfunc{normalize}$
operation, which converts an $\dom{Event}$ $e$ to solved DNF.
In particular, predicates with (possibly nonlinear) arithmetic expressions are converted to predicates
that contain only linear expressions (which is a property of $\dom{Transform}$
and $\domfunc{preimg}$; Appx.~\ref{appx:transforms}); e.g., as in
Eqs.~\eqref{eq:formation}--\eqref{eq:examinership}.
The next result, Prop.~\ref{prop:normalize}, follows from
$\Denotv{E}e = \Denotv{E}{\domfunc{dnf}\, e}$ and
denotations of $\dom{Union}$ (Lst.~\ref{lst:core-semantics-outcomes})
and $\dom{Disjunction}$ (Lst.~\ref{lst:core-semantics-events}).

\begin{proposition}
\label{prop:normalize}
$\forall e \in \dom{Event}$, $\Denotv{E}{e} \equiv \Denotv{E}{(\domfunc{normalize}\, e)}$.
\end{proposition}

\noindentparagraph{Disjoining an Event}
Suppose that $e\in\dom{Event}$ is in DNF and has $m \ge 2$ clauses.
A key inference subroutine is to rewrite $e$ in solved DNF
(Def.~\ref{def:dnf-solved}) where all the clauses are disjoint.
%
% Such an operation is known as hyperrectangle decomposition
% (\citep[Def.~4.1]{albarghouthi2017}) or .

\begin{definition}
\label{def:dnf-disjoint}
Let $e \in \dom{Event}$ be in DNF. Two clauses $e_i$ and $e_j$ of $e$
are said to be disjoint if both $e_i$ and $e_j$ are in solved DNF and
at least one of the following conditions holds:
\begin{align}
\exists x \in (\domfunc{vars}\,e_i).\
  \Denotv{E}{e_{ix}} x &\equiv \varnothing \\
\exists x \in (\domfunc{vars}\,e_j).\
  \Denotv{E}{e_{jx}} x &\equiv \varnothing \\
\exists x \in (\domfunc{vars}\,e_i) \cap (\domfunc{vars}\,e_j).\
  \Denotv{E}{e_{ix} \sqcap e_{jx}} x &\equiv \varnothing \label{eq:disjoint-3}
\end{align}
where $e_{ix}$ denotes the unique literal of $e_i$ that contains variable $x$
(for each $x \in \domfunc{vars}\, e_i$), and similarly for $e_j$.
\end{definition}

Lst.~\ref{lst:disjoint?} shows the $\domfunc{disjoint?}$ procedure,
which given a pair of clauses $e_i$ and $e_j$ that are in solved DNF
(as produced by $\domfunc{normalize}$), returns true if and only if one of
the conditions in Def.~\ref{def:dnf-disjoint} hold.
Lst.~\ref{lst:dnf-disjoin} presents the main algorithm
$\domfunc{disjoin}$, which decomposes an arbitrary $\dom{Event}$ $e$
into solved DNF whose clauses are mutually disjoint.
Prop.~\ref{prop:disjoin-dnf} establishes the correctness and worst-case
complexity of $\domfunc{disjoin}$.

%!TEX root = ../paper.tex

%% DNF ALGORITHMS
\begin{listing*}[!t]
% \centering
\footnotesize
% \hrule\vrule
\FrameSep0pt
\begin{framed}
% DNF FACTOR
\begin{sublisting}[b]{.45\textwidth}
\mathleft{1pt}
\begin{align*}
&\domfunc{normalize}: \dom{Event} \to \dom{Event} \\[-3pt]
&\domfunc{normalize}\; \sexpr{t\, \token{in}\, v} \defas
\bmatch\, \domfunc{preimg}\, t\, v \\[-3pt]
&\quad\begin{aligned}[t]
    &\vartriangleright v'_1 \amalg \dots \amalg v'_m \Rightarrow
      \sqcup_{i=1}^{m} \sexpr{\scall{Id}{x}\,\token{in}\,v'_i} \\[-3pt]
    &\vartriangleright v' \Rightarrow \sexpr{\scall{Id}{x}\,\token{in}\,v'},
    \qquad \mbox{where } \set{x} \defas \domfunc{vars}\, t
  \end{aligned} \\[-3pt]
&\domfunc{normalize}\; (e_1 \sqcap \dots \sqcap e_m) \defas
  \domfunc{dnf}\,
    \sqcap_{i=1}^{m}(\domfunc{normalize}\, e_i) \\[-3pt]
&\domfunc{normalize}\; (e_1 \sqcup \dots \sqcup e_m) \defas
  \domfunc{dnf}\,
    \sqcup_{i=1}^{m}(\domfunc{normalize}\, e_i)
\end{align*}
\captionsetup{textfont=bf, aboveskip=-8pt}
\caption{normalize}
\label{lst:dnf-normalize}
\end{sublisting}\vrule\hfill
\begin{sublisting}[b]{.54\textwidth}
\mathleft{5pt}
% DISJOIN.
\begin{subequations}
\begin{align}
&\domfunc{disjoin}: \dom{Event} \to \dom{Event}  \notag\\[-3pt]
&\domfunc{disjoin}\, e \defas
  \blet\; (e_1 \sqcup \dots \sqcup e_m) \;\bbe\; \domfunc{normalize}\, e
  \label{eq:reincarnationist-1} \\[-3pt]
&\quad\bin\, \blet_{2\le{i}\le{m}}\; \tilde{e} \;\bbe\;
\bigsqcap_{\substack{1\le j<i \; \mid\; \neg(\domfunc{disjoint?}\, \langle e_j, e_i\rangle)}}
    (\domfunc{negate}\,e_j)
  \label{eq:reincarnationist-2} \\[-3pt]
&\quad\bin\, \blet_{2\le{i}\le{m}} \;
  \tilde{e}_i \;\bbe\; (\domfunc{disjoin}\, (e_i \sqcap \tilde{e}_i))
  \label{eq:reincarnationist-3}
\\[-3pt]
&\quad\bin\; e_1 \sqcup \tilde{e}_2 \sqcup \dots \sqcup \tilde{e}_m
\notag
\end{align}
\end{subequations}
\captionsetup{textfont=bf,aboveskip=-7pt}
\caption{disjoin}
\label{lst:dnf-disjoin}
\end{sublisting}
\end{framed}
% \hrule
\captionsetup{aboveskip=2pt,belowskip=-8pt}
\caption{$\dom{Event}$ preprocessing algorithms used by $\domfunc{condition}$.}
\label{lst:dnf}
\end{listing*}

\begin{proposition}
\label{prop:disjoin-dnf}
Let $e$ be an $\dom{Event}$ with
$h \defas \abs{\domfunc{vars}\, e}$ variables,
and suppose that
$e_1 \sqcup \dots \sqcup e_m \defas (\domfunc{normalize}\, e)$
has exactly $m\ge 1$ clauses.
Put $\tilde{e} \defas (\domfunc{disjoin}\, e)$.
Then:
\begin{enumerate}[label={(\ref*{prop:disjoin-dnf}.\arabic*)},wide]
\item\label{item:disjoin-solved-dnf} $\tilde{e}$ is in solved DNF.
\item\label{item:disjoin-disjoint} $\forall_{1 \le i \ne j \le \ell}.\ \domfunc{disjoint?}\, \langle e_i, e_j \rangle$.
\item\label{item:disjoin-semantics} $\Denotv{E}{e} = \Denotv{E}{\tilde{e}}$.
\item\label{item:disjoin-bound}
  The number $\ell$ of clauses in $\tilde{e}$ satisfies
  $\ell \le (2m-1)^h$.
\end{enumerate}
\end{proposition}

\begin{proof}
Suppose first that $(\domfunc{normalize}\, e)$ has $m=1$ clause $e_1$.
Then $\tilde{e} = e_1$, so
  \ref{item:disjoin-solved-dnf} holds since $e_1 = \domfunc{normalize}\, e$;
  \ref{item:disjoin-disjoint} holds trivially;
  \ref{item:disjoin-semantics} holds by Prop.~\ref{prop:normalize}; and
  \ref{item:disjoin-bound} holds since $\ell = (2-1)^h = 1$.
Suppose now that $(\domfunc{normalize}\, e)$ has $m > 1$ clauses.
To employ set-theoretic reasoning, fix some $x \in \dom{Var}$
and define $\Denotv{E'}{e} \defas \Denotv{V}{\Denotv{E}{e}x} \subset \dom{Outcome}$,
for all $e \in \dom{Event}$.
We have
\begin{align}
&\Denotv{E'}{e_1 \sqcup \dots \sqcup e_m} \\
  &= \cup_{i=1}^{m} \Denotv{E'}{e_i} \\
  &= \cup_{i=1}^{m}\left(\Denotv{E'}{e_i} \cap
        \neg \left[ \cup_{j=1}^{i-1} (\Denotv{E'}{e_j}) \right]\right)
        \label{eq:muckthrift-2}\\
  &= \cup_{i=1}^{m}\left(\Denotv{E'}{e_i} \cap
        \left[ \cap_{j=1}^{i-1} (\neg\Denotv{E'}{e_j}) \right]\right) \label{eq:muckthrift-3}\\
  &= \cup_{i=1}^{m}\left(\Denotv{E'}{e_i} \cap
        \left[ \cap_{j\in k(i)} (\neg\Denotv{E'}{e_j}) \right]\right) \label{eq:muckthrift-4}
\end{align}
where we define for each $i = 1, \dots, m$,
\begin{align*}
k(i) \defas \big\{ 1 \le j \le i-1  \mid
    \Denotv{E'}{e_i}
    \cap \Denotv{E'}{e_j} \ne \varnothing \big\}.
\end{align*}
Eq.~\eqref{eq:muckthrift-4} follows from the fact that
for any $i=1,\dots,m$ and $j < i$, we have
\begin{align}
j \notin k(i) \implies \left[
  \left(\Denotv{E'}{e_i} \cap \neg \Denotv{E'}{e_j} \right)
  \equiv \Denotv{E'}{e_i}
\right].
\end{align}
As $\domfunc{negate}$ (Lst.~\ref{lst:negate})
computes set-theoretic complement $\neg$ in the $\dom{Event}$ domain and
$j \notin k(i)$ if and only if $(\domfunc{disjoint?}\, e_j \, e_i)$,
it follows that the
$\dom{Event}$s $e'_i \defas e_i \sqcap \tilde{e}_i$ ($i=2,\dots,m$)
in Eq.~\eqref{eq:reincarnationist-3} are pairwise disjoint
and are also disjoint from $e_1$, so that
$\Denotv{E}{e} = \Denotv{E}{e_1 \sqcup e'_2 \sqcup \dots \sqcup e'_m}$.
Thus, if $\domfunc{disjoin}$ halts, then
all of~\ref{item:disjoin-solved-dnf}--\ref{item:disjoin-semantics} follow by
induction.

We next establish that $\domfunc{disjoin}$
halts by upper bounding the number of clauses $\ell$ returned
by any call to $\domfunc{disjoin}$.
Recalling that $h \defas \abs{\domfunc{vars}\, e}$,
we assume without loss of generality that
all clauses $e_i$ $(i=1,\dots,n)$ in
Eq.~\eqref{eq:reincarnationist-1} have the same variables
$\set{x_1, \dots, x_h}$, by ``padding'' each $e_i$ with vacuously true
literals of the form
$\sexpr{\scall{Id}{x_i}\,\token{in}\,\dom{Outcomes}}$.
Next, recall that clause $e_i$ in Eq.~\eqref{eq:reincarnationist-1}
is in solved DNF and has $m_i \ge 1$ literals
$e_{ij} = \sexpr{\scall{Id}{x_{ij}}\,\token{in}\,v_{ij}}$ where $v_{ij} \notin \dom{Union}$
(Def.~\ref{def:dnf-solved}).
Thus, $e_i$ specifies exactly one hyperrectangle in $h$-dimensional
space, where $v_{ij}$ is the ``interval'' (possibly infinite)
along the dimension specified by $x_{ij}$
in literal $e_{ij}$ $(i=1,\dots,m; j = 1, \dots, m_i)$.
%
% We can then bound the number of sub-hyperrectangles which
A sufficient condition to produce the worst-case number of pairwise disjoint
primitive sub-hyperrectangles that partition the region
$e_1 \sqcup \dots \sqcup e_m$ is when the previous clauses $e_1, \dots, e_{m-1}$
\begin{enumerate*}[label=(\roman*)]
\item are pairwise disjoint (Def.~\ref{def:dnf-disjoint}); and
\item are strictly contained in $e_m$,
  i.e., $\forall x.\ \Denotv{E}{e_j} \subsetneq \Denotv{E}{e_m}$,
  $(j=1,\dots,m-1)$.
\end{enumerate*}
If these two conditions hold, then $\domfunc{disjoin}$ partitions the
interior of the $h$-dimensional hyperrectangle specified by $e_m$ into no more
than $2(m-1)^h$ sub-hyperrectangles that do not intersect one another (and thus,
produce no further recursive calls), thereby establishing~\ref{item:disjoin-bound}.
\end{proof}

%!TEX root = ./paper.tex

% CONDITION FIGURE
\begin{listing*}[!t]
% \centering
\footnotesize
% \hrule\vrule
\FrameSep0pt
\begin{framed}
%% left panel
\begin{sublisting}[b]{.5\textwidth}
\mathleft{10pt}
\begin{align*}
&\domfunc{condition}\;
  \token{Leaf}\sexpr{x\, d\, \sigma}\, e\; \defas
  \blet\; v \;\bbe\; \Denotv{E}{(\domfunc{subsenv}\, e\, \sigma)}x\; \bin\; \bmatch\; d \\[-3pt]
&\;\begin{aligned}[t]
% DistString.
  &\vartriangleright \token{DistS}\sexpr{\sexpr{s_i\, w_i}_{i=1}^{m}}
    \Rightarrow \bmatch\; v \\[-3pt]
  &\quad\begin{aligned}[t]
    &\vartriangleright \settt{s'_1\dots s'_{l}}^{b} \Rightarrow
    \begin{aligned}[t]
    &\blet_{1 \le i \le m}\; w'_i \;\bbe\;
      \begin{aligned}[t]
        &\bif\; \bar{b}\; \bthen\; w_i \mathbf{1}[{\exists_{1 \le j \le \ell}.s'_j = s_i}] \\[-3pt]
        &\belse\; w_i \mathbf{1}[{\forall_{1 \le j \le \ell}.s'_j \ne s_i}]
      \end{aligned} \\[-3pt]
    &\bin\; \token{Leaf}\sexpr{x\; \token{DistS}\sexpr{\sexpr{s_i\; w'_i}_{i=1}^{m}}\; \sigma}
    \end{aligned}\\[-3pt]
    &\vartriangleright \belse\; \bundef
    \end{aligned}
\end{aligned} \\[-3pt]
% DistReal Continuous.
&\;\begin{aligned}[t]
  &\vartriangleright \token{DistR}{\sexpr{F\; r_1\; r_2}}
    \Rightarrow
    \bmatch\; (\domfunc{intersection}\; \sexpr{\sexpr{\tfalse\,r_1}\; \sexpr{r_2\,\tfalse}}\; v)
    \\[-3pt]
  &\quad\begin{aligned}[t]
    % &\bmatch\; (\domfunc{intersection}\; \sexpr{\sexpr{\tfalse\,r_1}\; \sexpr{r_2\,\tfalse}}\; v)\\[-3pt]
    % FiniteSet
    &\vartriangleright
      \varnothing \gor \settt{r_1\; \dots\; r_m}
      \Rightarrow \bundef \\[-3pt]
    % Interval
    &\vartriangleright \sexpr{\sexpr{b_1\,r'_1}\, \sexpr{r'_2\,b_2}}
      \Rightarrow \token{Leaf}\sexpr{x\; \token{DistR}\sexpr{F\, r'_1\, r'_2}\; \sigma} \\[-3pt]
    % Union
    &\vartriangleright
      v_1 \amalg \dots \amalg v_m \Rightarrow
        \blet_{1\le i\le m}\; w_i \;\bbe\; \Denotv{D}{d}v_i \\[-3pt]
      &\quad\begin{aligned}[t]
        &\bin\,\blet\; \set{n_1, \dots, n_k}
          \;\bbe\; \set{n \mid 0 < w_n} \span\span \\[-3pt]
        &\bin\,\blet_{1\le i \le k}\; \begin{aligned}[t]
          S_i \;\bbe\; (\domfunc{condition}\;
            \;\token{Leaf}\sexpr{x\, d\, \sigma}\,
            \;(\token{Id}\sexpr{x}\, \token{in}\, v_{n_i}))
        \end{aligned} \\[-3pt]
        &\bin\; \bif\; (k = 1) \;\bthen\; S_1
        \;\belse\,
          \oplus_{i=1}^k\sexpr{S'_i\,w_{n_i}}
          % \sexpr{S_1\; w_{n_1}}
          % \oplus \dots \oplus
          % \sexpr{S_k\; w_{n_k}} \span\span
      \end{aligned}
    \end{aligned}
\end{aligned}\\[-3pt]
% DistReal Atomic.
&\;\begin{aligned}[t]
  &\vartriangleright
    \token{DistI}{\sexpr{F\; r_1\; r_2}} \Rightarrow
    \bmatch\; (\domfunc{intersection}\; \sexpr{\sexpr{\tfalse\,r_1}\; \sexpr{r_2\,\tfalse}}\; v)
    \\[-3pt]
  &\quad\begin{aligned}[t]
    % FiniteSet
    &\vartriangleright
      \settt{r_1\; \dots\; r_m}
      \Rightarrow \blet_{1 \le i \le m}\; w_i \;\bbe\; \Denotv{D}{d} \settt{r_i} \\[-3pt]
    &\quad \begin{aligned}[t]
          &\bin\,\blet\; \set{n_1, \dots, n_k}
            \;\bbe\; \set{n \mid 0 < w_n} \span\span \\[-3pt]
          &\bin\,\blet_{1 \le i \le k}\; S_i = \sexpr{x\;
            \token{DistI}\sexpr{F\; (r_{n_i}{-}1/2)\; r_{n_i}}\; \sigma}
            \\[-3pt]
        &\bin\; \bif\; (k = 1) \;\bthen\; S_1
          \;\belse\, \oplus_{i=1}^k\sexpr{S'_i\,w_{n_i}}
      \end{aligned}\\[-3pt]
    &\vartriangleright \belse\; \mbox{\ttfamily // same as last two cases for} \token{DistR}
    \end{aligned}
\end{aligned}
\end{align*}
\captionsetup{textfont=bf,aboveskip=0pt}
\caption{Conditioning Leaf}
\label{lst:condition-leaf}
\end{sublisting}\vrule
%% RIGHT PANEL
\begin{sublisting}[b]{.497\textwidth}
\mathleft{10pt}
%% SUM
\begin{sublisting}{\textwidth}
\centering
\begin{align*}
&\domfunc{condition}\;
  (\sexpr{S_1\; w_1} \oplus \dots \oplus \sexpr{S_m\; w_m})\; e \defas \\[-3pt]
&\; \begin{aligned}[t]
  &\blet_{1 \le i \le m}\;
    w'_i \;\bbe\; w_i \left(\Denotv{P}{S_i}e\right)
      % && (1 \le i \le k)
      \\[-3pt]
  % &\bin\,\blet\; w'_i \;\bbe\; \tilde{w}_i / (\tilde{w}_1 + \dots + \tilde{w}_m) \\[-3pt]
  &\bin\,\blet\;\set{n_1, \dots, n_k} \;\bbe\; \set{n \mid 0 < w'_n} \span\span \\[-3pt]
  &\bin\,\blet_{1 \le i \le k}\;S'_i \;\bbe\;
      (\domfunc{condition}\; S_{n_i}\; e)
      % &&(1 \le i \le k)
      \\[-3pt]
  &\bin\; \bif\; (k = 1)\;
    \bthen\; S'_1 \; \belse\, \oplus_{i=1}^k\sexpr{S'_i\,w'_{n_i}}
  \end{aligned}
\end{align*}
\captionsetup{textfont=bf,aboveskip=2pt}
\caption{Conditioning Sum}
\label{lst:condition-sum}
\end{sublisting}
\hrule
%% PRODUCT
\begin{align*}
&\domfunc{condition}\;
  (S_1 \otimes \dots \otimes S_m)\; e \defas \\[-3pt]
&\bmatch\; \domfunc{disjoin}\; e \\[-3pt]
&\vartriangleright e_1 \sqcap \dots \sqcap e_h \Rightarrow
  \mbox{\scriptsize\ttfamily//one $h$-dimensional hyperrectangle} \\[-3pt]
&\quad \bigotimes_{1 \le i \le m} \left[\begin{aligned}[m]
    &\bmatch\; \set{1 \le j \le h \mid (\domfunc{vars}\, e_j) \subset (\domfunc{scope}\, S_i) }\\[-4pt]
    &\vartriangleright \set{n_1, \dots, n_k} \\[-3pt]
    &\quad\Rightarrow \domfunc{condition}\; S_i\; (e_{n_1} \sqcap \dots \sqcap e_{n_k}) \\[-3pt]
    &\vartriangleright \set{} \Rightarrow S_i
    \end{aligned}\right] \\
&\vartriangleright e_1 \sqcup \dots \sqcup e_\ell \Rightarrow
  \mbox{\scriptsize\ttfamily//$\ell \ge 2$ disjoint hyperrectangles}
  \\[-3pt]
&\quad\begin{aligned}[t]
  &\blet_{1\le{i}\le{\ell}}\; w_i \;\bbe\; \Denotv{P}{S_1 \otimes \dots \otimes S_m}e_i \\[-3pt]
  &\bin\,\blet\; \set{n_1, \dots, n_k} \;\bbe\; \set{n  \mid 0 < w_n} \\[-3pt]
  &\bin\,\blet_{1\le{i}\le{k}}\; S'_i \;\bbe\; (\domfunc{condition}\, (S_1 \otimes \dots \otimes S_m)\, e_{n_i}) \\[-3pt]
  &\bin\; \bif\; (k = 1)\;
    \bthen\; S'_1 \; \belse\, \oplus_{i=1}^k\sexpr{S'_i\,w_{n_i}}
  \end{aligned}
\end{align*}
\captionsetup{textfont=bf,aboveskip=2pt}
\caption{Conditioning Product}
\label{lst:condition-product}
\end{sublisting}
\end{framed}

% \hrule
\captionsetup{aboveskip=2pt}
\caption{Implementation of $\domfunc{condition}$ for $\dom{Leaf}$, $\dom{Sum}$, and
$\dom{Product}$ expressions using distribution semantics in Lst.~\ref{lst:core-semantics-distributions}.}
\label{lst:condition}
\end{listing*}

% \noindent
\begin{example}
The left panel in Fig.~\ref{fig:partition} shows $m = 4$ rectangles in
$\dom{Real}\times \dom{Real}$. The right panel shows a grid (in red)
with $(2m-1)^2 = 49$ primitive rectangular regions that are pairwise
disjoint from one another and whose union over-approximates the union
of the 4 rectangles. In this case, 29 of these primitive rectangular
regions are sufficient (but excessive) to exactly partition the union
of the rectangles into a disjoint union.
No more than 49 primitive rectangles are ever needed to partition
any 4 rectangles in $\dom{Real}^2$, and this bound is tight.
The bound in~\ref{item:disjoin-bound} generalizes this idea to
hyperrectangles that live in $h$-dimensional space.
\begin{figure}[H]
\begin{subfigure}{.5\linewidth}
\includegraphics[width=\linewidth]{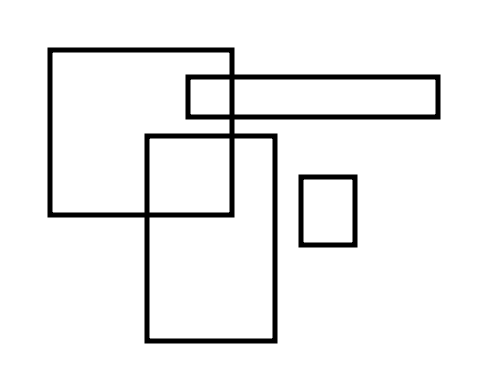}
\caption*{Conditioning Region}
\end{subfigure}%
\begin{subfigure}{.5\linewidth}
\includegraphics[width=\linewidth]{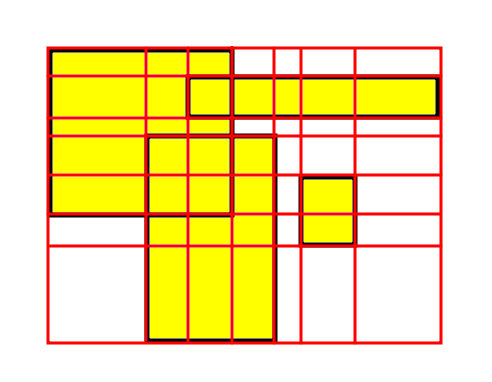}
\caption*{Partition into Rectangles}
\end{subfigure}
\caption{Example illustrating the upper bound~\ref{item:disjoin-bound}
on the number of disjoint rectangles in a worst-case partition of a
conditioning region in the two-dimensional $\dom{Real}$ plane.}
\label{fig:partition}
\end{figure}
\end{example}

\begin{remark}
When defining $\tilde{e}$ in Eq~\eqref{eq:reincarnationist-2} of
$\domfunc{disjoin}$, ignoring previous clauses that are disjoint
from $e_i$ is essential for $\domfunc{disjoin}$ to halt, so as to avoid
recursing on a primitive sub-rectangle in the interior.
That is, filtering out such clauses ensures that $\domfunc{disjoin}$
makes a finite number of recursive calls.
\end{remark}

\subsection{Algorithms for Conditioning Sum-Product Expressions on Positive Measure Events}
\label{subsec:condition-spe}

Having established the key background details, we now prove
Thm.~\ref{thm:closure} from the main text, which establishes the
closure under conditioning property of the $\dom{SP}$ domain.
\begin{proof}[Proof of Theorem.~\ref{thm:closure}]
We establish Eq.~\eqref{eq:sensize} by defining
\begin{equation}
\domfunc{condition}: \SPE \to \dom{Event} \to \SPE
\end{equation}
which satisfies
\begin{align}
\Denotv{P}{(\domfunc{condition}\, S\, e)}e' =
  \frac{\Denotv{P}{S}({e \sqcap e'})}{\Denotv{P}{S}e}
  \label{eq:appx-zemstvo}
\end{align}
for all $e' \in \dom{Event}$ and $e \in \dom{Event}$ for which
$\Denotv{P}{S}e > 0$.

We will define $\domfunc{condition}$ separately for each of the three
constructors $\dom{Leaf}$, $\dom{Sum}$, and $\dom{Product}$ from
Lst.~\ref{lst:core-syntax-sum-product}.
The proof is by structural induction, where $\dom{Leaf}$ is the base
case and $\dom{Sum}$ and $\dom{Product}$ are the recursive cases.

\noindentparagraph{Conditioning Leaf} Lst.~\ref{lst:condition-leaf}
shows the base cases of $\domfunc{condition}$.
The case of $d\in\dom{DistStr}$ is straightforward.
For $d\in\dom{DistReal}$, if the intersection (defined in second line of Lst.~\ref{lst:condition-leaf})
of $v$ with the support of $d$ is an interval $\sexpr{\sexpr{b_1'\,r'_1}\,\sexpr{r'_2,b_2'}}$, then
it suffices to return a $\dom{Leaf}$ restricting $d$ to the interval.
If the intersection is a $\dom{Union}$ $v_1 \amalg \dots \amalg v_m$
(recall fro Eq.~\eqref{eq:invariant-union-intersection} that $\domfunc{intersection}$ ensures
the $v_i$ are disjoint), then the conditioned $\SPE$ is a $\dom{Sum}$,
whose $i$th child is obtained by recursively calling $\domfunc{condition}$
on $v_i$ and $i$th (relative) weight is the probability of $v$ under $d$, since,
for any new $v' \in \dom{Outcomes}$, we have
\begin{align}
\begin{aligned}
\frac%
  {\Denotv{D}{d}{(\domfunc{intersect}\, v'\, (v_1 \amalg \dots \amalg v_m))}}
  {\Denotv{D}{d}{(v_1 \amalg \dots \amalg v_m)}} \\
=\frac%
  {\Denotv{D}{d}{\amalg_{i=1}^{m} (\domfunc{intersect}\, v'\, v_i) }}
  {\sum_{i=1}^{m}\Denotv{D}{d}{v_i}}.
\end{aligned}
\label{eq:ebulliently}
\end{align}
Eq.~\eqref{eq:ebulliently} follows from the additivity of $\Denotv{D}{d}$.
The plots of $X$ in Figs.~\ref{subfig:poly-invert-prior-dist}
  and~\ref{subfig:poly-invert-post-dist} illustrate the
  identity in Eq.~\eqref{eq:ebulliently}, where conditioning the unimodal
  normal distribution results in a mixture of three
  restricted normals whose weights are given by the relative prior probabilities
  of the three regions.

For $d\in\dom{DistInt}$, if the positive probability
$\dom{Outcomes}$ are $\settt{r_1 \dots r_m}$, then the conditioned $\SPE$ is a
$\dom{Sum}$ of ``delta''-CDFs whose atoms are located on the integers
$r_i$ and weights are the (relative) probabilities
$\Denotv{D}{d}\settt{r_i}$ $(i=1,\dots,m)$.
Since the atoms of $F$ for $\dom{DistInt}$ are integers, it suffices
to restrict $F$ to the interval $(r_i-1/2, r_i)$, for
each $r_i$ with a positive weight.
Correctness again follows from Eq.~\eqref{eq:ebulliently}, since
finite sets are unions of disjoint singleton sets.
For other positive probability $\dom{Outcomes}$, the conditioning
procedure $\dom{DistInt}$ is the same as that for $\dom{DistReal}$.

\clearpage

\noindentparagraph{Conditioning Sum}
Lst.~\ref{lst:condition-sum} shows $\domfunc{condition}$ for $S \in \dom{Sum}$.
Recalling the denotation $\Denotv{P}{S}$ for $S \in \dom{Sum}$ in
Lst.~\ref{lst:core-semantics-sum-product}, the correctness follows
from the following properties:
\begin{align}
&\frac%
  {\Denotv{P}{\sexpr{S_1\; w_1} \oplus \dots \oplus \sexpr{S_m\; w_m}}({e \sqcap e'})}
  {\Denotv{P}{\sexpr{S_1\; w_1} \oplus \dots \oplus \sexpr{S_m\; w_m}}e} \\
&= \frac%
  {\sum_{i=1}^{m}w_i\Denotv{P}{S_i}({e \sqcap e'})}
  {\sum_{i=1}^{m}w_i\Denotv{P}{S_i}e}  \label{eq:pravity-1} \\
&= \frac%
  {\sum_{i=1}^{m}w_i(\Denotv{P}{S_i}e)\Denotv{P}{(\domfunc{condition}\, S_i\, e)}{e'}}
  {\sum_{i=1}^{m}w_i\Denotv{P}{S_i}e}
  \label{eq:pravity-2} \\
&= \Denotv{P}{\oplus_{i=1}^{m} \sexpr{(\domfunc{condition}\, S_i\, e)\;, w_i\Denotv{P}{S_i}e}}e',
\end{align}
where Eq.~\eqref{eq:pravity-2} has applied
Eq.~\eqref{eq:appx-zemstvo} inductively for each $S_i$.
Eqs.~\eqref{eq:pravity-1}--\eqref{eq:pravity-2} assume for simplicity that
  $\Denotv{P}{S_i}e>0$ for each $i=1,\dots,m$, whereas
  Lst.~\ref{lst:condition-leaf} does not make this assumption.

\noindentparagraph{Conditioning Product}
Lst.~\ref{lst:condition-product} how $\domfunc{condition}$ operates on
$S \in \dom{Product}$.
The first step is to invoke $\domfunc{disjoin}$ to rewrite $(\domfunc{dnf}\, e)$
as $\ell \ge 1$ disjoint clauses $e'_1 \sqcup \dots \sqcup e'_\ell$
(recall from Prop.~\ref{prop:disjoin-dnf} that $\domfunc{disjoin}$
is semantics-preserving).
The first pattern in the $\bmatch$ statement corresponds $\ell = 1$,
and the result is a new $\dom{Product}$, where the $i$th child is
conditioned on the literals of $e_1$ whose variables are contained in
$\domfunc{scope}\, S_i$ (if any).
The second pattern returns a $\dom{Sum}$ of $\dom{Product}$, since
\newcommand{\Sp}{S_1 \otimes \dots \otimes S_m}
\begin{align}
&\frac%
  {\Denotv{P}{\Sp}{(e \sqcap e')}}
  {\Denotv{P}{\Sp}{e}} \\
&=\frac%
  {\Denotv{P}{\Sp}{((e_1 \sqcup \dots \sqcup e_{\ell}) \sqcap e')}}
  {\Denotv{P}{\Sp}{(e_1 \sqcup \dots \sqcup e_{\ell})}} \\
&=\frac%
  {\Denotv{P}{\Sp}{((e_1 \sqcap e') \sqcup \dots \sqcup (e_{\ell} \sqcap e'))}}
  {\sum_{i=1}^{\ell}\Denotv{P}{\Sp}{e_i}} \\
&=\frac%
  {\sum_{i=1}^{\ell}\Denotv{P}{\Sp}{(e_i \sqcap e')}}
  {\sum_{i=1}^{\ell}\Denotv{P}{\Sp}{e_i}} \\
&=\frac%
  {\sum_{i=1}^{\ell} \Denotv{P}{S}{e_i}\, \Denotv{P}{(\domfunc{condition}\, (\Sp)\, e_i)}{e'}}
  {\sum_{i=1}^{\ell}\Denotv{P}{\Sp}{e_i}} \label{eq:kain-4} \\
&= \Denotv{P}{\oplus_{i=1}^{\ell} \sexpr{(\domfunc{condition}\, S\, e_i)\;\, \Denotv{P}{S}e_i}}e'.
\end{align}
Eq~\eqref{eq:kain-4} follows from the induction hypothesis
Eq.~\eqref{eq:appx-zemstvo}
and $(\domfunc{disjoin}\, e_i) \equiv e_i$ (idempotence), so that
$(\domfunc{disjoin}\, e_i \sqcap e') \equiv (\domfunc{disjoin}\, e_i) \sqcap (\domfunc{disjoin}\, e') \equiv e_i \sqcap (\domfunc{disjoin}\, e')$.

Thm.~\ref{thm:closure} is thus established.
\end{proof}

Fig.~\ref{fig:hyperrectangle} in the main text shows an example of the
closure property from Thm.~\ref{thm:closure}, where conditioning
on a hyperrectangle changes the structure of the $\SPE$ from a
$\dom{Product}$ into a $\dom{Sum}$-of-$\dom{Product}$.
The algorithms in this section are the first to describe probabilistic
inference and closure properties for conditioning an $\SPE$ on a query
that involves transforms of random variables and predicates
with set-valued constraints.
We next establish Thm.~\ref{thm:condition-linear} from the main text,
which gives a sufficient condition for the runtime of
$\domfunc{condition}$ (Lst.~\ref{lst:condition}) to scale linearly in
the number of nodes in $S$; identical results hold for computing
$\dom{Event}$ probabilities ($\Denotv{P}{S}{e}$,
Lst.~\ref{lst:core-semantics-sum-product}) and probability densities
($\Denot[\mathbb{P}_0]{S}e$,
Lst.~\ref{lst:core-semantics-sum-product-density}).

\condlinear*

\begin{proof}
First, if $S$ is a $\dom{Sum}$ (Lst.~\ref{lst:condition-sum}) with $m$
children then $(\domfunc{condition}\, S\, e)$ makes no more than $m$ subcalls to
$\domfunc{condition}$ (one for each child), and if $S$ is a
$\dom{Leaf}$ (Lst.~\ref{lst:condition-leaf}) then there are zero
subcalls, independently of $e$.
Since each node has exactly one parent, we can can conclude that each
node in $S$ is visited exactly once by showing that the hypothesis on
$e$ implies that for any $S \in \dom{Product}$
(Lst.~\ref{lst:condition-product}) with $m$ children, there are makes
at most $m$ subcalls to $\domfunc{condition}$ from which we  (each
node has exactly one parent).
Suppose that $(\domfunc{disjoin}\, e)$ returns a single
$\dom{Conjunction}$.
Then the first pattern of the $\bmatch$ statement in
Lst.~\ref{lst:condition-product} is matched (one $h$-dimensional
rectangle), resulting in $m$ subcalls to $\domfunc{condition}$.
Thus, each node in $S$ is visited (at most) once by
$\domfunc{condition}$.
To complete the proof, note that the hypothesis that $e$ specifies a single
$\dom{Conjunction}$
$\sexpr{t_1\,\token{in}\,v_1} \sqcap \dots \sqcap
  \sexpr{t_m\,\token{in}\,v_m}$ of $\dom{Containment}$ constraints on
non-transformed variables is sufficient for $(\domfunc{disjoin}\, e)$
to return a single $\dom{Conjunction}$.
\end{proof}

\subsection{Conditioning Sum-Product Expressions on Measure Zero Equality Constraints}
\label{subsec:condition-spe-equality}

Recall from Remark~\ref{remark:condition-measure-zero} in the main
text that $\SPE$ is also closed under conditioning on a
$\dom{Conjuction}$ of possibly measure zero equality constraints of
non-transformed variable, such as
$\set{X=3, Y=\pi, Z=\dquote{foo}}$.
In this section, we describe the conditioning algorithm for this case,
which is implemented by
\begin{align}
\domfunc{condition}_0: \SPE \to \dom{Event} \to \SPE,
\label{eq:smytrie}
\end{align}
where $e \in \dom{Event}$ satisfies the follows requirements
with respect to $S \in \SPE$:
\begin{enumerate}
\item Either $e \equiv \sexpr{\scall{Id}{x}\,\token{in}\,\settt{\mli{rs}}}$ or
  $e$ is a $\dom{Conjunction}$ of such literals, where $\equiv$ here
  denote syntactic (not semantic) equivalence.
\item Every $\dom{Var}$ $x$ in each literal of $e$ is a
  non-transformed variable; i.e., for each $\dom{Leaf}$ expression $S$
  such that $x \in \domfunc{scope}\, S$, we have
  $S \equiv \scall{Leaf}{x\,d\,\sigma}$, for some $d$ and $\sigma$.
\end{enumerate}

With these requirements on $e$, Lst.~\ref{lst:condition-zero} presents
the implementation of $\domfunc{condition}_0$, leveraging the generalized
density semantics from Lst.~\ref{lst:core-semantics-sum-product-density} in the
main text.
The inference rules closely match those for standard sum-product
networks, except for the fact that a density from $\Denot[\mathbb{P}_0]{S}$ is a pair,
whose first entry is the number of continuous distributions participating in the
weight of the $\dom{Event}$ $e$ which must be correctly accounted for
by $\domfunc{condition}_0$.
In the reference implementation of \sppl{},
$\kw{condition}$ invokes $\domfunc{condition}$
and  $\kw{constrain}$ invokes $\domfunc{condition}_0$.
Analogously to the $\kw{prob}$ query, which returns
probabilities using the distribution semantics
$\mathbb{P}$ in
Lst.~\ref{lst:core-syntax-distributions}, \sppl{} also includes
the $\kw{density}$ query, which returns densities using the generalized
semantics $\mathbb{P}_0$ in Lst.~\ref{lst:core-semantics-sum-product-density}.

%!TEX root = ./paper.tex

% CONDITION FIGURE
\begin{listing}[t]
\footnotesize
\FrameSep0pt
\begin{framed}
\begin{sublisting}{\linewidth}
\mathleft{10pt}
\begin{align*}
&\domfunc{condition}_0\;
  {\scall{Leaf}{x\, d\, \sigma}}\;
  \sexpr{\scall{Id}{x}\,\token{in}\,\settt{\mli{rs}}}
  \defas \bmatch\, d \\[-4pt]
  &\quad \begin{aligned}[t]
    &\vartriangleright
      \scall{DistR}{F\, r_1\, r_2}
      \Rightarrow\, \bmatch\, {\mli{rs}} \\[-4pt]
    &\quad \begin{aligned}[t]
      &\vartriangleright r \Rightarrow
        \bmatch\, (
          \Denot[\mathbb{P}_0]{\scall{Leaf}{x\, d\, \sigma}}\;
          \sexpr{\scall{Id}{x}\,\token{in}\,\settt{\mli{rs}}}) \\[-4pt]
      &\quad \begin{aligned}[t]
          &\vartriangleright (1, 0) \Rightarrow \bundef \\[-4pt]
          &\vartriangleright \belse\,
            \blet\; \tilde{F}\, \bbe\, \left(\lambda r'.\ \bindicator{r \le r'}\right)\;
            \bin\; \scall{DistI}{\tilde{F}\, (r-1/2)\, r}
          \end{aligned}
        \\[-4pt]
      &\vartriangleright s \Rightarrow \bundef
    \end{aligned} \\[-4pt]
    &\vartriangleright \belse \Rightarrow
      \domfunc{condition}\;
      {\scall{Leaf}{x\, d\, \sigma}}\;
      \sexpr{\scall{Id}{x}\,\token{in}\,\settt{\mli{rs}}}
    \end{aligned}
\end{align*}
\captionsetup{textfont=bf,aboveskip=0pt}
\caption{Conditioning Leaf}
\label{lst:condition-zero-leaf}
\hrule

\begin{align*}
&\domfunc{condition}_0\;
  (\sexpr{S_1\; w_1} \oplus \dots \oplus \sexpr{S_m\; w_m})\;
  \left( \sqcap_{i=1}^{\ell} \sexpr{\scall{Id}{x_i}\,\token{in}\,\settt{\mli{rs_i}}}\right)
  \defas \\[-3pt]
&\; \begin{aligned}[t]
  &\blet_{1 \le i \le m}\;
    (d_i, p_i) \;\bbe\; \Denot[\mathbb{P}_0]{S_i}
      \left( \sqcap_{i=1}^{\ell} \sexpr{\scall{Id}{x_i}\,\token{in}\,\settt{\mli{rs_i}}}\right) \\[-3pt]
  &\bin\, \begin{aligned}[t]
    &\bif\, \forall_{1 \le i \le m}.\ p_i = 0\, \bthen\, \bundef \\
    &\belse\, \begin{aligned}[t]
      &\blet_{1 \le i \le m}\, w'_i \;\bbe\; w_i p_i \\[-3pt]
      &\bin\,\blet\, d^* \,\bbe\, \min\set{d_i \mid 1 \le i \le m, 0 < p_i} \\[-4pt]
      &\bin\,\blet \set{n_1, \dots, n_k} \;\bbe\; \set{n \mid 0 < w'_n, d_i = d^*} \span\span \\[-3pt]
      &\bin\,\blet_{1 \le i \le k}\;S'_i \;\bbe\;
        \left(\domfunc{condition}_0\; S_{n_i}\; \left( \sqcap_{i=1}^{\ell} \sexpr{\scall{Id}{x_i}\,\token{in}\,\settt{\mli{rs_i}}}\right)\right)
      % &&(1 \le i \le k)
      \\[-3pt]
      &\bin\; \bif\; (k = 1)\;
      \bthen\; S'_1 \; \belse\, \oplus_{i=1}^k\sexpr{S'_i\,w_{n_i}}
      \end{aligned}
    \end{aligned}
  \end{aligned}
\end{align*}
\captionsetup{textfont=bf,aboveskip=2pt}
\caption{Conditioning Sum}
\label{lst:condition-zero-sum}
\hrule

\begin{align*}
&\domfunc{condition}_0\;
  \left(S_1 \otimes \dots \otimes S_m \right)\;
  \sqcap_{i=1}^{\ell} \sexpr{\scall{Id}{x_i}\,\token{in}\,\settt{\mli{rs_i}}}
  \defas \\[-4pt]
  &\quad \begin{aligned}[t]
    &\blet_{1 \le i \le m}\, S'_i \,\bbe\,
      \bmatch\; \set{x_1, \dots, x_m} \cap (\domfunc{scope}\, S_i) \\[-4pt]
      &\quad \vartriangleright \set{n_1, \dots, n_k}
        \Rightarrow  \domfunc{condition}_0\, {S_i}\,
          \sqcap_{t=1}^{k} \sexpr{\scall{Id}{x_{n_t}}\,\token{in}\,\settt{\mli{rs_t}}} \\[-4pt]
      &\quad \vartriangleright \set{} \Rightarrow S_i \\[-4pt]
    &\bin\, S'_1 \otimes \dots \otimes S'_m
    \end{aligned}
\end{align*}
\captionsetup{textfont=bf,aboveskip=2pt}
\caption{Conditioning Product}
\label{lst:condition-zero-product}
% \hrule
\end{sublisting}
\end{framed}

\captionsetup{aboveskip=2pt}
\caption{Implementation of $\domfunc{condition}_0$ for $\dom{Leaf}$, $\dom{Sum}$, and
$\dom{Product}$ expressions using density semantics in Lst.~\ref{lst:core-semantics-sum-product-density}.}
\label{lst:condition-zero}
\end{listing}

%!TEX root = ./paper.tex

\section{Translating Sum-Product Expressions to \sppl{} Programs}
\label{appx:translation-reverse}

Lst.~\ref{lst:sppl-translation} in Sec.~\ref{sec:translation} presents
the relation $\translate$, that translates $C \in \dom{Command}$
(i.e., \sppl{} source syntax, Lst.~\ref{lst:sppl-syntax}) to a sum-product expression $S \in \SPE$
in the core language (Lst.~\ref{lst:core-syntax}).
% (the symbol $S_{\varnothing}$
% indicates an ``empty'' $\dom{SP}$).
%
Lst.~\ref{lst:spe-translation} defines a relation
$\translateR$ that reverses the $\translate$ relation:
it converts expression $S \in \SPE$ to $C \in \dom{Command}$.
Briefly, \begin{enumerate*}[label=(\roman*)]
\item a $\dom{Product}$ is converted to a sequence $\dom{Command}$;
\item a $\dom{Sum}$ is converted to an \kw{if}-\kw{else} $\dom{Command}$; and
\item a $\dom{Leaf}$ is converted to a sequence of
  sample (\texttt{\textasciitilde}) and transform (\texttt{=}).
\end{enumerate*}
The symbol $\Uparrow$ (whose definition is omitted) in the \ref{Leaf}
rule converts semantic elements such as $d\in\dom{Distribution}$ and
$t\in\dom{Transform}$ from the core calculus
(Lst.~\ref{lst:core-semantics}) to an \sppl{} expression $E \in \dom{Expr}$
(Lst.~\ref{lst:sppl-syntax}) in a straightforward way, e.g.,
\begin{align}
(\scall{Poly}{\scall{Id}{\texttt{X}}\; 1\; 2\; 3})
  \Uparrow
  (\texttt{1 + 2*X + 3*X**2}).
\end{align}

It is easy to see that chaining $\translate$
(Lst.~\ref{lst:sppl-translation}) and $\translateR$
(Lst.~\ref{lst:spe-translation}) for a given
\sppl{} program does not preserve either \sppl{} or core syntax, that is%
\footnote{The symbol $C \translateStar S$ means
$\langle C, S_{\varnothing}\rangle$ translates to $S$
in zero or more steps of $\translate$, where $S_{\varnothing}$ is an
``empty'' $\dom{SP}$ used for the initial translation step, and
similarly for $\translateRStar$.}
\begin{align*}
&((C \translateStar S) \translateRStar C') &&\mbox{ does not imply } C = C' \\
&((C \translateStar S) \translateRStar C') \translateStar S' &&\mbox{ does not imply } S = S'.
% \hspace{-1cm}
\end{align*}
Instead, it can be shown that $\translateR$ is a semantics-preserving
inverse of $\translate$, in the sense that for all $e \in \dom{Event}$
{\mathleft{0pt}\begin{align}
((C \translateStar S) \translateRStar C') \translateStar S'
&\implies \Denotv{P}{S}e = \Denotv{P}{S'}e.
\hspace{-.5cm}
\label{eq:eclipsation-2}
\end{align}}
Eq.~\eqref{eq:eclipsation-2} establish a formal semantic
correspondence between the \sppl{} language and the class of
sum-product expressions:
each \sppl{} program admits a representation as an $\SPE$, and each
valid element of $\SPE$ that satisfies conditions \ref{item:definedness-env-base}--\ref{item:definedness-sum-weights}
expression corresponds to some \sppl{} program.

Thus, in addition to synthesizing full \sppl{} programs from data
using the PPL synthesis systems~\citep{chasins2017,saad2019bayesian}
mentioned in Sec.~\ref{sec:related},
it is also possible with the
translation strategy in Lst.~\ref{lst:spe-translation} to synthesize
\sppl{} programs using the wide range of techniques for learning the
structure and parameters of sum-product
networks~\citep{gens2013,lee2014,vergari2019,trapp2019}.
With this approach, \sppl{}
\begin{enumerate*}[label=(\roman*)]
\item provides users with a uniform representation of existing sum-product
networks as generative source code in a formal PPL
(Lst.~\ref{lst:sppl-syntax});
\item allows users to extend these baseline programs with modeling
extensions supported by the core calculus
(Lst.~\ref{lst:core-semantics}), such as predicates for decision trees
and numeric transformations; and
\item delivers exact answers to an extended set of probabilistic
inference queries (Sec.~\ref{sec:condition}) within the modular and reusable
workflow from Fig.~\ref{fig:system-diagram}.
\end{enumerate*}

%!TEX root =  ../paper.tex
\begin{listing}[t]
\footnotesize

\staterule{Leaf}
{d \Uparrow D\texttt{(}E\texttt{)}, t_1 \Uparrow E_1, \dots, t_m \Uparrow E_m}
{\begin{aligned}[t]
  &(x\; d\; \set{x\mapsto \token{Id}\sexpr{x}, x_1\mapsto t_1, \dots, x_m\mapsto t_m}) \\
  &\qquad\qquad \translateR
  x\,\texttt{\textasciitilde}\,D\texttt{(}E\texttt{)} \texttt{;}
  x_1\, \texttt{=}\, E_1 \texttt{;}
  \dots \texttt{;}
  x_m\, \texttt{=}\, E_m\end{aligned}}

\bigskip

\staterule{Product}
{S_1 \translateR C_1,
  \dots,
  S_m \translateR C_m}
{\otimes_{i=1}^{m} S_i \translateR C_1\texttt{;} \dots \texttt{;} C_m}

\bigskip

\staterule{Sum}
{S_1 \translateR C_1, \dots, S_m \translateR C_m; \quad \textrm{where}\ b\ \mbox{is a fresh } \dom{Var}}
{\\[-2pt] \oplus_{i=1}^{m}\sexpr{S_i\, w_i}
  \translateR
  \left[\begin{aligned}[m]
    &b\,\texttt{\textasciitilde}\, \texttt{choice}\texttt{(\{\textquotesingle1\textquotesingle:} w_1\texttt{,}\dots\texttt{,}\texttt{\textquotesingle}m\texttt{\textquotesingle:}w_m\texttt{\})} \\
    &\kw{if}\, \texttt{(}b\, \texttt{==}\, \squote{1}\texttt{)}\, \settt{C_1} \\
    &\kw{elif} \dots \\
    &\kw{elif}\, \texttt{(}b\, \texttt{==}\, \texttt{\textquotesingle}m\texttt{\textquotesingle}\texttt{)}\, \settt{C_m} \\
  \end{aligned}\right]
}
\smallskip
\hrule
\captionsetup{aboveskip=5pt}
\caption{Translating an element of $\SPE$
(Lst.~\ref{lst:core-syntax-sum-product})
to an \sppl{} command $C$ (Lst.~\ref{lst:sppl-syntax}).}
\label{lst:spe-translation}
\end{listing}

%!TEX root = ./paper.tex

\begin{listing*}
\centering
\scriptsize
\FrameSep0pt
\begin{framed}
\begin{sublisting}[b]{.277\textwidth}
\mathleft{1pt}
\begin{sublisting}{\textwidth}
\begin{align*}
x &\in \dom{Var} \\[-4pt]
n &\in \dom{Natural}\\[-4pt]
b &\in \dom{Boolean} \defas \set{\ttrue, \tfalse}\\[-4pt]
u &\in \dom{Unit} \defas \set{\tunit}\\[-4pt]
w &\in [0,1]  \\[-4pt]
r &\in \dom{Real} \cup \set{-\infty, \infty} \\[-4pt]
s &\in \dom{String} \defas \dom{Char}^*
\end{align*}
\captionsetup{aboveskip=-5pt, textfont=bf}
\caption{Basic Sets}
\label{lst:core-syntax-basic}
\end{sublisting}
\hrule
\smallskip
\begin{align*}
\mli{rs} &\in \dom{Outcome} \defas \dom{Real} + \dom{String} \span\span\\[-4pt]
v &\in \dom{Outcomes} \\[-4pt]
    % \begin{aligned}[t]
    &\defas \varnothing                                   && [\dom{Empty}]\\[-4pt]
    &\gor   \settt{s_1 \dots s_m}^{b}                     && [\dom{FiniteStr}]\\[-4pt]
    &\gor   \settt{r_1 \dots r_m}                         && [\dom{FiniteReal}]\\[-4pt]
    &\gor   \sexpr{\sexpr{b_1\, r_1}\, \sexpr{r_2\, b_2}} && [\dom{Interval}]\\[-4pt]
    &\gor   v_1 \amalg \dots \amalg v_m                   && [\dom{Union}]
  % \end{aligned}
\end{align*}
\captionsetup{skip=-7pt, textfont=bf}
\caption{Outcomes}
\label{lst:core-syntax-outcomes}
\end{sublisting}\vrule
\begin{sublisting}[b]{.325\textwidth}
\mathleft{2pt}
\begin{sublisting}{\textwidth}
\begin{align*}
t &\in \dom{Transform} \\[-4pt]
% \begin{aligned}[t]\\[-4pt]
  &\defas \scall{Id}{x}                    && [\dom{Identity}]\\[-4pt]
  &\gor   \scall{Reciprocal}{t}            && [\dom{Reciprocal}]\\[-4pt]
  &\gor   \scall{Abs}{t}                   && [\dom{AbsValue}]\\[-4pt]
  &\gor   \token{Root}\sexpr{t\,n}         && [\dom{Radical}]\\[-4pt]
  &\gor   \token{Exp}\sexpr{t\,r}          && [\dom{Exponent}]\\[-4pt]
  &\gor   \token{Log}\sexpr{t\,r}          && [\dom{Logarithm}]\\[-4pt]
  % &\gor   \scall{Sin}{t\, r_1\, r_2}       && [\dom{Sinusoid}] \\[-4pt]
  &\gor   \scall{Poly}{t\; r_0 \dots r_m}  && [\dom{Polynomial}]\\[-4pt]
  &\gor   \token{Piecewise}
    \begin{aligned}[t]
      \texttt{(}&\sexpr{t_1\, e_1} \\[-4pt]
      &\dots \\[-4pt]
      &\sexpr{t_m\, e_m}\texttt{)}
    \end{aligned} && [\dom{Piecewise}]
  % && [\dom{Piecewise}]
  % &\gor   \token{Max}\sexpr{t_1\, t_2}      && [\dom{Maximum}] \\[-4pt]
  % &\gor   \token{Min}\sexpr{t_2\, t_2}      && [\dom{Minimum}]
\end{align*}
\captionsetup{skip=-5pt, textfont=bf}
\caption{Transformations}
\label{lst:core-syntax-trasnformations}
\hrule
\end{sublisting}
\begin{align*}
e &\in \dom{Event} \\[-4pt]
  &\defas  \sexpr{t\; \token{in}\, v}  && \mbox{[Containment]}\\[-4pt]
  &\gor    e_1 \sqcap \dots \sqcap e_m  && \mbox{[Conjunction]}\\[-4pt]
  &\gor    e_1 \sqcup \dots \sqcup e_m  && \mbox{[Disjunction]}
\end{align*}
\captionsetup{skip=-7pt, textfont=bf}
\caption{Events}
\label{lst:core-syntax-events}
\end{sublisting}\vrule
\begin{sublisting}[b]{.394\textwidth}
\mathleft{2pt}
\begin{sublisting}{\textwidth}
\begin{align*}
% e &\in \dom{Event}
% \begin{aligned}[t]
%   &\defas  t \; \token{in}\; v && \mbox{[Containment]}\\[-4pt]
%   &\gor    e_1 \sqcup \dots \sqcup e_m && \mbox{[Disjunction]}\\[-4pt]
%   &\gor    e_1 \sqcap \dots \sqcap e_m && \mbox{[Conjunction]}
% \end{aligned}\span\span\\[-4pt]
F &\in \dom{CDF} \subset \dom{Real} \to [0,1] \\[-4pt]
  &\defas
    \scall{Norm}{r_1,r_2}
    \gor \scall{Poisson}{r}
    \gor \scall{Binom}{n,w}
    \dots
  \span\span \\[-4pt]
&\mbox{ where $F$ is c{\`a}dl{\`a}g};\\[-4pt]
&\quad \lim_{r\to\infty}F(r)=1; \lim_{r\to-\infty}F(r)=0; \\[-4pt]
&\quad \mbox{and } F^{-1}(u) \defas \inf\set{r \mid u \le F(r)}.\\
d &\in \dom{Distribution} \span\span \\[-4pt]
% \begin{aligned}[t]
  &\defas \scall{DistR}{F\, r_1\, r_2}      && [\dom{DistReal}] \\[-4pt]
  &\gor \scall{DistI}{F\, r_1\, r_2}        && [\dom{DistInt}] \\[-4pt]
  &\gor \scall{DistS}{\sexpr{s_1\,w_1} \dots \sexpr{s_m\,w_m}}
                                            && [\dom{DistStr}]
% \end{aligned}
\end{align*}
\captionsetup{skip=-5pt, textfont=bf}
\caption{Primitive Distributions}
\label{lst:core-syntax-distributions}
\end{sublisting}
\hrule

\begin{align*}
\sigma &\in \dom{Environment} \defas \dom{Var} \to \dom{Transform}\span\span \\[-4pt]
S &\in \SPE\\[-4pt]
% \begin{aligned}[t]
  &\defas \scall{Leaf}{x\, d\, \sigma}              && [\dom{Leaf}]\\[-4pt]
  &\gor  \sexpr{S_1\, w_1}
          \oplus \dots \oplus \sexpr{S_m\,w_m}      && [\dom{Sum}] \\[-4pt]
  &\gor   S_1 \otimes \dots \otimes S_m             && [\dom{Product}]
% \end{aligned}
\end{align*}
\captionsetup{skip=-7pt, textfont=bf}
\caption{Sum-Product}
\label{lst:core-syntax-sum-product}
\end{sublisting}
\end{framed}
% \hrule
\captionsetup{aboveskip=2pt}
\caption{Core calculus.}
\label{lst:core-syntax}
\end{listing*}

%!TEX root = ./paper.tex

\begin{listing*}
\begin{framed}
\begin{align*}
\domfunc{complement}\; \settt{s_1 \dots s_m}^{b}
  &\defas \settt{s_1 \dots s_m}^{\neg b}
\\
\domfunc{complement}\;
  \sintvl{b_1\;r_1}{r_2\;b_2}
  &\defas
    \sintvl{\tfalse\;{-\infty}}{r_1\;\neg b_1}
    \amalg
    \sintvl{\neg b_2\;r_2}{\infty\;\tfalse}
\\
\domfunc{complement}\;
  \settt{r_1 \dots r_m}
  &\defas
    \begin{aligned}[t]
    &\sintvl{\tfalse\;{-\infty}}{r_1\;\ttrue} \\
    &\amalg \left[\amalg_{j=2}^{m} \sintvl{\ttrue\;r_{j-1}}{r_j\;\ttrue}\right] \\
    &\amalg \sintvl{\ttrue\;r_m}{\infty\;\tfalse}
  \end{aligned}
\\
\domfunc{complement}\; \varnothing
  &\defas
  \settt{}^{\ttrue} \amalg \sintvl{\tfalse\, {-\infty}}{\infty\, \tfalse}
\end{align*}
\end{framed}
\captionsetup{skip=2pt}
\caption{Implementation of $\domfunc{complement}$ on the sum domain $\dom{Outcomes}$.}
\label{lst:complement}
\end{listing*}

\begin{listing*}
\begin{framed}
\begin{align*}
\domfunc{vars}&: (\dom{Transform} + \dom{Event}) \to \mathcal{P}(\dom{Vars}) \\
\domfunc{vars}\, \mli{te} &= \bmatch\, \mli{te} \\
&\vartriangleright t \Rightarrow \begin{aligned}[t]
  &\bmatch\; t \\
  &\vartriangleright \scall{Id}{x} \Rightarrow \set{x} \\
  &\vartriangleright \scall{Root}{t'\; n}
        \gor \scall{Exp}{t'\; r}
        \gor \scall{Log}{t'\; r}
        \gor \scall{Abs}{t'} \\
        &\qquad \gor \scall{Reciprocal}{t'}
        \gor \scall{Poly}{t'\; r_0\; \dots\; r_m} \\
        &\qquad \Rightarrow \domfunc{vars}\; t' \\
  &\vartriangleright
    \scall{Piecewise}{\sexpr{t_i\; e_i}_{i=1}^m}
    \Rightarrow \cup_{i=1}^{m}((\domfunc{vars}\; t_i)\cup (\domfunc{vars}\; e_i))
  \end{aligned} \\
&\vartriangleright \sexpr{t\;\token{in}\;v} \Rightarrow \domfunc{vars}\; t \\
&\vartriangleright
  (e_1 \sqcap \dots \sqcap e_m) \gor
  (e_1 \sqcup \dots \sqcup e_m)
  \Rightarrow \cup_{i=1}^{m} \domfunc{vars}\; e_i
\end{align*}
\end{framed}
\captionsetup{skip=2pt}
\caption{Implementation of $\domfunc{vars}$, which returns the variables
in a $\dom{Transform}$ or $\dom{Event}$.}
\label{lst:vars}
\end{listing*}

\begin{listing*}
\begin{framed}
\begin{align*}
&\domfunc{scope}: \SPE \to \mathcal{P}(\dom{Var}) \\[-5pt]
&\domfunc{scope}\; \sexpr{x\, d\, \sigma} \defas \mathrm{dom}(\sigma) \\[-5pt]
&\domfunc{scope}\; \sexpr{S_1 \otimes \dots \otimes S_m}
  \defas \cup_{i=1}^{m} (\domfunc{scope}\; S_i)\\[-5pt]
&\domfunc{scope}\; \sexpr{\sexpr{S_1\, w_1}\oplus \dots \oplus \sexpr{S_m\,w_m}}
  \defas (\domfunc{scope}\; S_1)
\end{align*}
\end{framed}
\captionsetup{skip=2pt}
\caption{Implementation of $\domfunc{scope}$, which returns the set of
variables in an element of $\SPE$.}
\label{lst:core-semantics-auxfn-scope}
\end{listing*}

\begin{listing*}
\begin{framed}
% \begin{sublisting}{\textwidth}
\begin{align*}
&\domfunc{subsenv}: \dom{Event} \to \dom{Environment} \to \dom{Event} \\[-5pt]
&\domfunc{subsenv}\; e\; \sigma \defas \begin{aligned}[t]
  &\blet\, \set{x, x_1, \dots, x_m} = \mathrm{dom}(\sigma) \\[-5pt]
  &\bin\,\blet\,  e_1 \;\bbe\; \domfunc{subs}\; e\; x_m\; \sigma(x_m) \\[-5pt]
  &\dots \\[-5pt]
  &\bin\,\blet\, e_m \;\bbe\; \domfunc{subs}\; e_{m-1}\; x_1\; \sigma(x_1) \\[-5pt]
  &\bin\, e_m
\end{aligned}
\end{align*}
\end{framed}
\captionsetup{skip=2pt}
\caption{Implementation of $\domfunc{subsenv}$, which rewrites $e$ as an
$\dom{Event}$ $e'$ on one variable $x$.}
\label{lst:core-semantics-auxfn-subsenv}
\end{listing*}

\begin{listing*}
\begin{framed}
\begin{align*}
&\domfunc{negate}: \dom{Event} \to \dom{Event} \\
&\domfunc{negate}\, \sexpr{t\,\token{in}\,v} \defas \bmatch\; (\domfunc{complement}\; v) \\
&\qquad \begin{aligned}[t]
  &\vartriangleright v_1 \amalg \dots \amalg v_m \Rightarrow
    \sexpr{t\,\token{in}\,v_1} \sqcup \dots \sqcup \sexpr{t\,\token{in}\,v_m} \\
  &\vartriangleright v \Rightarrow \sexpr{t\,\token{in}\,v}
\end{aligned} \\
&\domfunc{negate}\, (e_1 \sqcap \dots \sqcap e_m) \defas \sqcup_{i=1}^{m}(\domfunc{negate}\, e_i) \\
&\domfunc{negate}\, (e_1 \sqcup \dots \sqcup e_m) \defas \sqcap_{i=1}^{m}(\domfunc{negate}\, e_i)
\end{align*}
\end{framed}
\captionsetup{skip=2pt}
\caption{Implementation of $\domfunc{negate}$, which applies De Morgan's laws
to an $\dom{Event}$.}
\label{lst:negate}
\end{listing*}

\begin{listing*}
\begin{framed}
\begin{align*}
&\domfunc{dnf}: \dom{Event} \to \dom{Event} \\[-3pt]
&\domfunc{dnf}\; \sexpr{t\; \token{in}\; v}
  \defas \sexpr{t\; \token{in}\; v} \\[-3pt]
&\domfunc{dnf}\; e_1 \sqcup \dots \sqcup e_m
  \defas \sqcup_{i=1}^{m}  (\domfunc{dnf}\, e_i )\\[-3pt]
&\domfunc{dnf}\; e_1 \sqcap \dots \sqcap e_m \defas \begin{aligned}[t]
  &\blet_{1\le{i}\le{m}}\; (e'_{j1}\sqcap\dots\sqcap e'_{j,k_i})
    \;\bbe\; \domfunc{dnf}\, e_i \ \\[-3pt]
  &\bin \bigsqcup_{\substack{1 \le {j_1} \le k_1 \\ \dots \\ 1 \le {j_m} \le k_m}}
    \bigsqcap_{i=1}^{m} e'_{i,j_i}
  \end{aligned}
\end{align*}
\end{framed}
\captionsetup{skip=2pt}
\caption{$\domfunc{dnf}$ converts and $\dom{Event}$ to DNF (Def.~\ref{def:dnf}).}
\label{lst:dnf-appx}
\end{listing*}

\begin{listing*}
\begin{framed}
\begin{align*}
&\domfunc{disjoint?}: \dom{Event} \times \dom{Event} \to \dom{Boolean} \\[-3pt]
&\domfunc{disjoint?}\, \langle e_1, e_2 \rangle \defas \bmatch\; \langle e_1, e_2 \rangle  \\[-3pt]
&\vartriangleright
  \langle
    \sqcap_{i=1}^{m_1} (\scall{Id}{x_{1,i}}\,\token{in}\,v_{1,i}),
    \sqcap_{i=1}^{m_2} (\scall{Id}{x_{2,i}}\,\token{in}\,v_{2,i})
  \rangle \\[-3pt]
&\; \Rightarrow \begin{aligned}[t]
  &\left[\exists_{1 \le i \le 2}.\exists_{1 \le j \le m_i}. v_{ij} = \varnothing) \right]\;
    \vee \left[ \begin{aligned}
  &\blet\;
    \set{\langle n_{1i}, n_{2i} \rangle}_{i=1}^{k}
    \;\bbe\; \set{ \langle i, j \rangle \mid x_{1,i}\,{=}\,x_{2,j} } \\[-3pt]
  &\bin\; (\exists_{1 \le i \le k}.
    (\domfunc{intersection}\, v_{1,n_{1,i}}\, v_{2,n_{2,i}}) = \varnothing)
    \end{aligned} \right]
  \end{aligned}
  \\[-3pt]
&\vartriangleright \belse \Rightarrow \bundef
\end{align*}
\end{framed}
\captionsetup{skip=2pt}
\caption{$\domfunc{disjoint?}$ returns $\ttrue$ if two $\dom{Event}$s are disjoint (Def.~\ref{def:dnf-disjoint}).}
\label{lst:disjoint?}
\end{listing*}

\begin{listing*}
\begin{framed}
\begin{align*}
\valfunc{T}&: \dom{Transform} \to (\dom{Real} \to \dom{Real})\\
\Denot[\valfunc{T}]{\scall{Id}{x}}
  &\defas \lambda r'.\, r' \\
\Denot[\valfunc{T}]{\scall{Reciprocal}{t}}
  &\defas \lambda r'.\, 1/\left({\Denot[\valfunc{T}]{t}(r')}\right) \\
\Denot[\valfunc{T}]{\scall{Abs}{t}}
  &\defas \lambda r'.\, \abs{\Denot[\valfunc{T}]{t}(r')} \\
\Denot[\valfunc{T}]{\scall{Root}{t\; n}}
  &\defas \lambda r'.\, \sqrt[n]{\Denot[\valfunc{T}]{t}(r')} \\
\Denot[\valfunc{T}]{\scall{Exp}{t\; r}}
  &\defas \lambda r'.\, r^{\left({\Denot[\valfunc{T}]{t}(r')} \right)}\;
  && (\mbox{iff } 0 < r) \\
\Denot[\valfunc{T}]{\scall{Log}{t\; r}}
  &\defas \lambda r'.\, \log_{r}\left({\Denot[\valfunc{T}]{t}(r')}\right)\;
  && (\mbox{iff } 0 < r) \\
\Denot[\valfunc{T}]{\scall{Poly}{t\; r_0\; \dots\; r_m}}
  &\defas \lambda r'.\, \textstyle\sum_{i=0}^{m} r_i\left({\Denot[\valfunc{T}]{t}(r')}\right)^i \\
\Denot[\valfunc{T}]{\scall{Piecewise}{\sexpr{t_i\; e_i}_{i=1}^m}}
  &\defas \lambda r'.\, \begin{aligned}[t]
        &\bif\; \left[ (\inj[r']{\dom{Real}}{\dom{Outcome}}) \in \Denotv{V}{\Denotv{E}{e_1}x} \right] \;
          \bthen\; \Denot[\valfunc{T}]{t_1}r' \\
        &\belse\; \bif\; \dots \\
        &\belse\; \bif\; \left[ (\inj[r']{\dom{Real}}{\dom{Outcome}}) \in \Denotv{V}{\Denotv{E}{e_m}x} \right] \;
            \bthen\; \Denot[\valfunc{T}]{t_m}r' \\
        &\belse\; \bundef
  \end{aligned} \span\span \\
&\qquad (\mbox{iff } \begin{aligned}[t]
  (\domfunc{vars}\; t_1)
    = \dots &= (\domfunc{vars}\; t_m) \\
            &= (\domfunc{vars}\; e_1) = \dots = (\domfunc{vars}\; e_m)
            \asdef \set{x}
  \end{aligned} \span\span
\end{align*}
\end{framed}
\captionsetup{skip=2pt}
\caption{Semantics of $\dom{Transform}$.}
\label{lst:transform}
\end{listing*}

\begin{listing*}
\begin{framed}
\begin{align*}
\domfunc{domainof} &: \dom{Transform} \to \dom{Outcomes} \\
\domfunc{domainof}\; \scall{Id}{x}
  &\defas \sintvl{\tfalse\; {-\infty}}{\infty\; \tfalse} \\
\domfunc{domainof}\; \scall{Reciprocal}{t}
  &\defas \sintvl{\tfalse\; 0}{\infty\; \tfalse} \\
\domfunc{domainof}\; \scall{Abs}{t}
  &\defas \sintvl{\tfalse\; {-\infty}}{\infty\; \tfalse} \\
\domfunc{domainof}\; \scall{Root}{t\; n}
  &\defas \sintvl{\tfalse\; 0}{\infty\; \tfalse} \\
\domfunc{domainof}\; \scall{Exp}{t\; r_0}
  &\defas \sintvl{\tfalse\; {-\infty}}{\infty\; \tfalse} \\
\domfunc{domainof}\; \scall{Log}{t\; r_0}
  &\defas \sintvl{\tfalse\; 0}{\infty\; \tfalse} \\
\domfunc{domainof}\; \scall{Poly}{t\; r_0\; \dots\; r_m}
  &\defas \sintvl{\tfalse\; {-\infty}}{\infty\; \tfalse} \\
\domfunc{domainof}\; \scall{Piecewise}{\sexpr{t_i\; e_i}_{i=1}^m} &\defas
  \domfunc{union}\;
    [(\domfunc{intersection}\; (\domfunc{domainof}\; t_i)\; (\Denotv{E}{e}x)]_{i=1}^{m}\\
    &\mbox{where } \set{x} \defas \domfunc{vars}\, t_1
\end{align*}
\end{framed}
\captionsetup{skip=2pt}
\caption{$\domfunc{domainof}$ returns
the $\dom{Outcomes}$ on which a $\dom{Transform}$ is defined.}
\label{lst:domainof}
\end{listing*}

\begin{listing*}
\setlength{\abovedisplayskip}{0pt}%
\setlength{\belowdisplayskip}{0pt}%
\setlength{\abovedisplayshortskip}{0pt}%
\setlength{\belowdisplayshortskip}{0pt}%
\begin{framed}
\begin{align*}
&\domfunc{preimg}\; t\; v\;
  \defas \domfunc{preimage'}\; t\; (\domfunc{intersection}\; (\domfunc{domainof}\; t)\; v) \\
% ID
&\domfunc{preimage'}\; \token{Id}\; v \defas v \\
% EMPTYSET
&\domfunc{preimage'}\; t\; \varnothing \defas \varnothing \\
% UNION
&\domfunc{preimage'}\; t\; (v_1 \amalg \dots \amalg v_m)
  \defas \domfunc{union}\;
    (\domfunc{preimg}\; t\; v_1)\;
    \dots\;
    (\domfunc{preimg}\; t\; v_m) \\
% FINITE SET
&\domfunc{preimage'}\; t\; \settt{r_1 \dots  r_m} \defas
  \domfunc{preimg}\; t'\;
    (\domfunc{union}\; (\domfunc{finv}\; t\; r_1)\; \dots\; (\domfunc{finv}\; t\; r_m)) \\
% INTERVAL
&\domfunc{preimage'}\; t\;
    \sintvl
      {b_{\mathrm{left}}\; r_{\mathrm{left}}}
      {r_{\mathrm{right}}\; b_{\mathrm{right}}}
    \defas \bmatch\; t \\
&\begin{aligned}[t]
% Case 1.
&\vartriangleright \scall{Radical}{t'\; n}
  \gor \scall{Exp}{t'\; r}
  \gor \scall{Log}{t'\; r}
  \Rightarrow
  \begin{aligned}[t]
  &\blet\; \settt{r'_{\mathrm{left}}}\; \bbe\; \domfunc{finv}\; t\; r_{\mathrm{left}} \\
  &\bin\,\blet\; \settt{r'_{\mathrm{right}}}\; \bbe\; \domfunc{finv}\; t\; r_{\mathrm{right}} \\
  &\bin\; \domfunc{preimg}\; t'\;
    \sintvl
      {b_{\mathrm{left}}\; r'_{\mathrm{left}}}
      {r'_{\mathrm{right}}\; b_{\mathrm{right}}}
    \end{aligned} \\
% Case 2.
&\vartriangleright \scall{Abs}{t'} \Rightarrow
  \begin{aligned}[t]
  &\blet\; v'_{\mathrm{pos}} \;\bbe\;
    \sintvl
      {b_{\mathrm{left}}\; r_{\mathrm{left}}}
      {r_{\mathrm{right}}\; b_{\mathrm{right}}} \\
  &\bin\,\blet\; v'_{\mathrm{neg}} \;\bbe\;
    \sintvl
      {b_{\mathrm{right}}\; {-r_{\mathrm{right}}}}
      {{-r_{\mathrm{left}}}\; b_{\mathrm{left}}} \\
  &\bin\; \domfunc{preimg}\; t'\;
    (\domfunc{union}\; v'_{\mathrm{pos}}\; v'_{\mathrm{neg}})
  \end{aligned} \\
% Case 3.
&\vartriangleright \scall{Reciprocal}{t'} \Rightarrow
  \begin{aligned}[t]
  &\blet\; \langle r'_{\mathrm{left}}, r'_{\mathrm{right}} \rangle \;\bbe\;
    \begin{aligned}[t]
    &\bif\; (0 \le r_{\mathrm{left}} < r_{\mathrm{right}}) \\
    &\bthen\; \langle
      \bif\; (0 < r_{\mathrm{left}})\;
      \bthen\; 1/r_{\mathrm{left}}\;
      \belse\; \infty,\\
    &\qquad\quad
      \bif\; (r_{\mathrm{right}} < \infty)\;
      \bthen\; 1/r_{\mathrm{right}}\;
      \belse\; 0
      \rangle \\
    &\belse\; \langle
      \bif\; ({-\infty} < r_{\mathrm{left}})\;
      \bthen\; 1/r_{\mathrm{left}}\;
      \belse\; 0, \\
    &\qquad\quad
      \bif\; (r_{\mathrm{right}} < 0)\;
      \bthen\; 1/r_{\mathrm{right}}\;
      \belse\; {-\infty}
      \rangle
    \end{aligned} \\
  &\bin\; \domfunc{preimg}\; t'\;
    \sintvl
      {b_{\mathrm{right}}\; r'_{\mathrm{right}}}
      {r'_{\mathrm{left}}\; b_{\mathrm{left}} }
    \end{aligned} \\
% Case 4.
&\vartriangleright \scall{Polynomial}{t\; r_0\; \dots\; r_m} \Rightarrow
  \begin{aligned}[t]
  &\blet\; v'_{\mathrm{left}} \;\bbe\; \domfunc{polyLte}\; \neg{b_{\mathrm{left}}}\; r_{\mathrm{left}}\; r_0\; \dots\; r_m \\
  &\bin\,\blet\; v'_{\mathrm{right}} \;\bbe\; \domfunc{polyLte}\; b_{\mathrm{right}}\; r_{\mathrm{right}}\; r_0\; \dots\; r_m \\
  &\bin\;
    \domfunc{preimg}\; t'\;
      (\domfunc{intersection}\; v'_{\mathrm{right}}\;
      (\domfunc{complement}\; v'_{\mathrm{left}}))
  \end{aligned} \\
% Case 5.
&\vartriangleright \scall{Piecewise}{\sexpr{t_i\; e_i}_{i=1}^m} \Rightarrow
  \begin{aligned}[t]
  % &\blet\; \set{x} \;\bbe\; \domfunc{vars}\, t_1 \\
  &\blet_{1 \le i \le m}\; v'_i \;\bbe\;
    \domfunc{preimg}\; t_i\;
      \sintvl
        {b_{\mathrm{left}}\; r_{\mathrm{left}}}
        {b_{\mathrm{right}}\; r_{\mathrm{right}}} \\
  &\bin\,\blet_{1 \le i \le m}\;
    v_i \;\bbe\; \domfunc{intersection}\; v'_i\; (\Denotv{E}{e_i}x), \\
  &\bin\; \domfunc{union}\; v_1\; \dots\; v_m
    \qquad \mbox{where } \set{x} \defas \domfunc{vars}\,t_1
  \end{aligned}
\end{aligned}
\end{align*}
\end{framed}
\captionsetup{skip=2pt}
\caption{$\domfunc{preimg}$ computes the generalized inverse of a many-to-one
$\dom{Transform}$.}
\label{lst:preimage}
\end{listing*}

\begin{listing*}
\begin{framed}
\begin{align*}
\domfunc{finv} &: \dom{Transform} \to \dom{Real} \to \dom{Outcomes} \\
\domfunc{finv}\; \scall{Id}{x}\; r
  &\defas \settt{r} \\
\domfunc{finv}\; \scall{Reciprocal}{t}\; r
  &\defas \bif\; (r = 0)\;
    \bthen\; \settt{{-\infty}\; \infty}
    \belse\; \settt{1/r} \\
\domfunc{finv}\; \scall{Abs}{t}\; r
  &\defas \settt{{-r}\; r} \\
\domfunc{finv}\; \scall{Root}{t\; n}\; r
  &\defas \bif\; (0 \le r)\; \bthen\;\settt{r^n}\; \belse\; \varnothing \\
\domfunc{finv}\; \scall{Exp}{t\; r_0}\; r
  &\defas \bif\; (0 \le r)\; \bthen\; \settt{\log_{r_0}(r)}\; \belse\; \varnothing \\
\domfunc{finv}\; \scall{Log}{t\; r_0}\; r   &\defas \settt{r_0^r} \\
\domfunc{finv}\; \sexpr{\token{Polynomial}\; t\; r_0\; \dots\; r_m}\; r
  &\defas \domfunc{polySolve}\; r\; r_0\; r_1\; \dots\; r_m \\
\domfunc{finv}\; \sexpr{\token{Piecewise}\; \sexpr{t_i\; e_i}_{i=1}^m}
  &\defas \domfunc{union}\;
    [(\domfunc{intersection}\; (\domfunc{finv}\; t_i\; r)\; (\Denotv{E}{e_i} x))]_{i=0}^{m},\\
    &\mbox{where } \set{x} \defas \domfunc{vars}\, t_1
\end{align*}
\end{framed}
\captionsetup{skip=0pt}
\caption{$\domfunc{finv}$ computes the generalized inverse of a many-to-one
transform at a single $\dom{Real}$.}
\label{lst:finv}
\end{listing*}

\begin{listing*}
\begin{framed}
\begin{align*}
&\domfunc{polyLim}: \dom{Real}^+ \to \dom{Real}^2 \\
&\domfunc{polyLim}\; r_0 \defas \langle r_0, r_0 \rangle \\
&\domfunc{polyLim}\; r_0\; r_1\; \dots\; r_m \defas \\
&\qquad \begin{aligned}
  &\blet\; n \;\bbe\; \max \set{j \mid r_{j} > 0 } \\
  &\bin\;  \bif\; (\domfunc{even}\; n)\; \begin{aligned}[t]
        \bthen\; &(\bif\; (r_n > 0)\;
          \bthen\; \langle \infty, \infty \rangle\;
          \belse\; \langle -\infty, -\infty \rangle)\\
        \belse\; &(\bif\; (r_n > 0)\;
          \bthen\; \langle -\infty, \infty \rangle\;
          \belse\; \langle \infty, -\infty \rangle)
        \end{aligned}
\end{aligned}
\end{align*}
\end{framed}
\captionsetup{skip=2pt}
\caption{$\domfunc{polyLim}$
computes the limits of a polynomial limits at the infinities.}
\label{lst:polyLim}
\end{listing*}

\begin{listing*}
\begin{framed}
\begin{align*}
&\domfunc{polySolve}: \dom{Real} \to \dom{Real}^+ \to \dom{Set} \\
&\domfunc{polySolve}: r\; r_0\; \dots\; r_m \defas \bmatch\; r \\
&\qquad \begin{aligned}[t]
  &\vartriangleright (\infty \gor {-\infty})
    &&\Rightarrow \begin{aligned}[t]
      &\blet\;
        \langle r_{\mathrm{neg}}, r_{\mathrm{pos}}\rangle\; &&\bbe\; \domfunc{polyLim}\; r_0\; \dots\; r_m \\
      &\bin\,\blet\;
        \domfunc{f}\; &&\bbe\; \lambda r'.\,
          \bif\; (r = \infty)\; \bthen\; (r' = \infty)\; \belse\; (r' = {-\infty})\; \\
      &\bin\,\blet\;
        v_{\mathrm{neg}}\; &&\bbe\;
          \bif\; (\domfunc{f}\; r_{\mathrm{neg}})\;
          \bthen\; \settt{-\infty}\;
          \belse\; \varnothing \\
      &\bin\,\blet\;
        v_{\mathrm{pos}}\; &&\bbe\;
          \bif\; (\domfunc{f}\; r_{\mathrm{pos}})\;
          \bthen\; \settt{\infty}\;
          \belse\; \varnothing \\
      &\bin\; &&\domfunc{union}\; v_{\mathrm{pos}}\; v_{\mathrm{neg}}
    \end{aligned} \\
  &\vartriangleright \belse
    &&\Rightarrow
      (\domfunc{roots}\; (r_0 - r)\; r_1\; \dots\; r_m)
  \end{aligned}
\end{align*}
\end{framed}
\captionsetup{skip=2pt}
\caption{$\domfunc{polySolve}$
computes the set of values at which a polynomial is equal to a specific value $r$.}
\label{lst:polySolve}
\end{listing*}

\begin{listing*}
\begin{framed}
\begin{align*}
&\domfunc{polyLte} : \dom{Boolean} \to \dom{Real} \to \dom{Real}^+ \to \dom{Outcomes} \\
&\domfunc{polyLte}\; b\; r\; r_0\; \dots\; r_m \defas \bmatch\; r \\
&\qquad \begin{aligned}[t]
  % Case 1.
  &\vartriangleright -\infty
    &&\Rightarrow \bif\; b\;
      \bthen\; \varnothing\;
      \belse\; (\domfunc{polySolve}\; r\; r_0\; \dots\; r_m) \\
  % Case 2.
  &\vartriangleright \infty
      &&\Rightarrow \begin{aligned}[t]
        \bif\; \neg{b}\;
          &\bthen\; \sintvl{\ttrue\; {-\infty}}{\infty\; \ttrue}\\
          &\belse\; \begin{aligned}[t]
            &\blet\; \langle r_{\mathrm{left}}, r_{\mathrm{right}}\rangle \;\bbe\;
              \domfunc{polyLim}\; r_0\; \dots\; r_m \\
            &\bin\,\blet\;
              \langle b_{\mathrm{left}}, b_{\mathrm{right}} \rangle \;\bbe\;
              \langle r_{\mathrm{left}} = \infty, r_{\mathrm{right}} = \infty \rangle \\
            &\bin\; \sintvl
              {b_{\mathrm{left}}\; {-\infty}}
              {\infty b_{\mathrm{right}}}
          \end{aligned}\\
      \end{aligned} \\
  % Case 3.
  &\vartriangleright \belse &&\Rightarrow \begin{aligned}[t]
    &\blet\;
      [r_{\mathrm{s}, i}]_{i=1}^{k} \;\bbe\;
      \domfunc{roots}\; (r_0 - r)\; r_1\; \dots\; r_m \\
    &\bin\,\blet\;
      [\langle r'_{\mathrm{left},i}, r'_{\mathrm{right},i}\rangle]_{i=0}^{k} \;\bbe\;
        [
          \langle {-\infty}, r_{\mathrm{s},0} \rangle,
          \langle r_{\mathrm{s}, 1}, r_{\mathrm{s}, 2} \rangle,
          \dots,
          \langle r_{\mathrm{s}, k-1}, r_{\mathrm{s}, k} \rangle,
          \langle r_{\mathrm{s}, k}, \infty \rangle
        ]\\
    &\bin\,\blet\;
      \domfunc{f}_{\rm mid} \;\bbe\; \lambda rr'.\, \begin{aligned}[t]
              &\bif\;         &&(r = -\infty)\; &&\bthen\; r'  \\
              &\belse \bif\;  &&(r' = \infty)\; &&\bthen\; r \\
              &\belse\; && (r + r') / 2
            \end{aligned} \\
          & t' \;\bbe\; \scall{Poly}{\scall{Id}{\token{x}}\; (r_0 -r)\; r_1\; \dots\; r_m}\\
    &\bin\; \domfunc{union}\; \left[\begin{aligned}[m]
          &\bif\;
            \Denot[\valfunc{T}]{t'}
              (\domfunc{f}_{\rm mid}\; r'_{\mathrm{left},i}\; r'_{\mathrm{right},i})
            &&\bthen\; \sintvl{b\; r'_{\mathrm{left},i}}{r'_{\mathrm{right},i}\; b} \\
            \span &&\belse\; \varnothing \\
        \end{aligned}\right]_{i=0}^k
  \end{aligned}
\end{aligned}
\end{align*}
\end{framed}
\captionsetup{skip=2pt}
\caption{$\domfunc{polyLte}$
computes the set of values at which a polynomial is less than a given value $r$.}
\label{lst:polyLte}
\end{listing*}

\end{document}